\documentclass[print]{nuthesis}

\usepackage{graphicx,wrapfig}
\usepackage{epstopdf}
\usepackage{caption}
\usepackage{listings}
\usepackage{epsfig}
\usepackage{enumerate}	
\usepackage{amsfonts}		
\usepackage{amssymb}
\usepackage{amsmath}
\usepackage{amsthm}
\usepackage{color}
\usepackage{hyperref}
\hypersetup{linktoc=all}
\usepackage[all]{xy}
\usepackage{color}
\newtheorem{theorem}{Theorem}
\newtheorem{lemma}{Lemma}
\newtheorem{proposition}{Proposition}
\newtheorem{corollary}{Corollary}

\newenvironment{definition}[1][Definition.]{\begin{trivlist}
\item[\hskip \labelsep {\bfseries #1}]}{\end{trivlist}}

\newenvironment{example}[1][Example.]{\begin{trivlist}
\item[\hskip \labelsep {\bfseries #1}]}{\end{trivlist}}

\newenvironment{question}[1][Question:]{\begin{trivlist}
\item[\hskip \labelsep {\bfseries #1}]}{\end{trivlist}}

\newenvironment{remark}[1][Remark:]{\begin{trivlist}
\item[\hskip \labelsep {\bfseries #1}]}{\end{trivlist}}

\newcommand{\F}{\mathbb F}
\newcommand{\R}{\mathbb R}
\newcommand{\B}{\mathcal B}
\newcommand{\U}{\mathcal U}
\newcommand{\PP}{\mathcal P}
\newcommand{\E}{\mathcal E}

\newcommand{\C}{\mathcal C}
\newcommand{\D}{\mathcal D}
\newcommand{\p}{\textbf p}
\newcommand{\q}{\textbf q}
\newcommand{\od}{\stackrel{\text{def}}{=}}

\newcommand{\supp}{\operatorname{supp}}
\def\A{\mathcal{A}}

\begin{document}
\frontmatter

\title{The Neural Ring: using algebraic geometry to analyze neural codes }
\author{Nora Esther Youngs}
\adviser{Professor Carina Curto}
\adviserAbstract{Carina Curto}
\major{Mathematics}
\degreemonth{August}
\degreeyear{2014}
\maketitle
\begin{abstract}
Neurons in the brain represent external stimuli via neural codes.   These codes often arise from stimulus-response maps, associating
to each neuron a convex receptive field.  An important problem confronted by the brain is to infer properties of a represented stimulus space
without knowledge of the receptive fields, using only the intrinsic structure of the neural code.  How does the brain do this?  To address this question,
it is important to determine what stimulus space features can - in principle - be extracted from neural codes.  This motivates us to define the neural ring
and a related neural ideal, algebraic objects that encode the full combinatorial data of a neural code.  We find that these objects can be expressed in a "canonical form'' that directly translates to a minimal description of the receptive field structure intrinsic to the neural code. We consider the algebraic properties of homomorphisms between neural rings, which naturally relate to maps between neural codes.  We show that maps between two neural codes are in bijection with ring homomorphisms between the respective neural rings, and define the notion of neural ring homomorphism, a special restricted class of ring homomorphisms which preserve neuron structure.  We also find connections to Stanley-Reisner rings, and use ideas similar to those in the theory of monomial ideals to obtain an algorithm for computing the canonical form associated to any neural code, providing the groundwork for inferring stimulus space features from neural activity alone.
\end{abstract}


%
%
%
%
%
%
%



\tableofcontents
\mainmatter


\chapter{Introduction}

Building accurate representations of the world is one of the basic functions of the brain.  It is well-known that when a stimulus is paired with pleasure
or pain, an animal quickly learns the association. Animals also learn, however, the (neutral) relationships between stimuli of the same type.
For example, a bar held at a 45-degree angle appears more similar to one held at 50 degrees than to a perfectly vertical one.  
Upon hearing a triple of distinct pure tones, one seems to fall ``in between'' the other two.  An explored environment is perceived not as a collection of
disjoint physical locations, but as a spatial map.  In summary, we do not experience the world as a stream of unrelated stimuli; rather, our brains
organize different types of stimuli into highly structured {\em stimulus spaces}.  

The relationship between neural activity and stimulus space \textit{structure} has, nonetheless, received remarkably little attention.  
In the field of neural coding, 
much has been learned about the coding properties of individual neurons
by investigating stimulus-response functions, such as place fields \cite{OKeefeDostrovsky,PathIntegration}, orientation tuning curves \cite{WatkinsBerkley74, Ben-Yishai1995}, and other examples of ``receptive fields'' obtained by measuring neural activity in response to experimentally-controlled stimuli.  Moreover, numerous studies have shown that neural 
activity, together with knowledge of the appropriate stimulus-response functions, can be used to accurately estimate a newly presented
stimulus \cite{Brown98,Deneve99,Ma}.  This paradigm is being actively extended and revised to include information
present in populations of neurons, spurring debates on the role of correlations in neural coding 
\cite{NirenbergLatham03,Averbeck2006,SchneidmanBialek2006}.
In each case, however, the underlying structure of the stimulus space is
assumed to be known, and is not treated as itself emerging from the activity of neurons.  This approach is particularly problematic when
one considers that the brain does {\em not} have access to stimulus-response functions, and
must represent the world {\em without} the aid of dictionaries that lend meaning to neural activity \cite{gap}.  In coding theory parlance, the brain 
does not have access to the encoding map, and must therefore represent stimulus spaces via the intrinsic structure of the neural code.

How does the brain do this?  In order to eventually answer this question, we must first tackle a simpler one:  
\medskip

\noindent\textbf{Question:} \textit{What can be 
inferred about the underlying stimulus space from neural activity alone?} I.e., what stimulus space features are encoded in the {\em intrinsic structure} of the neural code, and can thus be extracted without knowing the individual stimulus-response functions?
\medskip

\noindent Recently we have shown that, in the case of hippocampal place cell codes, certain topological features of 
the animal's environment can be inferred from the neural code alone, without knowing the place fields \cite{gap}.  
As will be explained in the next section, this information can be extracted from a simplicial
complex associated to the neural code.  \textit{What other stimulus space features can be inferred from the neural code?}
For this, we turn to algebraic geometry.
Algebraic geometry provides a useful framework for inferring geometric and topological characteristics
of spaces by associating rings of functions to these spaces.  All relevant features of the underlying space are encoded in the intrinsic structure of the ring, where coordinate functions become indeterminates, and the space itself is defined in terms of ideals in the ring.  Inferring features of a space from properties of functions -- without specified domains -- is similar to the task confronted by the brain, so it is natural to expect that this framework may shed light on our question.

Here we introduce the \textit{neural ring}, an algebro-geometric object that can be associated to any combinatorial neural code.
Much like the simplicial complex of a code, the neural ring encodes information about the underlying stimulus space in a way that discards specific knowledge of receptive field maps, and thus gets closer to the essence of how the brain might represent stimulus spaces.  Unlike the simplicial complex, the neural ring retains the full combinatorial data of a neural code, packaging this data in a more computationally tractable manner.  We find that this object, together with a closely related
{\em neural ideal}, can be used to algorithmically extract a compact, minimal description of the \textit{receptive field structure} dictated by the code.  This enables us to more directly
tie combinatorial properties of neural codes to features of the underlying stimulus space, a critical step towards answering our motivating question.

Although the use of an algebraic construction such as the neural ring is quite novel in the context of neuroscience, the neural code (as we define it) is at its core a combinatorial
object, and there is a rich tradition of associating algebraic objects to combinatorial ones \cite{MillerSturmfels}.  The most well-known example is perhaps the Stanley-Reisner ring \cite{StanleyBook}, which turns out to be closely related to the neural ring.  Within mathematical biology, associating polynomial ideals to combinatorial data has also been fruitful.  Recent examples include inferring wiring diagrams in gene-regulatory networks \cite{Laubenbacher2007, Alan2012} and applications to chemical reaction networks \cite{ShiuSturmfels2010}.   Our work also has parallels to the study of design ideals in algebraic statistics \cite{alg-stats-book}. 

From a data analysis perspective, it is useful to consider the codes as related to one another, not merely as isolated objects. A single code can give rise to a host of relatives through natural operations such as adding codewords or dropping neurons. Understanding how these relationships translate to structural information will allow us to extract information from multiple codes simultaneously.  From an algebraic perspective, relationships between neural rings stem from ring homomorphisms. We characterize the set of homomorphisms between neural rings, relating each to a code map via the pullback.

The organization of this dissertation is as follows. In Chapter 2, we discuss in greater depth the type of neural codes which motivate this work, and explore some previous results which give partial answers to our open questions.  In Chapter 3, we introduce our main object of study, the neural ring, an algebraic object which stores information from neural codes. We investigate some of its properties, and in Chapter 4 we determine a preferred ``canonical" presentation that allow us to extract stimulus space features.  In Chapter 5, we give two variations on an algorithm for obtaining this  canonical form.  In Chapter 6, we consider the primary decomposition of the neural ideal and its interpretations.  Finally, in Chapter 7, we consider the maps which relate one neural ring to another, and the relationship between these maps and the functions  which relate one neural code to another. Much of Chapters 1-6 appears in our recent paper \cite{NeuralRing2013}; however, substantial changes have been made to the algorithm for obtaining the canonical form.  Finally, in the appendix, we show all possible examples on 3 neurons, to show the wide variety of possibilities on even a small set of neurons, and present Matlab code for some of our algorithms.

\chapter{Neural Codes} \label{chap:neural-codes}

In this chapter, we introduce the basic objects of study: neural codes, receptive field codes, and convex receptive field codes.  We then discuss various ways in which
the structure of a convex receptive field code can constrain the underlying stimulus space.  
These constraints emerge most obviously from the simplicial complex of a neural code, but (as will be made clear) there are also constraints that arise from aspects
of a neural code's structure that go well beyond what is captured by the simplicial complex of the code.  First, we give a few basic definitions.

\begin{definition}
Given a set of neurons labelled $\{1, \dots, n\} \od [n]$, we define a {\em neural
  code} $\C \subset \{0,1\}^n$ as a set of binary patterns of neural activity.  
An element of a neural code is called a {\em codeword}, $c = (c_1,\ldots,c_n) \in \C$, and corresponds to a subset of
neurons
$$\supp(c) \od \{i \in [n] \mid c_i = 1\} \subset [n].$$ 
Similarly, the entire code $\C$ can be identified with a set of subsets of neurons,
$$\supp \C \od \{\supp(c) \mid c \in \C\} \subset 2^{[n]},$$
where $2^{[n]}$ denotes the set of all subsets of $[n]$.
Because we discard the
details of the precise timing and/or rate of neural activity, what we
mean by {neural code} is often referred to in the neural coding
literature as a {\em combinatorial code} \cite{BialekBerry,Bialek2008}.
\end{definition}

For simplicity's sake, we will henceforth dispense with vector notation and reduce to the simpler binary notation; e.g., the codeword $(1,0,1)$ will be written $101$.

\begin{example} Consider the code $\C = \{000,100,010, 110, 001\}$.  Here, we have $\supp(\C) = \{\emptyset, \{1\}, \{2\}, \{1,2\}, \{3\} \}$.  As a neural code, we interpret this as a set of activity patterns for 3 neurons, where we have observed the following:
\begin{itemize}[-]
\item At some point, no neurons were firing $(\emptyset$).
\item Each neuron fired alone at some point ($\{1\}, \{2\}$, and $\{3\}$).
\item At some point, neurons 1 and 2 fired together, while 3 was silent $(\{1,2\}$).
\end{itemize} 
\end{example}

\begin{definition} A set of subsets $\Delta \subset 2^{[n]}$
is an (abstract) {\em simplicial complex} if $\sigma \in \Delta$ and $\tau \subset \sigma$ implies $\tau \in \Delta$.  We will say that a neural code $\C$ is a simplicial complex if $\supp \C$ is a simplicial complex.  In cases where the code is \textit{not} a simplicial complex, we can complete the code
to a simplicial complex by simply adding in missing subsets of codewords.  This allows us to define the {\em simplicial complex of the code} as
$$\Delta(\C) \od \{\sigma \subset [n] \mid \sigma \subseteq \supp(c) \text{ for some } c \in \C\}.$$
Alternatively, $\Delta(\C)$ can be defined as the smallest simplicial complex that contains $\supp \C$.
\end{definition}

\begin{example} Our code in the previous example, $\C =\{000,100, 010, 110, 001\}$ is a simplicial complex.     

However, the code $\D = \{000,100,010, 110, 011\}$ is not, because the set $\{2,3\}$ is in $\supp(\D)$, but its subset $\{3\}$ is not.  We can take the simplicial complex of the code $\D$ by adding in the necessary subsets, to obtain $\Delta(\D) = \{000,100,010,110,  011, 001\}$.
\end{example}

\section{Receptive field codes (RF codes)}
Neurons in many brain areas have activity patterns that can be characterized by receptive fields.\footnote{In the vision literature, 
the term ``receptive field'' is reserved for subsets of the visual field; we use the term in a more general sense, applicable to any modality.}
Abstractly, a \textit{receptive field} is a map $f_i:X \rightarrow \R_{\geq 0}$ from a space of stimuli, $X$, to the average firing rate of a single neuron, $i$, in response to each stimulus.   Receptive fields are computed by correlating neural responses to independently measured external stimuli.  We follow a common abuse of language, where both the map and its support (i.e., the subset $U_i \subset X$ where $f_i$ takes on
positive values) are referred to as ``receptive fields.''  \textit{Convex} receptive fields are convex subsets of the stimulus space, for $X \subset \R^d$. 

\begin{definition} A subset $B\subset \R^n$ is \textit{convex} if, given any pair of points $x,y\in B$, the point $z=tx+(1-t)y$ is contained in $B$ for any $t\in [0,1].$
\end{definition}  

The paradigmatic examples are orientation-selective neurons in visual cortex \cite{WatkinsBerkley74, Ben-Yishai1995} and 
hippocampal place cells \cite{OKeefeDostrovsky,PathIntegration}.  

Orientation-selective neurons have \textit{tuning curves} that reflect a neuron's preference for a particular angle . When an animal is presented with stimuli in the form of bars at a certain angle, these neurons have a marked preference for one particular angle.  The neuron fires at higher and higher rates as the angle of the bars approaches the preferred angle, producing a tuning curve (see Figure 1A).

 Place cells are neurons that have \textit{place fields}; i.e., each neuron has a preferred (convex) region of the animal's physical environment where it has a high firing rate (see Figure 1B).  When the animal occupies that particular region, the neuron fires markedly more frequently; when the animal is in any other area.  the neuron's firing rate is comparatively very low. 
 
  Both tuning curves and place fields are examples of receptive fields. In both cases, the receptive field for each  neuron is convex (an interval of angles about the preferred angle, or a place field) but not all receptive fields are convex.  Grid cells are another type of neuron with a receptive field, very like place cells, but their receptive field consists of a set of distinct regions which form a triangular grid, and thus in this case the receptive field is not convex, or even connected \cite{2005GridCells}. 

\begin{figure}[h]\label{fig1}
\begin{center}
\includegraphics[width=5in]{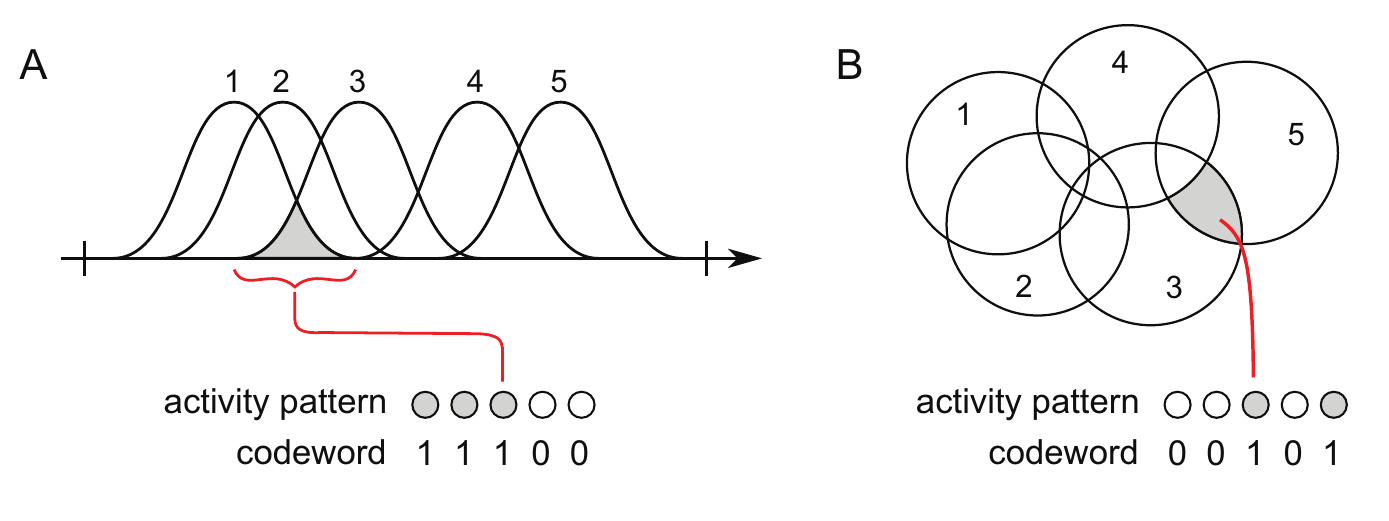}
\end{center}
\vspace{-.25in}
\caption{\small {Receptive field overlaps determine codewords in 1D and 2D RF codes.}  (A) Neurons in a 1D RF code have receptive fields that overlap on a line segment (or circle, in the case of orientation-tuning).  Each stimulus on the line
corresponds to a binary codeword.  Gaussians depict graded firing rates for neural responses; this additional information is discarded by the RF code.  (B) Neurons in a 2D RF code, such as a place field code, have receptive fields that partition a two-dimensional stimulus space into non-overlapping regions, as illustrated by the shaded area.  All stimuli within one of these regions will activate the same set of neurons, and hence have the same corresponding codeword. }
\end{figure}

A receptive field code (RF code) is a neural code that corresponds to the brain's representation of the 
stimulus space covered by the receptive fields.  When a stimulus lies in the intersection of several receptive
fields, the corresponding neurons may co-fire while the rest
remain silent.  The active subset $\sigma$ of neurons can be identified with a binary
codeword $c \in \{0,1\}^n$ via $\sigma = \supp(c)$.  Unless otherwise noted, a stimulus space $X$ need only be a topological space.
However, we usually have in mind $X \subset \R^d$, and this becomes important when we consider \textit{convex} RF codes.

\begin{definition}
Let $X$  be a stimulus space (e.g., $X \subset \R^d$), and let $\U = \{U_1,\ldots,U_n\}$ be a collection of open sets, 
with each $U_i \subset X$ the receptive field of the $i$-th neuron in a population of $n$ neurons.  
The {\em receptive field code (RF code)} $\C(\U) \subset \{0,1\}^n$ is the set of all binary codewords corresponding to stimuli in $X$: 
$$\C(\U) \od \{ c \in \{0,1\}^n \mid (\bigcap_{i \in \supp(c)} U_i) \setminus (\bigcup_{j \notin \supp(c)} U_j) \neq \emptyset \}.$$
If $X \subset \R^d$ and each of the $U_i$s is also a \textit{ convex} subset of $X$, then we say that $\C(\U)$ is a {\em convex} RF code.
\end{definition}

Our convention is that $\bigcap_{i \in \emptyset} U_i = X$ and $\bigcup_{i \in \emptyset} U_i = \emptyset$.
This means that if $\bigcup_{i=1}^n U_i \subsetneq X$, then $\C(\U)$ includes the all-zeros codeword corresponding to an ``outside'' point not covered by the receptive fields; 
on the other hand, if
$\bigcap_{i=1}^n U_i \neq \emptyset$, then $\C(\U)$ includes the all-ones codeword. 
Figure 1 shows examples of convex receptive fields covering one- and two-dimensional stimulus spaces, 
and examples of codewords corresponding to regions defined by the receptive fields.

Returning to our discussion in the Introduction, we have the following question:
If we can assume $\C = \C(\U)$ is a RF code, then \textit{ what can be learned about the underlying stimulus space $X$ from knowledge only of $\C$, and not of $\U$?} 
The answer to this question will depend critically on whether or not we can assume that the RF code is convex.  In particular, if we don't make any assumptions about the receptive fields beyond openness, then any code can be realized as a RF code in any dimension.  Thus, without some kind of assumption like convexity, the answer to the above question is ``nothing useful."

\begin{lemma}\label{lemma:RFform}
Let $\C \subset \{0,1\}^n$ be a neural code.  Then, for any $d \geq 1$, there exists a stimulus space $X \subset \R^d$ and a collection of open sets $\U = \{U_1,\ldots,U_n\}$ (not necessarily convex), with $U_i \subset X$ for each $i \in [n]$, such that $\C = \C(\U)$.
\end{lemma}

\begin{proof}
Let $\C\subset\{0,1\}^n$ be any neural code, and order the elements of $\C$ as $\{c^1,\ldots,c^m\}$, where $m = |\C|$. For each $c\in \C$, choose a distinct point $x_c\in \R^d$ and an open neighborhood $N_c$ of $x_c$ such that no two neighborhoods intersect.
Define $U_j \od \bigcup_{j\in \supp(c^k)} N_{c^k}$, let $\U=\{U_1,\ldots,U_n\}$, and $X=\bigcup_{i=1}^m N_{c^i}$.  Observe that if the all-zeros codeword is in $\C$, then
$N_{\textbf 0} = X \setminus \bigcup_{i=1}^{n} U_i$ corresponds to the ``outside point'' not covered by any of the $U_i$s.  By construction, $\C = \C(\U).$
\end{proof}

Although any neural code $\C \subseteq \{0,1\}^n$ can be realized as a RF code, it is \textit{not} true that any code can be realized as a \textit{convex} RF code.  Counterexamples
can be found in codes having as few as three neurons.

\begin{lemma}\label{lemma:convex-counterexample}
The neural code $\C = \{0,1\}^3 \setminus \{111, 001\}$ on three neurons cannot be realized as a convex RF code.
\end{lemma}

\begin{wrapfigure}{r}{.33\linewidth}
\vspace{-.3in}
 \includegraphics[width=2in]{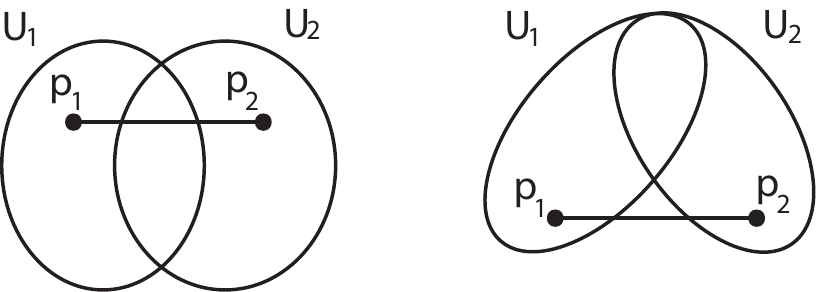} 
 \caption{\small Two cases in the proof of Lemma~\ref{lemma:convex-counterexample}. }
 \vspace{0in}
\end{wrapfigure}

\text{             }

\vspace{-.3in}
\begin{proof}

Assume the converse, and let $\U = \{U_1,U_2, U_3\}$ be a set of convex open sets in  $\R^d$ such that $\C = \C(\U)$.  The code
necessitates that $U_1 \cap U_2 \neq \emptyset$ (since $110 \in \C$), $(U_1 \cap U_3) \setminus U_2 \neq \emptyset$ (since $101 \in \C$), and 
$(U_2 \cap U_3) \setminus U_1 \neq \emptyset$ (since $011 \in \C$).  Let $p_1 \in (U_1 \cap U_3) \setminus U_2$ and $p_2 \in (U_2 \cap U_3) \setminus U_1.$
Since $p_1,p_2 \in U_3$ and $U_3$ is convex, the line segment $\ell = (1-t)p_1 + t p_2$ for $t \in [0,1]$ must also be contained in $U_3$.

Every point in $\ell$ is in $U_3$.  However, as there are no points in $U_1\cap U_2\cap U_3$ or in $U_3\backslash ( U_1\cup U_2$), then all points on $\ell$ are in either $U_1$ or $U_2$ but no point may be in both. Thus  $U_1\cap \ell$, $U_2\cap \ell$ are disjoint nonempty open sets which cover $\ell$, thus they disconnect $\ell$.  But as $\ell$ is a line segment, it should be connected in the subspace topology.  This is a contradiction, so no such realization can exist.

Figure 2 illustrates the impossibility of such a realization. There are really only two possibilities.  Case 1: $\ell$ passes through $U_1 \cap U_2$ (see Figure 2, left).  This implies $U_1 \cap U_2 \cap U_3 \neq \emptyset$, and hence $111 \in \C$, a contradiction.  Case 2: $\ell$ does not intersect $U_1 \cap U_2$.  Since $U_1, U_2$ are open sets, this implies $\ell$ passes outside of $U_1 \cup U_2$ (see Figure 2, right), and hence $001 \in \C$, a contradiction.  
\end{proof}

\section{Stimulus space constraints arising from convex RF codes}


It is clear from Lemma~\ref{lemma:RFform} that there is essentially no constraint on the stimulus space for realizing a code as a RF code.  However, if we demand that $\C$
is a \textit{convex} RF code, then the overlap structure of the $U_i$s sharply constrains the geometric and topological properties of the underlying stimulus space $X$.   
To see how this works, we first consider the simplicial complex of a neural code, $\Delta(\C)$.
Classical results in convex geometry and topology provide constraints on the underlying stimulus space $X$ for convex RF codes, based on the structure of $\Delta(\C)$.  
We will discuss these next.  We then turn to the question of constraints that arise from combinatorial properties of a neural code $\C$ that are \textit{not} captured
by $\Delta(\C)$.

\subsection{Helly's theorem and the Nerve theorem}\label{sec:helly-nerve}

Here we briefly review two classical  and well-known theorems in convex geometry and topology, Helly's theorem and the Nerve theorem, as they apply to convex RF codes.
Both theorems can be used to relate the structure of the simplicial complex of a code, $\Delta(\C)$, to topological features of the underlying stimulus space $X$.  

Suppose $\U = \{U_1,\ldots,U_n\}$ is a finite collection of convex open subsets of $\R^d$, with dimension $d<n$.  
We can associate to $\U$ a simplicial complex $N(\U)$ called the \textit{ nerve} of $\U$.  A subset $\{i_1,..,i_k\}\subset [n]$ belongs to $N(\U)$ if and only if  the appropriate intersection $\bigcap_{\ell=1}^kU_{i_\ell} $  is  nonempty. 
If we think of the $U_i$s as receptive fields, then $N(\U) = \Delta(\C(\U))$.  In other words, the nerve of the cover corresponds to the simplicial complex of the associated (convex) RF code.
\medskip

\noindent \textbf{ Helly's theorem.}
{\em Consider $k$ convex subsets, $U_1,\ldots,U_k \subset \R^d,$ for $d<k$.
If the intersection of every $d+1$ of these sets is nonempty, then the full intersection $\bigcap_{i=1}^k U_i$ is also nonempty.}
\medskip

\noindent A nice exposition of this theorem and its consequences can be found in \cite{helly-review}.
One straightforward consequence is that the nerve $N(\U)$ is completely determined by its $d$-skeleton, and corresponds to the largest simplicial complex with that $d$-skeleton.  For example, if $d = 1$, then $N(\U)$ is a clique complex (fully determined by its underlying graph).  Since $N(\U) = \Delta(\C(\U))$, Helly's theorem imposes constraints on the minimal dimension of the stimulus space $X$ when $\C = \C(\U)$ is assumed to be a convex RF code.     For example, if we have some collection of codewords, and there are three neurons (or more) where each pair of neurons is seen to fire together but there is no word where all fire together, then the minimal dimension of the stimulus space where this code could be realized as a convex receptive field code is 2.  

\medskip

\noindent \textbf{Nerve theorem.} 
{\em The homotopy type of $X(\U) \od \bigcup_{i=1}^n U_i$ is equal to the homotopy type of the nerve of the cover, $N(\U)$.  In particular, $X(\U)$ and $N(\U)$ have exactly the same homology groups.}
\medskip

\noindent The Nerve theorem  is an easy consequence of \cite[Corollary 4G.3]{Hatcher}.  This is a powerful theorem relating the simplicial complex of a RF code, $\Delta(\C(\U)) = N(\U)$, to topological features of the underlying space, such as homology groups and other homotopy invariants.  
In \cite{gap}, this theorem is used in the context of two-dimensional RF codes (specifically, place field codes for place cells in rat hippocampus) to show that topological features of the animal's environment could be inferred from the observed neural code, without knowing the place fields.
Note, however, that the similarities between $X(\U)$ and $N(\U)$ only go so far.  In particular, $X(\U)$ and $N(\U)$ typically have very different {dimension}.  It is also important to keep in mind that the Nerve theorem concerns the topology of  $X(\U) = \bigcup_{i=1}^n U_i$.  In our setup, if the stimulus space $X$ is larger, so that  $\bigcup_{i=1}^n U_i \subsetneq X$, then the Nerve theorem tells us only about the homotopy type of $X(\U)$, not of $X$.  Since the $U_i$ are open sets, however, conclusions about the dimension of $X$ can still be inferred.

In addition to Helly's theorem and the Nerve theorem, there is a great deal known about $\Delta(\C(\U))=N(\U)$ for collections of convex sets in $\R^d$.  In particular, the $f$-vectors of such simplicial complexes have been completely characterized by G. Kalai in \cite{kalai1,kalai2}.

\subsection{Beyond the simplicial complex of the neural code}\label{sec:beyond}

We have just seen how the simplicial complex of a neural code, $\Delta(\C)$, yields constraints on the stimulus space $X$ if we assume $\C$ can be realized as a convex RF code.  Consider the example described in Lemma~\ref{lemma:convex-counterexample}.   Nothing from Helly's theorem expressly said that $\C$ could not be realized in $\R^2$; indeed, $\Delta(\C)$ can be realized easily.  Yet we have proven it is impossible to realize the code $\C$ in any dimension at all.  This implies that other kinds of constraints on $X$ may emerge from the combinatorial structure of a neural code, even if there is no obstruction stemming from $\Delta(\C)$. 

\begin{wrapfigure}{r}{.6\linewidth}
\centering
   \includegraphics[width=2.75in]{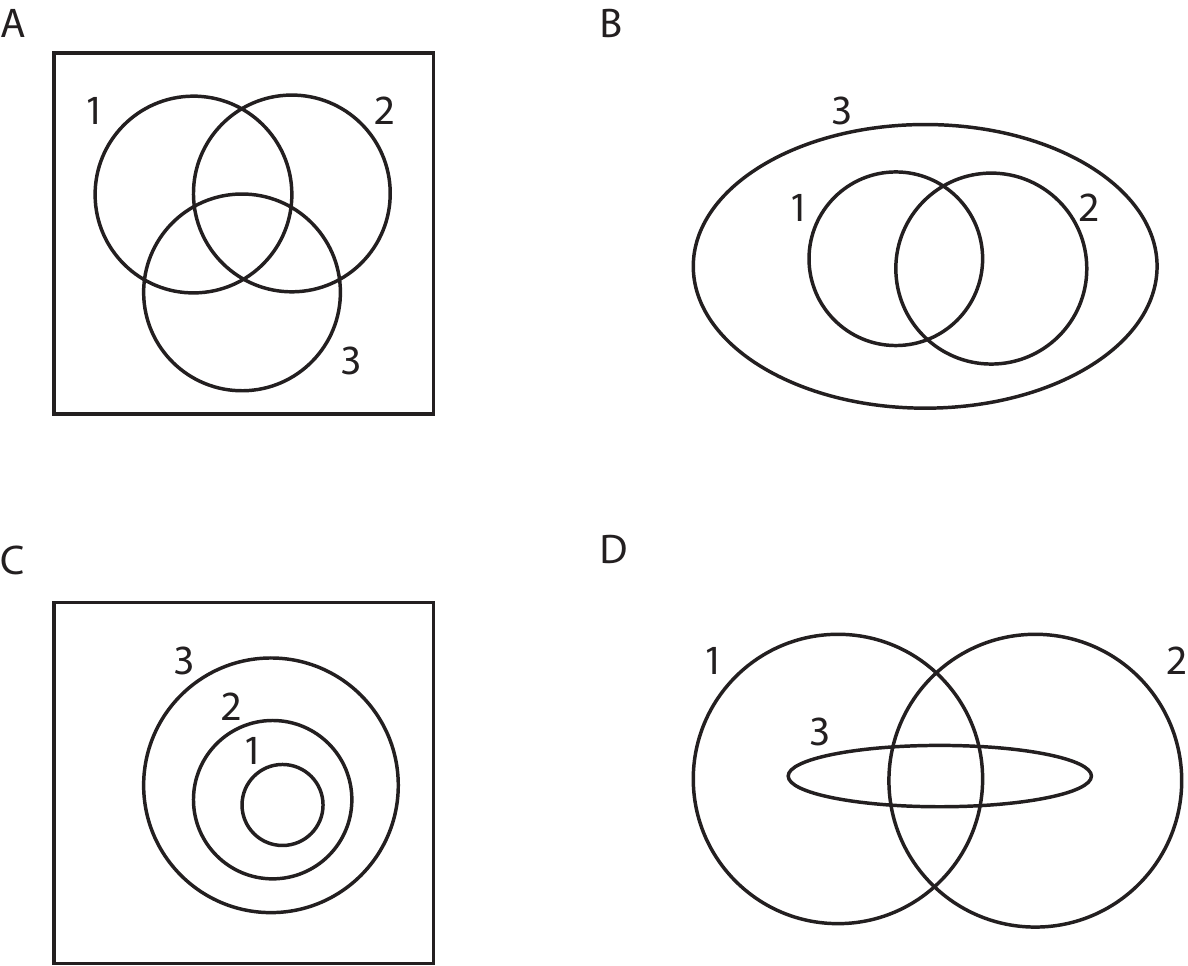} 
   \caption{\small Four arrangements of three convex receptive fields, $\U = \{U_1, U_2, U_3\}$, each having $\Delta(\C(\U)) = 2^{[3]}$.  
   Square boxes denote the stimulus space $X$ in cases where
   $U_1\cup U_2 \cup U_3 \subsetneq X$.
   (A) $\C(\U) = 2^{[3]}$, including the all-zeros codeword $000$.  (B) $\C(\U) = \{111, 101, 011, 001\}$, with $X = U_3$.
   (C) $\C(\U) = \{111, 011, 001, 000\}$.  (D) $\C(\U) = \{111, 101, 011, 110, 100, 010\},$ and $X = U_1 \cup U_2$.
The minimal embedding dimension for the codes in panels A and D is $d=2$, while for panels B and C it is $d=1$.}
\end{wrapfigure}


In Figure 3 we show four possible arrangements of three convex receptive fields in the plane.  Each convex RF code has the same corresponding simplicial complex $\Delta(\C) = 2^{[3]}$, since $111 \in \C$ for each code.  Nevertheless, the arrangements clearly have different combinatorial properties.  In Figure 3C, for instance, we have $U_1 \subset U_2 \subset U_3$, while Figure 3A has no special containment relationships among the receptive fields.  This ``receptive field structure'' (RF structure) of the code has impliciations for the underlying stimulus space.

Let $d$ be the minimal integer for which the code can be realized as a convex RF code in $\R^d$; we will refer to this as the {\em minimal embedding dimension} of $\C$.  Note that the codes in Figure 3A,D have $d=2$, whereas the codes in Figure 3B,C have $d=1$.  The simplicial complex, $\Delta(\C)$, is thus not sufficient
to determine the minimal embedding dimension of a convex RF code, but this information \textit{ is} somehow present in the RF structure of the code.   Similarly, in 
Lemma~\ref{lemma:convex-counterexample} we saw that $\Delta(\C)$ does not provide
sufficient information to determine whether or not $\C$ can be realized as a convex RF code; after working out
the RF structure, however, it was easy to see that the given code was not realizable.

\subsection{The receptive field structure (RF structure) of a neural code}

As we have just seen, the intrinsic structure of a neural code contains information about
the underlying stimulus space that cannot be inferred from the simplicial complex of the code alone. 
This information is, however, present in what we have loosely referred to as the ``RF structure'' of the code.
We now explain more carefully what we mean by this term.

Given a set of receptive fields $\U = \{U_1,\ldots,U_n\}$ in a stimulus space $X$, there are certain containment relations between intersections and unions of the $U_i$s that
are ``obvious,'' and carry no information about the particular arrangement in question.  These relationships are merely a result of unavoidable set relationships.  For example, $U_1 \cap U_2 \subseteq U_2 \cup U_3 \cup U_4$ is always guaranteed to be true, because it follows from $U_2 \subseteq U_2.$  On the other hand, a relationship such as $U_3 \subseteq U_1 \cup U_2$ (as in Figure 3D) is \textit{not} always present, and thus reflects something about the structure of a particular receptive field arrangement.

Let $\C \subset \{0,1\}^n$ be a neural code, and let $\U = \{U_1,\ldots,U_n\}$ be any arrangement of receptive fields in a stimulus space $X$ such that $\C = \C(\U)$ (this is guaranteed to exist by Lemma~\ref{lemma:RFform}).  The \textit{ RF structure} of $\C$ refers to the set of relations among the $U_i$s that are not ``obvious,'' 
and have the form:
$$\bigcap_{i \in \sigma} U_i \subseteq \bigcup_{j \in \tau} U_j, \;\; \text{ for } \;\; \sigma \cap \tau = \emptyset.$$
In particular, this includes any empty intersections $\bigcap_{i \in \sigma} U_i = \emptyset$ (here $\tau = \emptyset$).
In the examples in Figure 3, the panel A code has no unusual RF structure relations and is as general as possible; while panel B has $U_1 \subset U_3$ and $U_2 \subset U_3$; panel C has $U_1 \subset U_2 \subset U_3$; and panel D has $U_3 \subset U_1 \cup U_2$.

Our central goal is to develop a method to algorithmically extract a minimal description of the RF structure directly from a neural code $\C$, without first
realizing it as $\C(\U)$ for some arrangement of receptive fields.  We view this as a first step towards inferring stimulus space features that cannot be obtained from the simplicial complex $\Delta(\C)$.  To do this we turn to an algebro-geometric framework, that of neural rings and ideals.  These objects are defined in Section~\ref{chap:neural-ring} so as to capture the full combinatorial data of a neural code, but in a way that allows us to naturally and algorithmically infer a compact description of the desired RF structure, as shown in Chapter~\ref{chap:RF-structure}.

\chapter{Neural Rings and Neural Ideals} \label{chap:neural-ring}

In this chapter we define the neural ring $R_\C$ and a closely-related neural ideal, $J_\C$.  First, we briefly review some basic algebraic geometry background 
needed throughout the following sections.

\section{Basic algebraic geometry background}

The following definitions are standard (see, for example, \cite{cox-little-oshea}). 

\begin{definition}[Rings and ideals.]
Let $R$ be a commutative ring. 
A subset $I\subseteq R$ is an \textit{ideal} of $R$ if it has the following properties:
\begin{enumerate}
\item[(i)] $I$ is a subgroup of $R$ under addition.
\item[(ii)] If $a\in I$, then $ra\in I$ for all $r \in R$.
\end{enumerate}
An ideal $I$ is said to be \textit{generated by} a set $A$, and we write $I=\langle A\rangle$, if $$I=\{  r_1a_1+\cdots + r_n a_n \, |\, a_i\in A, r_i\in R, \text{ and } n\in \mathbb{N}\}.$$ 
In other words, $I$ is the set of all finite combinations of elements of $A$ with coefficients in $R$.

An ideal $I \subset R$ is \textit{proper} if $I\subsetneq R$.  An ideal $I \subset R$ is \textit{prime} if it is proper and has the following property: if $rs\in I$ for some $r,s\in R$, then 
$r\in I$ or $s\in I$.  An ideal $m \subset R$ is \textit{maximal} if it is proper and if for any ideal $I$ such that $m\subseteq I\subseteq R$, either $I=m$ or $I=R$.
 An ideal $I \subset R$ is \textit{radical} if $r^n\in I$ implies $r\in I$, for any $r \in R$ and $n \in  \mathbb{N}$.  An ideal $I \subset R$ is \textit{primary} if $rs\in I$ implies $r\in I$ or $s^n \in I$ for some $n \in  \mathbb{N}$.  A \textit{primary decomposition} of an ideal $I$ expresses $I$ as an intersection of finitely many primary ideals.
\end{definition}

\begin{definition}[Ideals and varieties.]
Let $k$ be a field, $n$ the number of neurons, and  $k[x_1,\ldots,x_n]$ a polynomial ring with one indeterminate $x_i$ for each neuron.
We will consider $k^n$ to be the neural activity space, where each point $v = (v_1,\ldots,v_n) \in k^n$ is a vector tracking the state $v_i$ of each neuron.   Note that any
polynomial $f \in k[x_1,\ldots,x_n]$ can be evaluated at a point $v \in k^n$ by setting $x_i = v_i$ each time $x_i$ appears in $f$.  We will denote this value $f(v)$.

Let $J \subset k[x_1,\ldots,x_n]$ be an ideal, and define the variety
$$V(J) \od \{v \in k^n \mid f(v) = 0 \text{ for all } f \in J\}.$$
Similarly, given a subset $S \subset k^n$, we can define the ideal of functions that vanish on this subset as
$$I(S) \od \{f \in k[x_1,\ldots.,x_n] \mid f(v) = 0 \text{ for all } v \in S\}.$$
The ideal-variety correspondence \cite{cox-little-oshea} gives us the usual order-reversing relationships: $I \subseteq J \Rightarrow V(J) \subseteq V(I)$, and
 $S \subseteq T \Rightarrow I(T) \subseteq I(S)$.  Furthermore, $V(I(V)) = V$ for any variety $V$, but
 it is not always true that $I(V(J)) = J$ for an ideal $J$ (see Section~\ref{sec:lemma-proofs}).
We will regard neurons as having only two states, ``on'' or ``off,'' and thus choose $k = \F_2 = \{0,1\}$.  
\end{definition}

\section{Definition of the neural ring}

Let $\C \subset \{0,1\}^n = \F_2^n$ be a neural code, and define the ideal $I_\C$ of $\F_2[x_1,\ldots,x_n]$ corresponding to the set of polynomials that vanish on all codewords in $\C$:
$$I_\C \od I(\C) = \{f \in \F_2[x_1,\ldots,x_n] \mid f(c) = 0 \text{ for all } c \in \C\}.$$
By design, $V(I_\C) \supseteq \C$; we will show that in fact $V(I_\C) = \C$ and hence $I(V(I_\C)) = I_\C$.  To see this, define an ideal $m_v = \langle x_1-v_1,...,x_n-v_n\rangle$ for every $v\in \F_2^n$; note that $V(m_v) = \{v\}$.  Then, for a code $\C\subset\F_2^n$, define the ideal $J = \bigcap_{v\in \C} m_v$.  As this intersection is finite, $\C = V(J)$, and thus we have $$V(I_\C) = V(I(\C)) =V(I(V(J_\C))) = V(J_\C) = \C.$$
Note that the ideal generated by the  {\em Boolean relations},
$$\B \od \langle x_1^2-x_1,\ldots,x_n^2-x_n \rangle,$$
is automatically contained in $I_\C$, irrespective of $\C$.

The {\em neural ring} $R_\C$ corresponding to the code $\C$ is the quotient ring
$$R_\C \od \F_2[x_1,\ldots,x_n]/I_\C,$$
together with the set of indeterminates $x_1,\ldots,x_n$.  We say that two neural rings are {\em equivalent} if there
is a bijection between the sets of indeterminates that yields a ring homomorphism.
\medskip

\noindent \textbf{Remark.} Due to the Boolean relations, any element $y \in R_\C$ satisfies $y^2 = y$ (cross-terms vanish because $2=0$ in $\F_2$), so the neural ring is a \textit{Boolean ring} isomorphic to $\F_2^{|\C|}$.  It is important to keep in mind, however, that $R_\C$ comes equipped with a privileged set of functions, $x_1,\ldots,x_n$; this allows the ring to keep track of considerably more structure than just the size of the neural code.  The importance of using this presentation will be clear as we begin to extract receptive field information.

\section{The spectrum of the neural ring}\label{sec:spec}
We can think of $R_\C$ as the ring of functions of the form $f:\C \rightarrow \F_2$ on the neural code, where each function assigns a $0$ or $1$
to each codeword $c \in \C$ by evaluating $f \in \F_2[x_1,\ldots,x_n]/I_\C$ through the substitutions $x_i = c_i$ for $i = 1,\ldots,n$.   To see this, note that two polynomials are in the same equivalence class in $R_\C$ if and only if they evaluate the same on every $c\in \C$. That is,  $f=g$ in $R_\C \Leftrightarrow f-g\in I_\C \Leftrightarrow f(c) - g(c) = 0$ for all $c\in \C,$ i.e., $f(c) = g(c)$ for all $c\in \C$. Quotienting the original
polynomial ring by $I_\C$ ensures that there is only one zero function in $R_\C$.

The spectrum of the neural ring, $\mathrm{Spec}(R_\C)$, consists 
of all prime ideals in $R_\C$.  We will see shortly that the elements of $\mathrm{Spec}(R_\C)$ are in one-to-one correspondence with the
elements of the neural code $\C$.  Indeed, our definition of $R_\C$ was designed for this to be true.

For any point $v \in \{0,1\}^n$ of the neural activity space, let
$$m_v \od I(v) =  \{f \in \F_2[x_1,\ldots,x_n] \mid f(v) = 0 \}$$
be the maximal ideal of $\F_2[x_1,\ldots,x_n]$ consisting of all functions that vanish on $v$.  
We can also write $m_v = \langle x_1-v_1,\ldots,x_n-v_n\rangle$ (see Lemma~\ref{lemma:mv} in Section~\ref{sec:lemma-proofs}).
Using this, we can characterize the spectrum of the neural ring.

\begin{lemma}  \label{lemma:spec}
$\mathrm{Spec}(R_\C)  = \{ \bar{m}_v \mid v \in \C\},$ where $\bar{m}_v$ is the quotient of $m_v$ in $R_\C$.
\end{lemma}

\noindent The proof is given in Section~\ref{sec:lemma-proofs}.  Note that because $R_\C$ is a Boolean ring, the maximal ideal spectrum and the prime ideal spectrum coincide.

\section{The neural ideal \& an explicit set of relations for the neural ring}

The definition of the neural ring is rather impractical, as it does not give us explicit relations for generating $I_\C$ and $R_\C$.
Here we define another ideal, $J_\C$, via
an explicit set of generating relations.  Although $J_\C$ is closely related to $I_\C$, it
turns out that $J_\C$ is a more convenient object to study, which is why we will use the term \textit{neural ideal} to refer to $J_\C$ rather than $I_\C$.

For any $v \in \{0,1\}^n$, consider the function $\rho_v \in \F_2[x_1,\ldots,x_n]$ defined as
$$\rho_v \od \prod_{i=1}^n(1-v_i-x_i) = \prod_{\{i\,|\,v_i=1\}}x_i\prod_{\{j\,|\,v_j=0\}}(1-x_j) = 
\prod_{i \in \supp(v)}x_i\prod_{j \notin \supp(v)}(1-x_j).$$
Note that $\rho_v(x)$ can be thought of as a characteristic function for $v$, since it satisfies $\rho_v(v) = 1$ and $\rho_v(x) = 0$ for any other $x \in \F_2^n$.  Now consider the ideal $J_\C \subseteq  \F_2[x_1,\ldots,x_n]$ generated by all functions $\rho_v$, for $v \notin \C$:
$$J_\C \od \langle \{ \rho_v \mid v \notin \C\}\rangle.$$
We call $J_\C$ the {\em neural ideal} corresponding to the neural code $\C$.  If $\C = 2^{[n]}$ is the complete code, we simply set $J_\C = 0$, the zero ideal.  $J_\C$ is related to $I_\C$ as follows, giving us explicit relations for the neural ring.

\begin{lemma} \label{lemma:explicit-relations}
Let $\C \subset \{0,1\}^n$ be a neural code.  Then,
$$I_\C = J_\C + \B = \big \langle \{ \rho_v \mid v \notin \C\}, \{x_i(1-x_i) \mid i \in [n]\} \big \rangle,$$
where $\B = \langle \{x_i(1-x_i) \mid i \in [n]\} \rangle$ is the ideal generated by the Boolean relations, and $J_\C$
is the neural ideal.
\end{lemma}

\noindent The proof is given in Section~\ref{sec:lemma-proofs}.

\section{Proof of Lemmas~\ref{lemma:spec} and~\ref{lemma:explicit-relations}}\label{sec:lemma-proofs}

To prove Lemmas~\ref{lemma:spec} and~\ref{lemma:explicit-relations}, we need a version of the Nullstellensatz for finite fields.
The original ``Hilbert's Nullstellensatz'' applies when $k$ is an algebraically closed field. It states that if $f \in k[x_1,\ldots,x_n]$ vanishes on $V(J)$, then $f \in \sqrt{J}$.  In other words,
$$I(V(J)) = \sqrt{J}.$$
Because we have chosen $k = \F_2 = \{0,1\}$, we have to be a little careful about the usual ideal-variety correspondence, as there are some subtleties introduced in the case of finite fields.   In particular, $J = \sqrt{J}$ in $\F_2[x_1,\ldots,x_n]$ does not imply $I(V(J)) = J$.

The following lemma and theorem are well-known.
Let $\F_q$ be a finite field of size $q$, and $\F_q[x_1,\ldots,x_n]$ the $n$-variate polynomial ring
over $\F_q$.

\begin{lemma}\label{lemma:radical}
For any ideal $J \subseteq \F_q[x_1,\ldots,x_n]$, the ideal $J+\langle x_1^q-x_1, \ldots, x_n^q-x_n \rangle$ is a radical ideal.
\end{lemma}

\begin{theorem}[Strong Nullstellensatz in Finite Fields]  For an arbitrary finite field $\F_q$, let $J \subseteq \F_q[x_1,\ldots,x_n]$ be an ideal.  Then,
$$I(V(J)) = J + \langle x_1^q-x_1, \ldots, x_n^q-x_n \rangle.$$
\end{theorem}


\subsection*{Proof of Lemma~\ref{lemma:spec}}

 We begin by describing the maximal ideals of $\F_2[x_1,\ldots,x_n]$.
Recall that
$$m_v \od I(v) =  \{f \in \F_2[x_1,\ldots,x_n] \mid f(v) = 0 \}$$
is the maximal ideal of $\F_2[x_1,\ldots,x_n]$ consisting of all functions that vanish on $v \in \F_2^n$.  We will use the notation $\bar{m}_v$
to denote the quotient of $m_v$ in $R_\C$, in cases where $m_v \supset I_\C$.

\begin{lemma}  \label{lemma:mv}
$m_v = \langle x_1-v_1, \ldots, x_n-v_n \rangle \subset \F_2[x_1,\ldots,x_n]$, and is a radical ideal.
\end{lemma}
\begin{proof}  Denote $A_v = \langle x_1-v_1, \ldots, x_n-v_n \rangle$, and observe that $V(A_v) = \{v\}$.
It follows that $I(V(A_v)) = I(v) =  m_v$.  On the other hand, using the Strong Nullstellensatz in Finite Fields we have
$$I(V(A_v)) = A_v + \langle x_1^2-x_1,\ldots,x_n^2-x_n \rangle = A_v,$$
where the last equality is obtained by observing that, since $v_i \in \{0,1\}$ and $x_i^2-x_i = x_i(1-x_i)$, each generator of $ \langle x_1^2-x_1,\ldots,x_n^2-x_n \rangle$ is already contained in $A_v$.
We conclude that $A_v = m_v$, and the ideal is radical by Lemma~\ref{lemma:radical}.
\end{proof}

\noindent In the proof of Lemma~\ref{lemma:spec}, we make use of the following correspondence: for any quotient ring $R/I$, the maximal ideals of $R/I$ are exactly the quotients $\bar m = m/I$, where $m$ is a maximal ideal of $R$ that contains $I$  \cite{atiyah-macdonald}.

\begin{proof}[Proof of Lemma~\ref{lemma:spec}]
First, recall that because $R_\C$ is a Boolean ring, $\mathrm{Spec}(R_\C) = \mathrm{maxSpec}(R_\C)$, the set of all maximal ideals of $R_\C$.  We
also know that any maximal ideal of $\F_2[x_1,\ldots,x_n]$ which contains $I_\C$ is of the form $m_v$ for $v\in \F_2^n$.  To see this, we only need show that for maximal ideal $m\supseteq I_\C$, we have $V(m) \neq \emptyset$ (since if $v\in V(m)$, then $m\subseteq m_v$, and as $m$ is maximal, $m=m_v$).  To show this, suppose that $V(m) = \emptyset$. Using the Strong Nullstellensatz, since $m\supseteq I_\C\supseteq \langle x_1^2-x_1,...,x_n^2-x_n\rangle$, we have $$m=m+  \langle x_1^2-x_1,...,x_n^2-x_n\rangle = I(V(m)) = I(\emptyset) = \F_2[x_1,...,x_n] $$ which is a contradiction. 

  By the correspondence stated above, to show that
$\mathrm{maxSpec}(R_\C)  = \{ \bar{m}_v \mid v \in \C\}$ it suffices to show $m_v \supset I_\C$ if and only if $v\in \C$. To see this, note that for each $v\in \C$, $I_\C\subseteq m_v$ because, by definition, all elements of $I_\C$ are functions that vanish on each $v \in \C$. On the other hand, if $v\notin \C$ then $m_v \not\supseteq I_\C$; in particular,
the characteristic function $\rho_v \in I_\C$ for $v \notin \C$, but  $\rho_v \notin m_v$ because $\rho_v(v) = 1$.
Hence, the maximal ideals of $R_\C$ are exactly those of the form $\bar m_v$ for $v\in \C$.
\end{proof} 

We have thus verified that the points in $\mathrm{Spec}(R_\C)$ correspond to codewords in $\C$.  This was expected given our original
definition of the neural ring, and  suggests that the relations on $\F_2[x_1,\ldots,x_n]$ imposed by $I_\C$ are simply relations ensuring that 
$V(\bar m_v) = \emptyset$ for all $v \notin \C$.  

\subsubsection*{Proof of Lemma~\ref{lemma:explicit-relations}}
Here we find explicit relations for $I_\C$ in the case of an arbitrary neural code. 
Recall that
$$\rho_v = \prod_{i=1}^n((x_i-v_i)-1) = \prod_{\{i\,|\,v_i=1\}}x_i\prod_{\{j\,|\,v_j=0\}}(1-x_j),$$
and that $\rho_v(x)$ can be thought of as a characteristic function for $v$, since it satisfies $\rho_v(v) = 1$ and $\rho_v(x) = 0$ for any other
$x \in \F_2^n$.  This immediately implies that
$$V(J_\C) = V(\langle \{ \rho_v \mid v \notin \C\}\rangle) = \C.$$
We can now prove Lemma~\ref{lemma:explicit-relations}.
\begin{proof}[Proof of Lemma~\ref{lemma:explicit-relations}]
Observe that $I_\C = I(\C) = I(V(J_\C))$, since $V(J_\C) = \C$.  On the other hand, the
Strong Nullstellensatz  in Finite Fields implies $I(V(J_\C)) = J_\C + \langle x_1^2-x_1,\ldots,x_n^2-x_n \rangle = J_\C + \B.$
\end{proof}

\chapter{How to infer RF structure using the neural ideal}\label{chap:RF-structure}

We begin by presenting an alternative set of relations that can be used to define the neural ring.  These relations
enable us to easily interpret elements of $I_\C$ as receptive field relationships, clarifying the connection between the neural ring and ideal and the RF structure of the code. 

We next introduce pseudo-monomials and pseudo-monomial ideals, and use these notions to obtain a minimal description of the neural ideal, which we call
the ``canonical form.''  Theorem~\ref{thm:canonical-form} enables us to use the canonical form of $J_\C$ in order to ``read off'' a minimal description of the RF structure of the code.
Finally, we present an algorithm that inputs a neural code $\C$ and outputs the canonical form $CF(J_\C)$, and illustrate its use in a detailed example.

\section{An alternative set of relations for the neural ring}
Let $\C \subset \{0,1\}^n$ be a neural code, and recall by Lemma~\ref{lemma:RFform} that $\C$ can always be realized as a RF code $\C = \C(\U)$, provided we don't require the $U_i$s to be convex.
Let $X$ be a stimulus space and $\U = \{U_i\}_{i=1}^n$ a collection of open sets in $X$, and consider the RF code $\C(\U)$.
The neural ring corresponding to this code is $R_{\C(\U)}.$

Observe that the functions $f \in R_{\C(\U)}$ can be evaluated at any point $p \in X$ by assigning
$$x_i(p) = \left\{\begin{array}{cc} 1 & \text{if}\; p \in U_i \\ 0 & \text{if}\; p \notin U_i \end{array}\right.$$
each time $x_i$ appears in the polynomial $f$.  The vector $(x_1(p),\ldots,x_n(p)) \in \{0,1\}^n$ represents the neural response to the stimulus $p$.
Note that if $p \notin \bigcup_{i=1}^n U_i$, then $(x_1(p),\ldots,x_n(p)) = (0,\ldots,0)$ is the all-zeros codeword.
For any $\sigma \subset [n]$, define
$$U_\sigma \od \bigcap_{i \in \sigma} U_i, \;\text{  and  }\; x_\sigma \od \prod_{i \in \sigma} x_i.$$
Our convention is that $x_\emptyset = 1$ and $U_{\emptyset} = X$, even in cases where $X \supsetneq \bigcup_{i = 1}^n U_i$.
Note that for any $p \in X$,
$$x_\sigma(p) = \left\{\begin{array}{cc} 1 & \text{if}\; p \in U_\sigma \\ 0 & \text{if}\; p \notin U_\sigma. \end{array}\right.$$

The relations in $I_{\C(\U)}$ encode the combinatorial data of $\U$.  For example, if $U_\sigma = \emptyset$ then we cannot have $x_\sigma = 1$ at any point of the stimulus space $X$, and
must therefore impose the relation $x_\sigma$
to ``knock off'' those points.
On the other hand, if $U_\sigma \subset U_i \cup U_j,$ then $x_\sigma = 1$ implies either $x_i = 1$ or $x_j = 1$, something that is guaranteed by imposing the relation
$x_\sigma(1-x_i)(1-x_j)$.  These observations lead us to an alternative ideal, $I_\U \subset \F_2[x_1,\ldots,x_n]$, defined directly
from the arrangement of receptive fields $\U = \{U_1,\ldots,U_n\}$:
$$I_\U \od \big \langle \big\{ x_\sigma \prod_{i \in \tau} (1-x_i) \mid U_\sigma \subseteq
\bigcup_{i \in \tau}U_i \big\} \big \rangle .$$
Note that if $\tau = \emptyset$, we only get a relation for $U_\sigma = \emptyset$, and this is $x_\sigma$.
If $\sigma = \emptyset$, then $U_\sigma = X$, and we only get relations of this type if $X$ is contained in the union of the $U_i$s.  This is equivalent to the requirement that there is no ``outside point'' corresponding to the all-zeros codeword.

Perhaps unsurprisingly, it turns out that $I_\U$ and $I_{\C(\U)}$ exactly coincide, so $I_\U$ provides an alternative set of relations that can be used to define $R_{\C(\U)}$.

\begin{theorem} \label{thm:ideal-equivalence}
$I_\U = I_{\C(\U)}.$
\end{theorem}

Recall that for a given set of receptive fields $\U = \{U_1,\ldots,U_n\}$ in some stimulus space $X$, the ideal $I_\U \subset \F_2[x_1,\ldots,x_n]$ was 
defined as:
$$I_\U \od \big \langle \{ x_\sigma \prod_{i \in \tau} (1-x_i) \mid U_\sigma \subseteq
\bigcup_{i \in \tau}U_i \} \big \rangle .$$
The Boolean relations are present in $I_\U$ irrespective of $\U$, as it is always true that $U_i \subseteq U_i$ and this yields the relation $x_i(1-x_i)$ for each $i$.  By analogy with our definition of $J_\C$, it makes sense to define an ideal $J_\U$ which is obtained by stripping away the Boolean relations.  This will then be used in the proof of Theorem~\ref{thm:ideal-equivalence}.

Note that if $\sigma \cap \tau \neq \emptyset$, then for any $i \in \sigma \cap \tau$ we have $U_\sigma \subseteq U_i \subseteq \bigcup_{j \in \tau} U_i$, and the corresponding relation is a multiple of the Boolean relation $x_i(1-x_i)$.
We can thus restrict attention to relations in $I_\U$ that have $\sigma \cap \tau = \emptyset,$ so long as we include separately the Boolean relations.
These observations are summarized by the following lemma.

\begin{lemma}\label{lemma:JU}
$I_\U = J_\U + \langle x_1^2-x_1,\ldots,x_n^2-x_n \rangle,$
where
$$J_\U \od  \big \langle \{ x_\sigma \prod_{i \in \tau} (1-x_i) \mid \sigma \cap \tau = \emptyset \;\;\mathrm{ and }\;\; U_\sigma \subseteq \bigcup_{i \in \tau}U_i \} \big \rangle.$$
\end{lemma}

\begin{proof}[Proof of Theorem~\ref{thm:ideal-equivalence}]
We will show that $J_\U=J_{\C(\U)}$ (and thus that  $I_\U = I_{\C(\U)}$) by showing that each ideal contains the
generators of the other. 

First, we show that all generating relations of $J_{\C(\U)}$ are contained in $J_\U$.  Recall that the generators of $J_{\C(\U)}$ are of the form
$$\rho_v = \prod_{i \in \supp(v)}x_i\prod_{j \notin \supp(v)}(1-x_j) \;\; \text{for} \;\; v\notin \C(\U).$$
If $\rho_v$ is a generator of $J_{\C(\U)}$, then $v \notin \C(\U)$ and this implies (by the definition of $\C(\U)$)
that $U_{\supp(v)} \subseteq \bigcup_{j \notin \supp(v)} U_j$.  Taking $\sigma = \supp(v)$ and $\tau = [n] \setminus \supp(v)$, we have $U_{\sigma} \subseteq \bigcup_{j \in \tau} U_j$ with $\sigma \cap \tau = \emptyset$.  This in turn tells us (by the definition of $J_\U$) that $x_\sigma \prod_{j \in \tau} (1-x_j)$ is a generator of $J_\U$.  Since 
$\rho_v = x_\sigma \prod_{j \in \tau} (1-x_j)$ for our choice of $\sigma$ and $\tau$, we conclude that $\rho_v \in J_\U$.  Hence, $J_{\C(\U)} \subseteq J_\U$.

Next, we show that all generating relations of $J_\U$ are contained in $J_{\C(\U)}$.
If $J_\U$ has generator  $x_\sigma\prod_{i\in \tau}(1-x_i)$, then $U_\sigma\subseteq \bigcup_{i\in \tau} U_i$ and $\sigma \cap \tau = \emptyset$.  This in turn implies that $\bigcap_{i \in \sigma} U_i \setminus \bigcup_{j \in \tau} U_j = \emptyset$, and thus (by the definition of $\C(\U)$) we have $v\notin \C(\U)$ for any $v$ such that $\supp(v)\supseteq \sigma$ and $\supp(v)\cap \tau=\emptyset$.
It follows that $J_{\C(\U)}$ contains the relation $x_{\supp(v)}\prod_{j\notin \supp(v)}(1-x_j)$ for any such $v$.  This includes all relations of the form
  $x_\sigma \prod_{j\in \tau}(1-x_j) \prod_{k\notin \sigma\cup \tau}P_k$, where $P_k\in \{x_k,1-x_k\}$.  Taking 
  $f =  x_\sigma \prod_{j\in \tau}(1-x_j)$ in Lemma~\ref{lemma:induction-arg} (below), we can conclude that $J_{\C(\U)}$ contains $x_\sigma \prod_{j\in \tau}(1-x_j)$.  Hence, $J_\U \subseteq J_{\C(\U)}$.
  \end{proof}

\begin{lemma} \label{lemma:induction-arg}
For any $f\in k[x_1,\ldots,x_n]$ and $\tau\subseteq [n]$, the ideal $\langle
\big\{f\prod_{i\in \tau} P_i \mid P_i \in \{x_i, 1-x_i\}\big\}
\rangle = \langle f\rangle.$
\end{lemma}
\begin{proof}

First, denote
$I_f(\tau)\stackrel{\text{def}}{=}\langle\big\{f\prod_{i\in \tau}
P_i \mid P_i \in \{x_i, 1-x_i\}\big\} \rangle $.  We wish to prove
that $I_f(\tau)=\langle f \rangle$, for any $\tau \subseteq [n]$. Clearly, $I_f(\tau)\subseteq \langle f\rangle$,
since every generator of $I_f(\tau)$ is a multiple of $f$. We will prove
$I_f(\tau)\supseteq \langle f\rangle$ by induction on $|\tau|$.

If $|\tau|=0$, then $\tau=\emptyset$ and $I_f(\tau)=\langle f\rangle$. If $|\tau|=1$, so that $\tau=\{i\}$ for some
$i\in [n]$, then $I_f(\tau)=\langle f(1-x_i), fx_i\rangle$.  Note that
 $f(1-x_i)+fx_i=f$, so $f\in I_f(\tau)$, and thus
 $I_f(\tau)\supseteq \langle f\rangle$.

 Now, assume that for some $\ell \geq 1$ we have $I_f(\sigma)\supseteq \langle f\rangle$ for any $\sigma\subseteq [n]$
 with $|\sigma| \leq \ell$.  If $\ell \geq n$, we are done, so we need only show that if $\ell < n$, 
 then $I_f(\tau)\supseteq \langle f\rangle$ for any $\tau$ of size $\ell + 1$.
 Consider $\tau\subseteq [n]$ with $|\tau|=\ell+1$, and let $j \in \tau$ be
 any element. Define $\tau'=\tau\backslash
 \{j\}$, and note that $|\tau'|=\ell$. By our inductive assumption,
 $I_f(\tau')\supseteq\langle f\rangle$.  We will show that 
 $I_f(\tau)\supseteq I_f(\tau')$, and hence $I_f(\tau)\supseteq
 \langle f\rangle$.

 Let $g = f\prod_{i\in \tau'}P_i$ be any generator of
 $I_f(\tau')$ and observe that both $f(1-x_j)\prod_{i\in \tau'}P_i$ and
 $f x_j\prod_{i\in \tau'}P_i$ are both generators of $I_f(\tau)$.  It follows that 
 their sum, $g$, is also in $I_f(\tau)$, and hence $g \in I_f(\tau)$ 
 for any generator $g$ of  $I_f(\tau')$. We conclude that $I_f(\tau)\supseteq I_f(\tau')$, as desired.
\end{proof}

\section{Interpreting neural ring relations as receptive field relationships}\label{sec:types}

Theorem~\ref{thm:ideal-equivalence} suggests that we can interpret elements of $I_\C$ in terms of relationships between receptive fields.  

\begin{lemma}\label{lemma:corresp}
Let $\C \subset \{0,1\}^n$ be a neural code, and let $\U = \{U_1,\ldots,U_n\}$ be any collection of open sets (not necessarily convex) in a stimulus space $X$ 
such that $\C = \C(\U)$. 
Then, for any pair of subsets $\sigma,\tau \subset [n]$,
$$x_\sigma \prod_{i \in \tau} (1-x_i) \in I_\C \; \Leftrightarrow \; U_\sigma \subseteq \bigcup_{i \in \tau} U_i.$$
\end{lemma}

\begin{proof}  ($\Leftarrow$) This is a direct consequence of Theorem~\ref{thm:ideal-equivalence}.
($\Rightarrow$)  We distinguish two cases, based on whether or not $\sigma$ and $\tau$ intersect.
If $x_\sigma \prod_{i \in \tau} (1-x_i) \in I_\C$ and $\sigma \cap \tau \neq \emptyset$, then $x_\sigma \prod_{i \in \tau} (1-x_i) \in \B$, where 
$\B = \langle \{x_i(1-x_i) \mid i \in [n]\} \rangle$ is the ideal generated by the Boolean relations.  Consequently, the relation does not give us any information about the code,
and $U_\sigma \subseteq \bigcup_{i \in \tau} U_i$ follows trivially from the observation that  $ U_i \subseteq U_i$ for any $i \in \sigma \cap \tau$.
If, on the other hand, $x_\sigma \prod_{i \in \tau} (1-x_i) \in I_\C$ and $\sigma \cap \tau = \emptyset$, then $\rho_v \in I_\C$ for each $v \in \{0,1\}^n$ such that $\supp(v) \supseteq \sigma$
and $\supp(v) \cap \tau = \emptyset$.  Since $\rho_v(v) = 1$, it follows that $v \notin \C$ for any $v$ with  $\supp(v) \supseteq \sigma$
and $\supp(v) \cap \tau = \emptyset$.  To see this,
recall from the original definition of $I_\C$ that for all $c \in \C$, $f(c) = 0$ for any $f \in I_\C$; it follows
that $\rho_v(c) = 0$ for all $c \in \C$.  Because $\C = \C(\U)$, the fact that $v \notin \C$ for any $v$ such that $\supp(v) \supseteq \sigma$
and $\supp(v) \cap \tau = \emptyset$ implies 
$\bigcap_{i \in \sigma} U_i \setminus \bigcup_{j \in \tau} U_j = \emptyset.$
We can thus conclude that $U_\sigma \subseteq \bigcup_{j \in \tau} U_j.$
\end{proof}

\noindent Lemma~\ref{lemma:corresp} allows us to extract RF structure from the different types of relations that appear in $I_\C$:
\begin{itemize}
\item Boolean relations: $\{x_i(1-x_i)\}$.  The relation
$x_i(1-x_i)$ corresponds to $U_i \subseteq U_i$, which does not contain any information about the code $\C$.
\item Type 1 relations: $\{x_\sigma\}$.  The relation $x_\sigma$ corresponds to $U_\sigma = \emptyset$.
\item Type 2 relations: $\big\{ x_\sigma \prod_{i \in \tau} (1-x_i) \mid \sigma,\tau \neq \emptyset, \;\sigma \cap \tau = \emptyset, \; U_\sigma \neq \emptyset \text{ and } 
\bigcup_{i \in \tau} U_i \neq X \big \}$.  \\The relation $x_\sigma \prod_{i \in \tau} (1-x_i)$
corresponds to $U_\sigma \subseteq \bigcup_{i \in \tau} U_i$.
\item Type 3 relations: $\big\{\prod_{i \in \tau} (1-x_i)\big\}$. The relation $\prod_{i \in \tau} (1-x_i)$ corresponds to $X \subseteq \bigcup_{i \in \tau} U_i$.
\end{itemize}
The somewhat complicated requirements on the Type 2 relations ensure that they do not include polynomials that are multiples of Type 1, Type 3, or Boolean relations.  Note that the constant polynomial $1$ may appear as both a Type 1 and a Type 3 relation, but only if $X = \emptyset$.  The four types of relations listed above are otherwise disjoint.   Type 3 relations only appear if $X$ is fully covered by the receptive fields, and there is thus no all-zeros codeword corresponding to an ``outside'' point.

Not all elements of $I_\C$ are one of the above types, of course, but we will see that these are sufficient to generate $I_\C$.  This follows from the observation (see Lemma~\ref{lemma:gen-types}) that the neural ideal $J_\C$ is generated by the Type 1, Type 2 and Type 3 relations, and recalling that $I_\C$ is obtained from $J_\C$ be adding in the Boolean relations (Lemma~\ref{lemma:explicit-relations}).
At the same time, not all of these relations are necessary to generate the neural ideal.
Can we eliminate redundant relations to come up with
a ``minimal'' list of generators for $J_\C$, and hence $I_\C$, 
that captures the essential RF structure of the code?
This is the goal of the next section.

\section{Pseudo-monomials \& a canonical form for the neural ideal} \label{sec:canonical-form}

The Type 1, Type 2, and Type 3 relations are all products of linear terms of the form $x_i$ and $1-x_i$, and are thus very similar to monomials.  By analogy with square-free monomials and square-free monomial ideals \cite{MillerSturmfels}, we define the notions of pseudo-monomials and pseudo-monomial ideals.
Note that we do not allow repeated indices in our definition of pseudo-monomial, so the Boolean relations are explicitly excluded.

\begin{definition}
If $f \in \F_2[x_1,\ldots,x_n]$ has the form $f = \prod_{i \in \sigma} x_i \prod_{j\in \tau} (1-x_j)$ for some $\sigma,\tau \subset [n]$ with $\sigma\cap \tau=\emptyset$, then we say that $f$ is a {\em pseudo-monomial}.  
\end{definition}

\begin{definition}
An ideal $J \subset \F_2[x_1,\ldots,x_n]$ is a {\em pseudo-monomial ideal} if $J$ can be generated by a finite set of pseudo-monomials.  
\end{definition}

\begin{definition}
Let $J \subset \F_2[x_1,\ldots,x_n]$ be an ideal, and $f \in J$ a pseudo-monomial.  We say that $f$ is a {\em minimal} pseudo-monomial of $J$ if there does not exist another pseudo-monomial $g \in J$ with $\deg(g) < \deg(f)$ such that $f = hg$ for some $h \in \F_2[x_1,\ldots,x_n]$.
\end{definition}

\noindent By considering the set of \textit{all} minimal pseudo-monomials in a pseudo-monomial ideal $J$, we obtain a unique and compact description of $J$, which we call the ``canonical form'' of $J$.

\begin{definition}
We say that a pseudo-monomial ideal $J$ is in {\em canonical form} if we present it as $J = \langle f_1,\ldots,f_l \rangle$, where the set $CF(J) \od \{f_1, \ldots, f_l\}$ is the set of \textit{all} minimal pseudo-monomials of $J$.  Equivalently, we refer to $CF(J)$ as the {\em canonical form} of $J$.  
\end{definition}

\noindent Clearly, for any pseudo-monomial ideal $J \subset \F_2[x_1,\ldots,x_n]$, $CF(J)$ is unique and $J = \langle CF(J) \rangle$.  On the other hand, it is important to keep in mind that although $CF(J)$ consists of minimal pseudo-monomials, it is not necessarily a minimal set of generators for $J$.  To see why, consider the pseudo-monomial ideal $J = \langle x_1(1-x_2), x_2(1-x_3) \rangle.$ 
This ideal in fact contains a third minimal pseudo-monomial:
$x_1(1-x_3)=(1-x_3)\cdot[x_1(1-x_2)]+x_1\cdot[x_2(1-x_3)].$
It follows that $CF(J) = \{ x_1(1-x_2), x_2(1-x_3), x_1(1-x_3) \}$, but clearly we can remove $x_1(1-x_3)$ from this set and still generate $J$.

For any code $\C$, the neural ideal $J_\C$ is a pseudo-monomial ideal because $J_\C = \langle \{\rho_v \mid v \notin \C\}\rangle$, and each of the $\rho_v$s is a pseudo-monomial.  (In contrast, $I_\C$ is rarely a pseudo-monomial ideal, because it is typically necessary to include the Boolean relations as generators.)
Theorem~\ref{thm:canonical-form} describes the canonical form of $J_\C$.
In what follows, we say that $\sigma \subseteq [n]$ is \textit{minimal with respect to} property $P$ if $\sigma$ satisfies $P$, but $P$ is not satisfied for any $\tau \subsetneq \sigma$.  For example, if $U_\sigma = \emptyset$ and for all $\tau \subsetneq \sigma$ we have $U_\tau \neq \emptyset$, then we say that ``$\sigma$ is minimal w.r.t. $U_\sigma = \emptyset$.''

\begin{theorem}\label{thm:canonical-form}
Let $\C \subset \{0,1\}^n$ be a neural code, and let $\U = \{U_1,\ldots,U_n\}$ be any collection of open sets (not necessarily convex) in a nonempty stimulus space $X$ such that $\C = \C(\U)$.   The canonical form of $J_\C$ is:
\begin{eqnarray*}
J_{\C}  &=& \big \langle \big\{x_\sigma \mid \sigma \text{ is minimal w.r.t. } U_\sigma = \emptyset \big\},\\
&& \big\{x_\sigma \prod_{i\in \tau} (1-x_i) \mid \sigma,\tau \neq \emptyset,\; \sigma \cap \tau =\emptyset,\; U_\sigma \neq \emptyset, \; \bigcup_{i \in \tau} U_i \neq X , \text{ and } \sigma, \tau 
\text{ are each minimal }\\ && \text{ w.r.t. } U_\sigma\subseteq \bigcup_{i\in \tau}U_i  \big\},
 \big\{\prod_{i \in \tau} (1-x_i) \mid \tau \text{ is minimal w.r.t. } X \subseteq \bigcup_{i\in \tau} U_i\big\}  \big \rangle.
\end{eqnarray*}
We call the above three (disjoint) sets of relations comprising $CF(J_\C)$ the minimal Type 1 relations, the minimal Type 2 relations, and the minimal Type 3 relations, respectively.
\end{theorem}

\noindent The proof is given in Section~\ref{sec:proof2}.
Note that, because of the uniqueness of the canonical form, if we are given $CF(J_\C)$ then Theorem~\ref{thm:canonical-form} allows us to
read off the corresponding (minimal) relationships that must be satisfied by any receptive field representation 
of the code as $\C = \C(\U)$:

\begin{itemize}
\item Type 1: $x_\sigma \in CF(J_\C)$ implies that $U_\sigma = \emptyset$, but all lower-order intersections $U_\gamma$ with $\gamma \subsetneq \sigma$ are non-empty.
\item Type 2: $x_\sigma \prod_{i\in \tau} (1-x_i)\in CF(J_\C)$ implies that $U_\sigma\subseteq \bigcup_{i\in \tau}U_i$, but
no lower-order intersection is contained in $\bigcup_{i\in \tau}U_i$, and all the $U_i$s are necessary for $U_\sigma\subseteq \bigcup_{i\in \tau}U_i$.
\item Type 3: $\prod_{i \in \tau} (1-x_i) \in CF(J_\C)$ implies that $X \subseteq \bigcup_{i \in \tau} U_i,$ but $X$ is not contained in any lower-order union $\bigcup_{i \in \gamma} U_i$ for $\gamma \subsetneq \tau$.
\end{itemize}
The canonical form $CF(J_\C)$ thus provides a minimal description of the RF structure dictated by the code $\C$.

The Type 1 relations in $CF(J_\C)$ can be used to obtain a (crude) lower bound on the minimal embedding dimension of the neural code, as defined in Section~\ref{sec:beyond}.
Recall Helly's theorem (Section~\ref{sec:helly-nerve}), and observe that if $x_\sigma \in CF(J_\C)$ then $\sigma$ is minimal with respect to $U_\sigma=\emptyset$; this in turn implies that $|\sigma|\leq d+1$.   (If $|\sigma| > d+1$, by minimality all $d+1$ subsets intersect and by Helly's theorem we must have $U_\sigma \neq \emptyset.$)
We can thus obtain a lower bound on the minimal embedding dimension $d$ as 
$$d \geq \max_{\{\sigma \mid x_\sigma \in CF(J_\C)\}} |\sigma|-1,$$ 
where the maximum is taken over all $\sigma$ such that $x_\sigma$ is a Type 1 relation in $CF(J_\C)$.  
This bound only depends on $\Delta(\C)$, however, and
does not provide any insight regarding the different minimal embedding dimensions observed in the examples of Figure 3.  
These codes have no Type 1 relations in their canonical forms, but they are nicely differentiated by their minimal Type 2 and Type 3 relations.
From the receptive field arrangements depicted in Figure 3, we can easily write down $CF(J_\C)$ for each of these codes.
\begin{itemize}
\item[A.] $CF(J_\C) = \{0\}.$  There are no relations here because $\C = 2^{[3]}$.
\item[B.]  $CF(J_\C) = \{1-x_3\}.$  This Type 3 relation reflects the fact that $X = U_3$.
\item[C.] $CF(J_\C) = \{x_1(1-x_2), x_2(1-x_3), x_1(1-x_3)\}.$  These Type 2 relations correspond to $U_1 \subset U_2$, $U_2 \subset U_3$, and $U_1 \subset U_3$.
Note that the first two of these receptive field relationships imply the third; correspondingly, the third canonical form relation satisfies:
$x_1(1-x_3)=(1-x_3)\cdot[x_1(1-x_2)]+x_1\cdot[x_2(1-x_3)].$
\item[D.] $CF(J_\C) = \{(1-x_1)(1-x_2)\}.$ This Type 3 relation reflects $X = U_1 \cup U_2$, and implies $U_3 \subset U_1 \cup U_2$.
\end{itemize}

\section{Proof of Theorem~\ref{thm:canonical-form}}\label{sec:proof2}

We begin by showing that $J_\U,$ first defined in Lemma~\ref{lemma:JU}, can be generated using the Type 1, Type 2 and Type 3 relations introduced in Section~\ref{sec:types}.   From the proof of Theorem~\ref{thm:ideal-equivalence},
we know that $J_\U = J_{\C(\U)},$ so the following lemma in fact shows that $J_{\C(\U)}$ is generated by the Type 1, 2 and 3 relations as well.

\begin{lemma} \label{lemma:gen-types}
For $\U = \{U_1,\ldots, U_n\}$ a collection of sets in a stimulus space $X$,
\begin{eqnarray*}
J_{\U}  &=& \big \langle \{x_\sigma \mid U_\sigma = \emptyset \}, \big\{\prod_{i \in \tau} (1-x_i) \mid X \subseteq \bigcup_{i\in \tau} U_i\big\},
\\
&& \big\{x_\sigma \prod_{i\in \tau} (1-x_i) \mid \sigma,\tau \neq \emptyset,\; \sigma \cap \tau =\emptyset,\; U_\sigma \neq \emptyset,\; \bigcup_{i \in \tau} U_i \neq X ,  \text{ and } U_\sigma\subseteq \bigcup_{i\in \tau}U_i  \big\} \big \rangle.
\end{eqnarray*}
$J_\U$ (equivalently, $J_{\C(\U)}$) is thus generated by the Type 1, Type 3 and Type 2 relations, respectively.
\end{lemma}

\begin{proof} Recall that in Lemma~\ref{lemma:JU} we defined $J_\U$ as:
$$J_\U \od  \big \langle \{ x_\sigma \prod_{i \in \tau} (1-x_i) \mid \sigma \cap \tau = \emptyset \;\;\mathrm{ and }\;\; U_\sigma \subseteq \bigcup_{i \in \tau}U_i \} \big \rangle.$$
Observe that if $U_\sigma = \emptyset$, then we can take $\tau = \emptyset$ to
obtain the Type 1 relation $x_\sigma$, where we have used the fact that $\prod_{i \in \emptyset}(1-x_i) = 1$.  Any other relation with $U_\sigma = \emptyset$ and $\tau \neq \emptyset$ would be a multiple of $x_\sigma$.  We can thus write:
$$J_\U = \big \langle \{x_\sigma \mid U_\sigma = \emptyset\}, \{ x_\sigma \prod_{i \in \tau} (1-x_i) \mid  \tau \neq \emptyset, \; \sigma \cap \tau = \emptyset, \;U_\sigma \neq \emptyset, \;\mathrm{ and }\;\; U_\sigma \subseteq \bigcup_{i \in \tau}U_i \} \big \rangle.$$
Next, if $\sigma = \emptyset$ in the second set of relations above, then we have the relation
 $\prod_{i \in \tau} (1-x_i)$ with
$U_\emptyset = X \subseteq \bigcup_{i \in \tau}U_i.$  Splitting off these Type 3 relations, and removing multiples of them that occur if $\bigcup_{i \in \tau} U_i = X $, we obtain the desired result.
\end{proof}

Next, we show that $J_\U$ can be generated by reduced sets of  the Type 1, Type 2 and Type 3 relations given above.
First, consider the Type 1 relations in Lemma~\ref{lemma:gen-types}, and observe that if $\tau\subseteq \sigma$, then $x_\sigma$ is a multiple of $x_\tau$.  We can thus reduce the set of Type 1 generators needed by taking only those corresponding to minimal $\sigma$ with $U_\sigma=\emptyset$:
$$\langle \{x_\sigma \mid U_\sigma = \emptyset\}\rangle = \langle \{x_\sigma \mid \sigma \text{ is minimal w.r.t. } U_\sigma = \emptyset\}\rangle.$$
Similarly, we find for the Type 3 relations:
$$\big\langle \big\{\prod_{i \in \tau} (1-x_i) \mid X \subseteq \bigcup_{i\in \tau} U_i\big\}\big \rangle = 
\big\langle \big\{\prod_{i \in \tau} (1-x_i) \mid \tau  \text{ is minimal w.r.t. } X \subseteq \bigcup_{i\in \tau} U_i\big\}\big \rangle. $$
Finally, we reduce the Type 2 generators.  If $\rho\subseteq \sigma$ and $x_\rho\prod_{i\in \tau} (1-x_i) \in J_\U$, then we also have $x_\sigma\prod_{i\in \tau} (1-x_i) \in J_\U$.  So we can restrict ourselves to only those generators for which $\sigma$ is minimal with respect to $U_\sigma\subseteq \bigcup_{i\in \tau}U_i$.   Similarly, we can reduce to minimal $\tau$ such that $U_\sigma\subseteq \bigcup_{i\in \tau}U_i$.  In summary:
\begin{eqnarray*}
&&\big\langle \big\{x_\sigma \prod_{i\in \tau} (1-x_i) \mid \sigma,\tau \neq \emptyset,\; \sigma \cap \tau =\emptyset,\; U_\sigma \neq \emptyset,\;\bigcup_{i \in \tau} U_i \neq X, \text{ and } U_\sigma\subseteq \bigcup_{i\in \tau}U_i  \big\} \big \rangle = \\
&&\big\langle \big\{x_\sigma \prod_{i\in \tau} (1-x_i) \mid \sigma,\tau \neq \emptyset,\; \sigma \cap \tau =\emptyset,\; U_\sigma \neq \emptyset,\;\bigcup_{i \in \tau} U_i \neq X, \text{ and } \sigma, \tau \text{ are each minimal }\\
&&\text{ w.r.t. } U_\sigma\subseteq \bigcup_{i\in \tau}U_i  \big\} \big \rangle.
\end{eqnarray*}

\noindent We can now prove Theorem~\ref{thm:canonical-form}.

\begin{proof}[Proof of Theorem~\ref{thm:canonical-form}]

Recall that $\C = \C(\U)$, and that by the proof of Theorem~\ref{thm:ideal-equivalence} we have $J_{\C(\U)} = J_\U$. By the reductions given above for the Type 1, 2 and 3 generators, we also know that $J_\U$ can be reduced to the form given in the statement of Theorem~\ref{thm:canonical-form}.  We conclude that $J_\C$ can be expressed in the desired form.

To see that $J_\C$, as given in the statement of Theorem~\ref{thm:canonical-form}, is in canonical form, we must show that the given set of generators is exactly the complete set of minimal
pseudo-monomials for $J_\C$.  First, observe that the generators are all pseudo-monomials.   If $x_\sigma$ is one of the Type 1 relations, and $x_\sigma \in \langle g \rangle$
with $\langle x_\sigma \rangle \neq \langle g \rangle$, then $g = \prod_{i \in \tau} x_i$ for some $\tau \subsetneq \sigma$.  Since $U_\tau \neq \emptyset$, however, it follows that $g \notin J_\C$ and hence $x_\sigma$ is a minimal pseudo-monomial of $J_\C$.  By a similar argument, the Type 2 and Type 3 relations above are also minimal pseudo-monomials in $J_\C$.   

It remains only to show that there are no additional minimal pseudo-monomials in $J_\C$.  Suppose $f = x_\sigma \prod_{i \in \tau}(1-x_i)$ is a minimal pseudo-monomial in $J_\C$.    By Lemma~\ref{lemma:corresp}, $U_\sigma \subseteq \bigcup_{i \in \tau} U_i$ and $\sigma \cap \tau = \emptyset$, so $f$ is a generator in the original definition of $J_\U$ (Lemma~\ref{lemma:JU}).  Since $f$ is a minimal pseudo-monomial of $J_\C$, 
there does not exist a $g \in J_\C$ such that $g =  x_{\sigma'} \prod_{i \in \tau'}(1-x_i)$ with either $\sigma' \subsetneq \sigma$ or $\tau' \subsetneq \tau$.  Therefore, $\sigma$ and $\tau$ are each minimal with respect to $U_\sigma \subseteq \bigcup_{i \in \tau} U_i$.  We conclude that $f$ is one of the generators for $J_\C$ given in the statement of Theorem~\ref{thm:canonical-form}.  It is a minimal Type 1 generator if $\tau = \emptyset$, a minimal Type 3 generator if $\sigma = \emptyset$, and is otherwise a minimal Type 2 generator.  The three sets of minimal generators are disjoint because the Type 1, Type 2 and Type 3 relations are disjoint, provided $X \neq \emptyset$.
\end{proof}
Nevertheless, we do not yet know how to infer the minimal embedding dimension from $CF(J_\C)$.  
In Appendix 2 (Section~\ref{sec:appendix2}), we provide a complete list of neural codes on three neurons, up to permutation, and their
respective canonical forms.

\section{Comparison to the Stanley-Reisner ideal}

Readers familiar with the Stanley-Reisner ideal \cite{MillerSturmfels,StanleyBook} will recognize that this kind of ideal is generated by the Type 1 relations of a neural code $\C$.  The corresponding simplicial complex is $\Delta(\C)$, the smallest simplicial complex that contains the code.

\begin{lemma}  Let $\C = \C(\U)$.  The ideal generated by the Type 1 relations,
$\langle x_\sigma \mid U_\sigma=\emptyset\rangle,$ is the Stanley-Reisner ideal of $\Delta(\C)$.
Moreover, if $\supp\C$ is a simplicial complex, then $CF(J_\C)$ contains no Type 2 or Type 3 relations, and $J_\C$ is thus the Stanley-Reisner ideal for $\supp\C$.
\end{lemma}

\begin{proof}  
To see the first statement, observe that the {\em Stanley-Reisner ideal} of a simplicial complex $\Delta$ is the ideal 
$$I_\Delta \od \langle x_\sigma \mid \sigma \notin \Delta \rangle,$$
and recall that $\Delta(\C)=\{\sigma\subseteq[n] \mid \sigma\subseteq\supp(c)$ for some $c\in \C\}$.  As $\C=\C(\U)$, an equivalent characterization is $\Delta(\C)=\{\sigma\subseteq[n]\mid U_\sigma\neq \emptyset\}$.  
Since these sets are equal, so are their complements in $2^{[n]}$: 
$$\{\sigma\subseteq [n]\mid \sigma\notin\Delta(\C)\}=\{\sigma\subseteq [n] \mid U_\sigma = \emptyset\}.$$
Thus, $\langle x_\sigma\mid U_\sigma=\emptyset\rangle= \langle x_\sigma\mid \sigma\notin\Delta(\C)\rangle$, which is the Stanley-Reisner ideal for $\Delta(\C)$.

\
To prove the second statement, suppose that $\supp\C$ is a simplicial complex.  Note that $\C$ must contain the all-zeros codeword, so  $X \supsetneq \bigcup_{i=1}^n U_i$ and there can be no Type 3 relations.
Suppose the canonical form of $J_{\C}$ contains a Type 2 relation $x_\sigma \prod_{i\in \tau} (1-x_i)$, for some $\sigma,\tau \subset [n]$ satisfying $\sigma, \tau \neq \emptyset$,  $\sigma \cap \tau = \emptyset$ and
 $U_\sigma \neq \emptyset$.  The existence of this relation indicates that $\sigma \notin \supp \C$, while there does exist an $\omega \in \C$ such that $\sigma \subset \omega.$
 This contradicts the assumption that $\supp \C$ is a simplicial complex.  We conclude that $J_{\C}$ has no Type 2 relations.
 \end{proof}

The canonical form of $J_\C$ thus enables us to immediately read off, via the Type 1 relations, the minimal forbidden faces of the simplicial complex $\Delta(\C)$ associated to the code, and also the minimal deviations of $\C$ from being a simplicial complex, which are captured by the Type 2 and Type 3 relations.

\chapter{Algorithms for the Canonical Form }

Now that we have established that a minimal description of the RF structure can be extracted from the canonical form of the neural ideal, the most pressing question is the following:
\medskip
 
\noindent \textbf{Question:} How do we find the canonical form $CF(J_\C)$ if all we know is the code $\C$, and we are \textit{not} given a representation of the code as $\C = \C(\U)$?
\medskip

\section{Algorithm \#1}
\noindent In this section we describe an algorithmic method for finding $CF(J_\C)$ from knowledge only of $\C$.
finding the minimal pseudo-monomials.  

\noindent \textbf{Canonical form algorithm 1}
\medskip

\noindent\textbf{Input:} A neural code $\C \subset \{0,1\}^n$.\medskip

\noindent\textbf{Output:} The canonical form of the neural ideal, $CF(J_\C)$.

\begin{itemize}
\item[Step 1:]  From $\C \subset \{0,1\}^n$,  for each $c\in \C$ take the ideal $\p_c = \langle x_i - c_i \mid i=1,...,n\rangle$. 


\item[Step 2:]  Observe that any pseudo-monomial $f \in J_\C$ is a multiple of one of the linear generators of $\p_c$ for each $c \in \mathcal{C}$.
Compute the following set of elements of $J_\C$:
$$\mathcal{M}(J_\C) = \big\{\prod_{c \in \C} g_c \mid g_c = x_i - c_i \text{ for some } i =1,...,n \big\}.$$
$\mathcal{M}(J_\C)$ consists of all polynomials obtained as a product of linear generators $g_a$, 
one for each prime ideal $\p_c$.  $\mathcal{M}(J_\C)$ is therefore the set of generators for the product of the ideals $\prod_{c\in \C}\p_c$.  
\item[Step 3:] Reduce the elements of $\mathcal{M}(J_\C)$ by imposing $x_i(1-x_i) = 0$.  This eliminates elements that are not pseudo-monomials.  It also reduces the degrees of some of the remaining elements, as it implies $x_i^2 = x_i$ and $(1-x_i)^2 = (1-x_i)$.  
We are left with a set of pseudo-monomials of the form $f = \prod_{i \in \sigma} x_i \prod_{j\in \tau} (1-x_j)$ with $\tau \cap \sigma = \emptyset.$
Call this new reduced set $\mathcal{\tilde M}(J_\C).$ 
\item[Step 4:] Finally, remove all elements of $\mathcal{\tilde M}(J_\C)$ that are multiples of lower-degree elements in $\mathcal{\tilde M}(J_\C).$ 
\end{itemize}

\begin{proposition}\label{prop:prop1} The resulting set is the canonical form $CF(J_\C)$.
\end{proposition}

\noindent Note that every polynomial obtained by the canonical form algorithm is a pseudo-monomial of $J_\C$.  This is because the algorithm constructs products of factors of the form $x_i$ or $1-x_i$, and then reduces them in such a way that no index is repeated in the final product, and there are no powers of any $x_i$ or $1-x_i$ factor; we are thus guaranteed to end up with pseudo-monomials.  Moreover, since the products each have at least one factor in each prime ideal of the primary decomposition of $J_\C$, the pseudo-monomials are all in $J_\C$.   Proposition~\ref{prop:prop1} states that this set of pseudo-monomials is precisely the canonical form $CF(J_\C)$.

To prove Proposition~\ref{prop:prop1}, we will make use of the following technical lemma.
Here $z_i, y_i\in \{x_i, 1-x_i\}$, and thus any pseudo-monomial in $\F_2[x_1,\ldots,x_n]$ is of the form $\prod_{j\in \sigma}  z_j$ for some index set $\sigma \subseteq [n]$.  

\begin{lemma}\label{lemma:matching}
If $y_{i_1}\cdots y_{i_m}\in \langle z_{j_1},\ldots, z_{j_\ell}\rangle$ where $\{i_k\}$ and $\{j_r\}$ are each distinct sets of indices, then $y_{i_k}=z_{j_r}$ for some $k\in [m]$ and $r\in [\ell]$.
\end{lemma}

\begin{proof} 
Let $f = y_{i_1}\cdots y_{i_m}$ and $P = \{ z_{j_1},\ldots,z_{j_\ell} \}$.  
Since $f\in \langle P \rangle$, then $\langle P \rangle  = \langle P , f \rangle$, and so $V(\langle P \rangle ) = V(\langle P , f  \rangle)$.  
We need to show that $y_{i_k} = z_{j_r}$ for some pair of indices $i_k, j_r.$
Suppose by way of contradiction that there is no $i_k, j_r$ such that $y_{i_k}=z_{j_r}$.  

Select $a\in \{0,1\}^n$ as follows: for each $j_r \in \{j_1,\ldots,j_\ell\}$, let $a_{j_r} = 0$ 
if $z_{j_r} = x_{j_r}$, and let $a_{j_r} = 1$ if $z_{j_r} = 1-x_{j_r}$; when evaluating at $a$, we thus have $z_{j_r}(a) = 0$ for all $r \in [\ell]$.
Next, for each $i_k\in \omega \od \{i_1,\ldots,i_m\}\backslash \{j_1,..,j_\ell\}$, let $a_{i_k} = 1$ if $y_{i_k} = x_{i_k}$, and let $a_{i_k} = 0$ if $y_{i_k} = 1-x_{i_k}$, so that
$y_{i_k}(a)=1$ for all $i_k \in \omega$.  For any remaining indices $t$, let $a_t=1$.  
Because we have assumed that $y_{i_k}\neq z_{j_r}$ for any $i_k, j_r$ pair, we have
 for any $i \in \{i_1,\ldots,i_m\} \cap \{j_1,\ldots,j_\ell\}$ that $y_i(a) = 1-z_i(a) = 1.$  It follows that $f(a) = 1$.

Now, note that $a\in V(\langle P\rangle )$ by construction.  We must therefore have $a\in V(\langle P, f \rangle )$, and hence $f(a) = 0$, a contradiction.
We conclude that there must be some $i_k, j_r$ with $y_{i_k}=z_{j_r},$ as desired.
\end{proof}

\begin{proof}[Proof of Proposition~\ref{prop:prop1}]
It suffices to show that after Step 4 of the algorithm, the reduced set $\tilde{\mathcal{M}}(J_\C)$ consists entirely of pseudo-monomials of $J_\C$, and includes all \textit{minimal} pseudo-monomials of $J_\C$.  If this is true, then after removing multiples of lower-degree elements in Step 5 we are guaranteed to obtain the set of minimal pseudo-monomials, $CF(J_\C)$, since it is precisely the non-minimal pseudo-monomials that will be removed in the final step of the algorithm.

Recall that $\mathcal{M}(J_\C)$, as defined in Step 3 of the algorithm, is precisely the set of all polynomials $g$ that are obtained by choosing one linear factor
from the generating set of each $\p_c$:
$$\mathcal{M}(J_\C) = \{g = z_{p_1}\cdots z_{p_s} \mid z_{p_i} \text{ is a linear generator of } \p_{c^i} \}.$$ 
Furthermore, recall that $\tilde{\mathcal{M}}(J_\C)$ is obtained from $\mathcal{M}(J_\C)$ by the reductions in Step 4 of the algorithm.
Clearly, all elements of $\tilde{\mathcal{M}}(J_\C)$ are pseudo-monomials that are contained in $J_\C$.

To show that $\tilde{\mathcal{M}}(J_\C)$ contains all \textit{minimal} pseudo-monomials of $J_\C$, we will show that if $f \in J_\C$ is a pseudo-monomial, then there exists another pseudo-monomial $h \in \tilde{\mathcal{M}}(J_\C)$ (possibly the same as $f$) such that $h | f$.  To see this, let
 $f = y_{i_1}\cdots y_{i_m}$ be a pseudo-monomial of $J_\C$.  Then,
$f \in P_i$ for each $i \in [s].$  For a given $P_i = \langle z_{j_1},\ldots,z_{j_\ell}\rangle,$ by Lemma~\ref{lemma:matching} we have $y_{i_k} = z_{j_r}$ for some $k \in [m]$ and $r \in [\ell]$.  In other words, each prime ideal $P_i$ has a generating term, call it $z_{p_i},$ that appears as one of the linear factors of 
$f$.  Setting $g = z_{p_1}\cdots z_{p_s}$, it is clear that $g \in \mathcal{M}(J_\C)$ and that either $g | f$, or $z_{p_i} = z_{p_j}$ for some distinct pair $i,j$.  By removing repeated factors in $g$ one obtains a pseudo-monomial $h \in \tilde{\mathcal{M}}(J_\C)$ such that $h | g$ and $h | f$.  If we take $f$ to be a minimal pseudo-monomial, we find
$f = h  \in \tilde{\mathcal{M}}(J_\C)$.  
\end{proof}

\section{Algorithm \# 2}
In practice, we have found it is more practical to perform an inductive version of this algorithm. 
Given a code $\C$, let $|\C| = k$.  Then order the codewords $c^1,...,c^k$.  Let $\C_i = \{c^1,...,c^i\}$. Note that as $J_\C = \bigcap_{i=1}^k \p_{c^i}$, we have $J_{\C_i} = J_{\C_{i-1}}\cap \p_{c^i}$. Using this fact, we proceed inductively, finding each canonical form from the previous one.

\noindent \textbf{Canonical form algorithm 2}
\medskip

\noindent\textbf{Input:} A neural code $\C \subset \{0,1\}^n$.\medskip

\noindent\textbf{Output:} The canonical form of the neural ideal, $CF(J_\C)$.

\begin{enumerate}
\item[Step 1:] Note $\C_1 = \{c^1\}$, so $CF(J_{\C_1}) = \p_{c^1}$. Set $i=1$.
\item[Step 2:] If $i=n$, we are done; $CF(J_\C) = CF(J_{\C_n})$. \\
If $i<n$, let $i=i+1$.  Take the set $M_i = \{ f z_j \mid z_j = x_j - c^i_j, f\in CF(J_{\C_{i-1}})\}$ (the set of products of an element of the canonical form $CF(J_{\C_{i-1}})$ with a generator for $\p_{c^i})$. 
\item[Step 3:] Reduce the set by imposing $x_i (1-x_i) = 0$ as in CF Algorithm 1 to get a set $M_i'$ of pseudo-monomials.
\item[Step 4:] Reduce the set $M_i'$ by removing all elements of $M_i'$ that are multiples of lower degree elements to form a new set $\tilde M_i$.  Then $\tilde M_i = CF(J_{\C_i})$. 
Go back to Step 2.

\end{enumerate}

\begin{proposition} The resulting set $CF(\C_n)$ is  $CF(J_\C)$.
\end{proposition}

\begin{proof} It suffices to show that $\tilde M_i = CF(J_{\C_i})$.

Suppose $h \in CF(J_{\C_i})$. Then $h\in J_{\C_{i-1}}$ also, as $J_{\C_i} =  J_{\C_{i-1}}\cap \p_{c^i}$.  As $h$ is a pseudo-monomial in $J_{\C_{i-1}}$, this means one of the following two cases holds:

- $h\in CF(J_{\C_{i-1}})$:  Then, as $h\in \p_{c^i}$ and $h$ is a pseudo-monomial, we must have $z_j\big| h$ for some $z_j=x_i-c_j^i$ by Lemma \ref{lemma:matching}.  So $z_j h \in M_i$, and when we reduce by $x_j = x_j^2$ or $(1-x_j) = (1-x_j)^2$, we get $h\in M'_i$.  

- $h\notin CF(J_{\C_{i-1}})$: then as $h$ is a pseudo-monomial in $J_{\C_{i-1}}$, we have $h=gf$ for some pseudo-monomials $g,f$ with $f\in CF(J_{\C_{i-1}})$.  And as $h$ is a pseudo-monomial, then $g,f$ share no indices.  As $h\in \p_{c^i}$, then by the above technical lemma there is some linear factor $z_j = x_j-c^i_j$ with $z_j \big|h$.  It can't be that $z_j\big| f$, or else $f(c^i) = 0$, and thus $f(c) = 0$ for all $c\in \C_i$, so $f\in J_{\C_2}$ would be more minimal than $h$, contradicting $h\in CF(J_{\C_i})$.  Thus, $ z_j\big| g$.  But the pseudo monomial $z_jf$ is in $\tilde M_i$, so if $g\neq z_j$ then $h$ is not minimal.  So $h=z_if$, and thus $h\in M_i'$.

Thus in both cases, $h\in M_i'$. The only way that $h$ would be removed in Step 4 and not appear in $\tilde M_i$ is if there were some pseudo-monomial $f\in \tilde M_i$ with $f\big| h$, but as $f\in J_{\C_i}$, this would contradict the minimality given by $h\in CF(J_{\C_i})$.  Thus $h\in \tilde M_i$.

Now, if $h\in \tilde M_i$  then clearly $h(c) = 0$ for all $c\in \C_i$, so $h\in J_{\C_i}$.  As $h$ has been reduced by $x_i = x_i^2$ and $(1-x_i)=(1-x_i)^2$, then $h$ is a pseudo-monomial.  As shown above, all pseudo-monomials in $CF(J_{\C_i})$ appear in $\tilde M_i$, so if $h$ is not minimal, $h$ will be removed by the reduction step.  Thus, if $h\in \tilde M_i$, then $h\in CF(J_{\C_i})$.

\end{proof}

The MATLAB code to compute the canonical form using this strategy is found in the Appendix.

\section{An example}
Now we are ready to use the canonical form algorithm in an example, illustrating how to obtain a possible arrangement of convex receptive fields from a neural code.

Suppose a neural code $\C$ has the following 13 codewords, and 19 missing words: 
\begin{small}
\begin{eqnarray*}
\C&=&\begin{array}{ccccc} \{ 00000, & 10000, &  01000, & 00100, & 00001, \\ 
  11000, & 10001, & 01100, & 00110, & 00101, \\ 
 00011, &  11100, & 00111 \} & & \end{array}\\
  \{0,1\}^5\backslash \C &=& \begin{array}{ccccc}
 \{ 00010, & 10100, & 10010, & 01010, & 01001, \\
 11010, & 11001, & 10110, & 10101, & 10011,\\
 01110, & 01101, & 01011, & 11110, & 11101,\\
 11011, &    10111,& 01111,& 11111 \}.& \\
 \end{array}
\end{eqnarray*}
 \end{small}
 
 Thus, the neural ideal $J_{\C}$ has 19 generators, using the original definition $J_\C = \langle \{ \rho_v \mid v \notin \C\}\rangle$:
 
 \begin{small}
 $$J_\C = \big \langle x_4(1-x_1)(1-x_2)(1-x_3)(1-x_5),
 x_1x_3(1-x_2)(1-x_4)(1-x_5), x_1x_4(1-x_2)(1-x_3)(1-x_5),$$ $$
 x_2x_4(1-x_1)(1-x_3)(1-x_5), x_2x_5(1-x_1)(1-x_3)(1-x_4), x_1x_2x_4(1-x_3)(1-x_5),
 $$ $$x_1x_2x_5(1-x_3)(1-x_4),
 x_1x_3x_4(1-x_2)(1-x_5), x_1x_3x_5(1-x_2)(1-x_4),
 x_1x_4x_5(1-x_2)(1-x_3), $$ $$x_2x_3x_4(1-x_1)(1-x_5),
 x_2x_3x_5(1-x_1)(1-x_4), x_2x_4x_5(1-x_1)(1-x_3),
  x_1x_2x_3x_4(1-x_5)$$ $$x_1x_2x_3x_5(1-x_4), x_1x_2x_4x_5(1-x_3), x_1x_3x_4x_5(1-x_2), 
  x_2x_3x_4x_5(1-x_1), x_1x_2x_3x_4x_5\big \rangle.$$
  \end{small}

\noindent Despite the fact that we are considering only five neurons, this looks like a complicated ideal. Considering the canonical form of $J_\C$ will help us to extract the relevant combinatorial information and allow us to create a possible 
arrangement of receptive fields $\U$ that realizes this code as $\C = \C(\U)$.
Following Step 1-2 of our canonical form algorithm, we take the products of linear generators of the ideals $\p_c$. 
 Then, as described in Steps 3-4 of the algorithm, we reduce by the relation $x_i(1-x_i)=0$ (note that this gives us $x_i=x_i^2$ and hence we can say $x_i^k=x_i$ for any $k>1$).  We also remove any polynomials that are multiples of smaller-degree pseudo-monomials in our list.  This process leaves us with six minimal pseudo-monomials, yielding the canonical form:  
$$J_{\C} = \langle CF(J_\C) \rangle = \langle x_1x_3x_5,\, x_2x_5,\, x_1x_4,\, x_2x_4,\, x_1x_3(1-x_2),\, x_4(1-x_3)(1-x_5)\rangle.$$
Note in particular that every  generator we originally put
  in $J_{\mathcal C}$ is a multiple of one of the six
  relations in $CF(J_\C)$.  Next, we consider what the relations in $CF(J_\C)$ tell us 
  about the arrangement of receptive fields that would be needed to realize the code as $\C = \C(\U)$.

\begin{enumerate}

\item $x_1x_3x_5 \in CF(J_\C) \Rightarrow U_1\cap U_3\cap U_5=\emptyset$, while $U_1\cap U_3, U_3\cap U_5$ and $U_1\cap U_5$ are all nonempty.

\item $x_2x_5 \in CF(J_\C) \Rightarrow U_2\cap U_5=\emptyset$, while $U_2, U_5$ are both nonempty.  

\item $x_1x_4 \in CF(J_\C) \Rightarrow U_1\cap U_4=\emptyset$, while $U_1, U_4$ are both nonempty.

\item $x_2x_4 \in CF(J_\C) \Rightarrow U_2\cap U_4=\emptyset$, while $U_2, U_4$ are both nonempty.

\item $x_1x_3(1-x_2)\in CF(J_\C) \Rightarrow U_1\cap U_3\subseteq U_2$, while $U_1\not\subseteq U_2, U_3\not\subseteq U_2$, and $U_1\cap U_3\neq \emptyset$.

\item $x_4(1-x_3)(1-x_5)\in CF(J_\C) \Rightarrow U_4\subseteq U_3\cup U_5$, while $U_4\neq \emptyset$, and that $U_4\not\subseteq U_3, U_4\not\subseteq U_5$.  

\end{enumerate}

The minimal Type 1 relations (1-4) tell us that we should draw $U_1, U_3$ and $U_5$ with all pairwise intersections, 
but leaving a ``hole'' in the middle since the triple intersection is empty.  Then $U_2$ should be drawn to intersect $U_1$ and $U_3$, but not $U_5$.
Similarly, $U_4$ should intersect $U_3$ and $U_5$, but not $U_1$ or $U_2$.
The minimal Type 2 relations (5-6) tell us that $U_2$ should be drawn to contain the intersection $U_1 \cap U_3$, while $U_4$ lies in the union $U_3 \cup U_5$, but
is not contained in $U_3$ or $U_5$ alone.  There are no minimal Type 3 relations, as expected for a code that includes the all-zeros codeword.

Putting all this together, and assuming convex receptive fields, we can completely infer the receptive field structure, and draw
the corresponding picture (see Figure 4).
It is easy to verify that the code $\C(\U)$ of the pictured arrangement indeed coincides with $\C$.  

\begin{figure}[h]
\centering
   \vspace{-.1in}
   \includegraphics[width=2.5in]{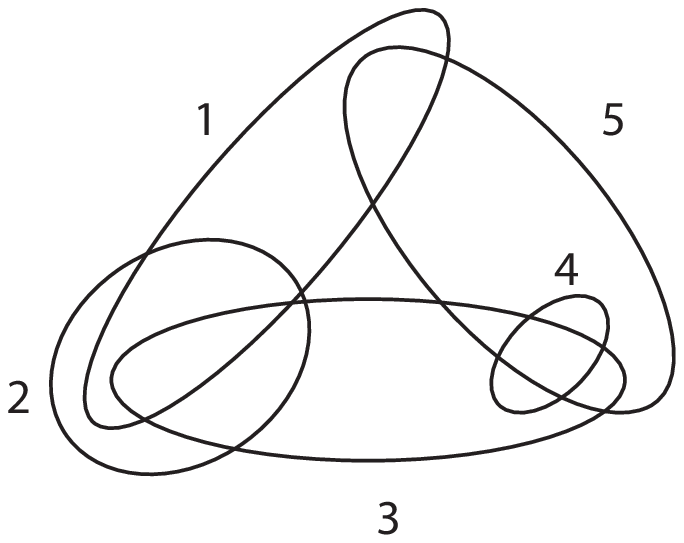}    
    \vspace{-.3in}
   \caption{\small An arrangement of five sets that realizes $\C$ as $\C(\U)$. }
 \end{figure}

\chapter{Primary Decomposition} \label{chap:prim-decomp}

Let $\C \subset \{0,1\}^n$ be a neural code.  The primary decomposition of $I_\C$ is boring: $$I_\C = \bigcap_{c \in \C} m_c,$$ 
where $m_v$ for any $v \in \{0,1\}^n$ is the maximal ideal $I(v)$ defined in Section~\ref{sec:spec}.
This simply expresses $I_\C$ as the intersection of all maximal ideals $m_c$ for $c \in \C$, because
the variety $\C = V(I_\C)$ is just a finite set of points and the primary decomposition reflects no additional structure of the code.

On the other hand, the primary decomposition of the neural ideal $J_\C$ retains the full combinatorial structure of $\C$.  Indeed, we have seen that computing this decomposition is a critical step towards obtaining $CF(J_\C)$, which captures the receptive field structure of the neural code.  In this section, we describe the primary decomposition of $J_\C$ and discuss its relationship to some natural decompositions of the neural code.  We end with an algorithm for obtaining primary decomposition of any pseudo-monomial ideal.

\section{Primary decomposition of the neural ideal}

We begin by defining some objects related to $\F_2[x_1,\ldots,x_n]$ and $\{0,1\}^n$, without reference to any particular neural code.
For any $a \in \{0,1,*\}^n$, we define the variety
$$V_a \od \{v \in \{0,1\}^n \mid v_i = a_i \text{ for all } i \text{ s.t. } a_i \neq *\} \subseteq \{0,1\}^n.$$
This is simply the subset of points compatible with the word ``$a$'', where $*$ is viewed as a ``wild card'' symbol.  Note that $V_v = \{v\}$ for any $v \in \{0,1\}^n$.
We can also associate a prime ideal to $a$,
$$\p_a \od \langle \{x_i - a_i \mid a_i \neq *\}\rangle \subseteq \F_2[x_1,\ldots,x_n],$$
consisting of polynomials in $\F_2[x_1,\ldots,x_n]$ that vanish on all points compatible with $a$.  To obtain all such polynomials,
we must add in the Boolean relations (see Section~\ref{sec:lemma-proofs}):
$$\q_a \od I(V_a) = \p_a + \langle x_1^2-x_1,\ldots,x_n^2-x_n \rangle.$$
Note that $V_a = V(\p_a) = V(\q_a)$.

Next, let's relate this all to a code $\C \subset \{0,1\}^n$. Recall the definition of the neural ideal,
$$J_\C \od  \langle \{ \rho_v \mid v \notin \C\}\rangle = \langle \{ \prod_{i=1}^n((x_i-v_i)-1) \mid v \notin \C\}\rangle.$$
We have the following correspondences.

 \begin{lemma}\label{lemma:lemma6}  
 $J_\C \subseteq \p_a \Leftrightarrow V_a \subseteq \C.$
 \end{lemma}

 \begin{proof}
 ($\Rightarrow$)  $J_\C \subseteq \p_a \Rightarrow V(\p_a) \subseteq V(J_\C).$  Recalling that  $V(\p_a) = V_a$ and   $V(J_\C) = \C$, this gives $V_a \subseteq \C.$\\
($\Leftarrow$) $V_a \subseteq \C \Rightarrow I(\C) \subseteq I(V_a) \Rightarrow I_\C \subseteq \q_a.$  Recalling that both $I_\C$ and $\q_a$ differ
from $J_\C$ and $\p_a$, respectively, by the addition of the Boolean relations, we obtain $J_\C \subseteq \p_a$.
 \end{proof}

\begin{lemma}\label{lemma:lemma7}
 For any $a,b\in\{0,1,*\}^n$,
$V_a\subseteq V_b \Leftrightarrow \mathbf{p}_b\subseteq\mathbf{p}_a.$
\end{lemma}

\begin{proof}
($\Rightarrow$) Suppose $V_a\subseteq V_b$. Then, for any $i$ such that $b_i\neq *$ we have $a_i=b_i$. It follows that each generator of $\mathbf{p}_b$ is also in $\mathbf{p}_a$, so $\p_b \subseteq \p_a$.
($\Leftarrow$) Suppose $\mathbf{p}_b\subseteq\mathbf{p}_a$. Then, $V_a=V(\mathbf{p}_a)\subseteq V(\mathbf{p}_b)=V_b.$
\end{proof}

Recall that a an ideal $\p$ is said to be a {\em minimal prime} over $J$ if $\p$ is a prime ideal that contains $J$, and there is no other prime ideal $\p'$ such that $\p \supsetneq \p' \supseteq J$.  Minimal primes $\p_a \supseteq J_\C$ correspond to maximal varieties $V_a$ such that $V_a \subseteq \C$.
Consider the set
 $$\A_\C \od \{ a \in \{0,1,*\}^n \mid V_a \subseteq \C\}.$$
We say that $a \in \A_\C$ is {\em maximal} if there does not exist another element $b \in \A_\C$ such that $V_a \subsetneq V_b$ (i.e., $a \in \A_\C$ is maximal if $V_a$ is maximal such that $V_a \subseteq \C$).

\begin{lemma}\label{lemma:correspondence}
The element $a \in \A_\C$ is maximal if and only if $\p_a$ is a minimal prime over $J_\C$.
\end{lemma}

\begin{proof}
Recall that $a \in \A_\C \Rightarrow V_a \subseteq \C$, and hence $J_\C \subseteq \p_a$ (by Lemma~\ref{lemma:lemma6}).
($\Rightarrow$) Let $a \in \A_\C$ be maximal, and choose $b \in \{0,1,*\}$ such that 
$J_\C \subseteq \p_b \subseteq \p_a$.
By Lemmas~\ref{lemma:lemma6} and~\ref{lemma:lemma7}, $V_a \subseteq V_b \subseteq \C$.  Since $a$ is maximal, we conclude that $b = a$, and hence $\p_b = \p_a$.  
It follows that $\p_a$ is a minimal prime over $J_\C$.  ($\Leftarrow$) Suppose $\p_a$ is a minimal prime over $J_\C$.  Then by Lemma~\ref{lemma:lemma6}, $a \in \A_\C$.  Let
$b$ be a maximal element of $\A_\C$ such that $V_a \subseteq V_b \subseteq \C$.  Then $J_\C \subseteq \p_b \subseteq \p_a$.  Since $\p_a$ is a minimal prime over $J_\C$, $\p_b = \p_a$ and hence $b = a$.  Thus $a$ is maximal in $\A_\C$.
\end{proof}

\noindent We can now describe the primary decomposition of $J_\C$.  Here we assume the neural code $\C \subseteq \{0,1\}^n$ is non-empty, so that $J_\C$ is a proper pseudo-monomial ideal.

\begin{theorem}\label{thm:prim-decomp}
$J_\C = \bigcap_{i=1}^\ell \p_{a_i}$ is the unique irredundant primary decomposition of $J_\C$, where $\p_{a_1},\ldots,\p_{a_\ell}$ are the minimal primes over $J_\C$. 
\end{theorem}

\noindent The proof is given in Section~\ref{sec:prim-decomp-proof}.

  Combining this theorem
with Lemma~\ref{lemma:correspondence}, we have:

\begin{corollary}
$J_\C = \bigcap_{i=1}^\ell \p_{a_i}$ is the unique irredundant primary decomposition of $J_\C$, where 
$a_1,\ldots,a_\ell$ are the maximal elements of $A_\C$.
\end{corollary}

Proof of Theorem~\ref{thm:prim-decomp}:

Recall that $J_\C$ is always a proper pseudo-monomial ideal for any nonempty neural code $\C \subseteq \{0,1\}^n$. Theorem~\ref{thm:prim-decomp} is thus a direct consequence of the following proposition.

\begin{proposition}\label{prop:primarydec}
Suppose $J \subset \F_2[x_1,\ldots,x_n]$ is a proper pseudo-monomial ideal. Then, $J$ has a unique irredundant primary decomposition of the form
$J = \bigcap_{a \in \A} \p_a,$
where $\{\p_a\}_{a \in \A}$ are the minimal primes over $J$.
\end{proposition}

\begin{proof}
By Proposition~\ref{prop:prim-decomp}, we can always (algorithmically) obtain an irredundant set $\PP$ of prime ideals such that $J = \bigcap_{I \in \PP} I$.
Furthermore, each $I \in \PP$ has the form $I = \langle z_{i_1},\ldots,z_{i_k}\rangle$, where $z_i \in \{ x_i, 1-x_i\}$ for each $i$.
Clearly, these ideals are all prime ideals of the form $\p_a$ for $a \in \{0,1,*\}$.  
It remains only to show that this primary decomposition is unique, and that
the ideals $\{\p_a\}_{a \in \A}$ are the minimal primes over $J$.  
This is a consequence of some well-known facts summarized in Lemmas~\ref{lemma:inter} and~\ref{lemma:rad}, below.  First, observe by Lemma~\ref{lemma:inter} that $J$ is a radical ideal.  Lemma~\ref{lemma:rad} then tells us that the decomposition in terms of minimal primes is the unique irredundant primary decomposition for $J$.
\end{proof}

\begin{lemma}\label{lemma:inter}
If $J$ is the intersection of prime ideals, $J=\bigcap_{i=1}^\ell \mathbf{p}_i$, 
then $J$ is a radical ideal. 
\end{lemma}
\begin{proof}
Suppose $p^n\in J$. Then $p^n\in \mathbf{p}_i$ for all $i \in [\ell]$, and hence $p\in \mathbf{p}_i$ for all $i  \in [\ell]$. Therefore, $p\in J$.
\end{proof}

The following fact about the primary decomposition of radical ideals is true over any field, as a consequence of the Lasker-Noether theorems 
\cite[pp. 204-209]{cox-little-oshea}.

\begin{lemma}\label{lemma:rad}
If $J$ is a proper radical ideal, then it has a unique irredundant primary decomposition consisting of the minimal prime ideals over $J$.
\end{lemma}

\section{Decomposing the neural code via intervals of the Boolean lattice}\label{sec:boolean-lattice}

From the definition of $\A_\C$, it is easy to see that the maximal elements yield a kind of ``primary'' decomposition
of the neural code $\C$ as a union of maximal $V_a$s.

\begin{lemma} \label{lemma:code-decomp}
$\C = \bigcup_{i=1}^\ell V_{a_i}$, where $a_1,\ldots,a_\ell$ are the maximal elements of $\A_\C$.  (I.e., $\p_{a_1},\ldots,\p_{a_\ell}$
are the minimal primes in the primary decomposition of $J_\C$.)
\end{lemma}

\begin{proof}
Since $V_a \subseteq \C$ for any $a \in \A_\C$, clearly $\bigcup_{i=1}^\ell V_{a_i} \subseteq \C$.
To see the reverse inclusion, note that for any $c \in \C$, $c \in V_c \subseteq V_a$ for some maximal $a \in \A_\C$.
Hence, $\C \subseteq \bigcup_{i=1}^\ell V_{a_i}.$
\end{proof}

Note that Lemma~\ref{lemma:code-decomp}
could also be regarded as a corollary of Theorem~\ref{thm:prim-decomp}, since $\C = V(J_\C) = V(\bigcap_{i=1}^\ell \p_{a_i}) = \bigcup_{i=1}^\ell V(\p_{a_i}) = \bigcup_{i=1}^\ell V_{a_i}$, and the maximal $a \in \A_\C$ correspond to minimal primes $\p_a \supseteq J_\C$.
Although we were able to prove Lemma~\ref{lemma:code-decomp} directly, in practice we use the primary decomposition in order to find (algorithmically) the maximal elements $a_1,\ldots,a_\ell \in \A_\C$, and thus determine the $V_a$s for the above decomposition of the code.

It is worth noting here that the decomposition of $\C$ in Lemma~\ref{lemma:code-decomp} is not necessarily \textit{minimal}.  This is because one can have fewer $\q_a$s such that
$$\bigcap_{i \in \sigma \subsetneq [\ell]} \q_{a_i} = \bigcap_{i \in [\ell]} \p_{a_i}.$$
Since $V(\q_{a_i}) = V(\p_{a_i}) = V_{a_i}$, this would lead to a decomposition of $\C$ as a union of fewer $V_{a_i}$s.  In contrast, the primary decomposition of $J_\C$ in Theorem~\ref{thm:prim-decomp} is irredundant, and hence
none of the minimal primes can be dropped from the intersection.

\subsection*{Neural activity ``motifs'' and intervals of the Boolean lattice}

We can think of an element $a \in \{0,1,*\}^n$ as a neural activity ``motif''.  That is, $a$ is a pattern of activity and silence for a subset of the neurons,
while $V_a$ consists of all activity patterns on the full population of neurons that are consistent with this motif (irrespective of what the code is).
For a given neural code $\C$, the set of maximal $a_1,\ldots,a_l \in \A_\C$ corresponds to a set of minimal motifs that define the code (here ``minimal'' is used in the sense of having the fewest number of neurons that are constrained to be ``on'' or ``off'' because $a_i \neq *$).
If $a \in \{0,*\}^n$, we refer to $a$ as a neural {\em silence} motif, since it corresponds to a pattern of silence.  In particular, silence motifs correspond to simplices in $\supp \C$,
since $\supp V_a$ is a simplex in this case.  If $\supp \C$ is a simplicial complex, then Lemma~\ref{lemma:code-decomp} gives the decomposition of $\C$ as a union of minimal silence motifs (corresponding to \textit{facets}, or maximal simplices, of $\supp \C$).

More generally, $V_a$ corresponds to an {\em interval} of the Boolean lattice $\{0,1\}^n$.  Recall the poset structure of the Boolean lattice: for any pair of elements $v_1,v_2 \in \{0,1\}^n$, we have $v_1 \leq v_2$ if and only if $\supp(v_1) \subseteq \supp(v_2)$.  An \textit{interval} of the Boolean lattice is thus a subset of the form:
$$[u_1,u_2] \od \{ v \in \{0,1\}^n \mid u_1 \leq v \leq u_2 \}.$$
Given an element $a \in \{0,1,*\}^n$, we have a natural interval consisting of all Boolean lattice elements ``compatible'' with $a$.  Letting $a^0 \in \{0,1\}^n$ be the element
obtained from $a$ by setting all $*$s to $0$, and $a^1 \in \{0,1\}^n$ the element obtained by setting all $*$s to $1$, we find that
$$V_a = [a^0,a^1] = \{ v \in \{0,1\}^n \mid a^0 \leq v \leq a^1 \}.$$
Simplices correspond to intervals of the
form $[0,a^1]$, where $0$ is the bottom ``all-zeros'' element in the Boolean lattice.

While the primary decomposition of $J_\C$ allows a neural code $\C \subseteq \{0,1\}^n$ to be decomposed as a union of intervals of the Boolean lattice,
as indicated by Lemma~\ref{lemma:code-decomp}, the canonical form $CF(J_\C)$ provides a decomposition of the \textit{complement} of $\C$ as a union of intervals.
First, notice that to any pseudo-monomial $f \in CF(J_\C)$ we can associate an element $b \in \{0,1,*\}$ as follows: $b_i = 1$ if $x_i | f$, $b_i = 0$ if $(1-x_i) | f$, and $b_i = *$ otherwise.  In other words,
$$f = f_b \od \prod_{\{i \mid b_i = 1\}} x_i  \prod_{\{j \mid b_j = 0\}} (1-x_j).$$
As before, $b$ corresponds to an interval $V_b = [b^0,b^1]\subset \{0,1\}^n$.  Recalling the $J_\C$ is generated by pseudo-monomials corresponding to non-codewords,
it is now easy to see that the complement of $\C$ in $\{0,1\}^n$ can be expressed as the union of $V_b$s, where each $b$ corresponds to a pseudo-monomial in the canonical form.   The canonical form thus provides an alternative description of the code, nicely complementing Lemma~\ref{lemma:code-decomp}.

\begin{lemma}
$\C = \{0,1\}^n \setminus \bigcup_{i=1}^k V_{b_i}$, where $CF(J_\C) = \{f_{b_1},\ldots,f_{b_k}\}$.
\end{lemma}

\begin{wrapfigure}{r}{.5\linewidth}
\vspace{-.35in}
   \centering
   \includegraphics[width=1.8in]{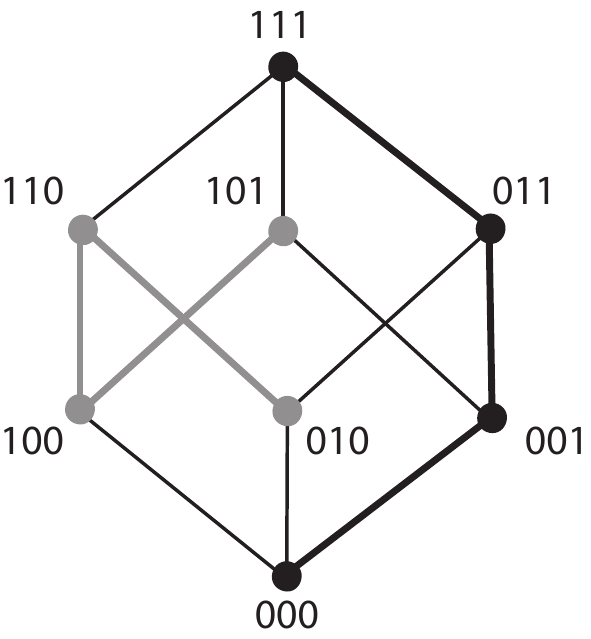} 
   \vspace{-.1in}
   \caption{\small Boolean interval decompositions of the code $\C = \{000, 001, 011, 111\}$ (in black) and of its complement (in gray), arising
   from the primary decomposition and canonical form of $J_\C$, respectively.}
\end{wrapfigure}

\noindent We now illustrate both decompositions of the neural code with an example.
\medskip

\noindent \textbf{Example.} Consider the neural code $\C = \{000, 001, 011, 111\} \subset \{0,1\}^3$ corresponding to a set of receptive fields satisfying $U_1 \subsetneq U_2 \subsetneq U_3 \subsetneq X$.  The primary decomposition of $J_\C \subset \F_2[x_1,x_2,x_3]$ is given by 
$$\langle x_1,x_2 \rangle \cap \langle x_1,1-x_3 \rangle \cap \langle 1-x_2, 1-x_3 \rangle,$$
while the canonical form is
$$CF(J_\C) = \langle x_1(1-x_2), x_2(1-x_3), x_1(1-x_3) \rangle.$$

From the primary decomposition, we can write $\C = V_{a_1} \cup V_{a_2} \cup V_{a_3}$ for $a_1 = 00*$, $a_2 = 0{*}1$, and $a_3 = *11$.
The corresponding Boolean lattice intervals are $[000,001]$, $[001,011]$, and $[011,111]$, respectively, and are depicted in black in Figure 5.
As noted before, this decomposition of the neural code need not be minimal; indeed, we could also write $\C = V_{a_1} \cup V_{a_3}$, as the 
middle interval is not necessary to cover all codewords in $\C$.

From the canonical form, we obtain $\C = \{0,1\}^3 \setminus (V_{b_1} \cup V_{b_2} \cup V_{b_3})$, where $b_1 = 10*$, $b_2 = *10$, and $b_3=1{*}0.$
The corresponding Boolean lattice intervals spanning the complement of $\C$ are $[100,101]$, $[010,110]$, and $[100,110]$, respectively; these are 
depicted in gray in Figure 5.  Again, notice that this decomposition is not minimal -- namely, $V_{b_3} = [100,110]$ could be dropped.

\section{An algorithm for primary decomposition of pseudo-monomial ideals}\label{sec:prim-decomp-algorithm}

We have already seen that computing the primary decomposition of the neural ideal $J_\C$ is a critical step towards extracting the canonical form $CF(J_\C)$, and that it also yields a meaningful decomposition of $\C$ in terms of neural activity motifs.
Recall from  Section~\ref{sec:canonical-form} that $J_\C$ is always a \textit{pseudo-monomial ideal} -- i.e., $J_\C$ is generated by pseudo-monomials, which are polynomials $f \in \F_2[x_1,\ldots,x_n]$ of the form
$$f = \prod_{i \in \sigma} z_i, \;\; \text{where} \;\; z_i \in \{x_i, 1-x_i\} \;\; \text{for any} \;\; i \in [n].$$
In this section, we provide an explicit algorithm for finding
the primary decomposition of such ideals.

In the case of \textit{monomial} ideals, there are many algorithms for obtaining the primary decomposition, and there are already fast implementations of such algorithms in algebraic geometry software packages such as Singular and Macaulay2 \cite{macaulay-book}.  
Pseudo-monomial ideals are closely related to square-free monomial ideals, but there are some differences which require a bit of care.  
In particular, if $J \subseteq F_2[x_1,\ldots,x_n]$ is a pseudo-monomial ideal  and $z \in \{x_i, 1-x_i\}$ for some $i \in [n]$, then
for $f$ a pseudo-monomial:
$$f \in \langle J, z \rangle \not\Rightarrow f \in J \text{ or } f \in \langle z \rangle.$$
To see why, observe that $x_1 \in \langle x_1(1-x_2), x_2 \rangle$, because
$x_1 = 1 \cdot x_1(1-x_2) + x_1 \cdot x_2,$ but $x_1$ is not a multiple of either $x_1(1-x_2)$ or $x_2$.
We can nevertheless adapt ideas from (square-free) monomial ideals to obtain an
algorithm for the primary decomposition of pseudo-monomial ideals.  The following lemma allows us to handle the above complication.

\begin{lemma}\label{new-lemma}
Let $J\subset \F_2[x_1,\ldots,x_n]$ be a pseudo-monomial ideal, and let $z \in \{x_i, 1-x_i\}$ for some $i \in [n]$.  For any pseudo-monomial $f$,
$$f \in \langle J,z \rangle \Rightarrow f \in J \text{ or } f \in \langle z \rangle \text{ or } (1-z)f \in J.$$
\end{lemma}

\begin{proof}[Proof of Lemma~\ref{new-lemma}]  Assume $f \in \langle J,z \rangle$ is a pseudo-monomial.  Then $f = z_{i_1}z_{i_2}\cdots z_{i_r}$, where $z_i \in \{x_i, 1-x_i\}$ for each $i$, and the $i_k$ are distinct.  Suppose $f \notin \langle z \rangle.$  This implies $z_{i_k} \neq z$ for all factors appearing in $f$.  We will show that either $f \in J$ or $(1-z)f \in J$.

Since $J$ is a pseudo-monomial ideal, we can write 
$$J = \langle z g_1, \ldots, z g_k, (1-z) f_1, \ldots, (1-z) f_l, h_1, \ldots, h_m \rangle,$$
where the $g_j, f_j$ and $h_j$ are pseudo-monomials that contain no $z$ or $1-z$ term.  This means
$$f = z_{i_1}z_{i_2}\cdots z_{i_r} = z \sum_{j=1}^k u_j g_j  + (1-z) \sum_{j=1}^l v_j f_j + \sum_{j=1}^m w_j h_j + y z,$$
for polynomials $u_j, v_j, w_j,$ and $y \in \F_2[x_1,\ldots,x_n]$. 
Now consider what happens if we set $z = 0$ in $f$:
$$f|_{z=0} =  z_{i_1}z_{i_2}\cdots z_{i_r}|_{z=0} = \sum_{j=1}^l v_j|_{z=0} f_j + \sum_{j=1}^m w_j|_{z=0} h_j.$$
Next, observe that after multiplying the above by $(1-z)$ we obtain an element of $J$:
$$(1-z) f|_{z=0} = (1-z) \sum_{j=1}^l v_j|_{z=0} f_j + (1-z) \sum_{j=1}^m w_j|_{z=0} h_j \in J,$$
since $(1-z)f_j \in J$ for $j = 1,\ldots,l$ and $h_j \in J$ for $j=1,\ldots,m$.
There are two cases: 
\begin{itemize}

\item[Case 1:] If $1-z$ is a factor of $f$, say $z_{i_1} = 1-z$, then $f|_{z=0} =  z_{i_2}\cdots z_{i_r}$ and thus
$f = (1-z) f|_{z=0} \in J.$

\item[Case 2:] If $1-z$ is {\em not} a factor of $f$, then $f = f|_{z=0}.$  Multiplying by $1-z$ we obtain 
$(1-z) f  \in J.$

\end{itemize}
We thus conclude that $f \notin \langle z \rangle$ implies $f \in J$ or $(1-z)f \in J$.
\end{proof}

Using Lemma~\ref{new-lemma} we can prove the following key lemma for our algorithm, which mimics the case of square-free monomial ideals.

\begin{lemma} \label{lemma:z-decomp}
Let $J \subset \F_2[x_1,\ldots,x_n]$ be a pseudo-monomial ideal, and let
$\prod_{i \in \sigma} z_i$ be a pseudo-monomial, with $z_i \in \{x_i, 1-x_i\}$ for each $i$.  Then,
$$\langle J, \prod_{i \in \sigma} z_i \rangle = \bigcap_{i \in \sigma} \langle J, z_i \rangle.$$
\end{lemma}

\begin{proof}[Proof of Lemma~\ref{lemma:z-decomp}]
Clearly, $\langle J, z_\sigma \rangle \subseteq \bigcap_{i \in \sigma} \langle J, z_i \rangle.$  To see the reverse inclusion, consider
$f \in \bigcap_{i \in \sigma} \langle J, z_i \rangle.$  We have three cases.
\begin{itemize}
\item[Case 1:] $f \in J$.  Then, $f \in \langle J, z_\sigma \rangle.$
\item[Case 2:] $f \notin J$, but $f \in \langle z_i \rangle$ for all $i \in \sigma$.  Then $f \in \langle z_\sigma \rangle$, and hence $f \in \langle J, z_\sigma \rangle.$
\item[Case 3:] $f \notin J$ and $f \notin \langle z_i \rangle$ for all $i \in \tau \subset \sigma$, but $f \in \langle z_j \rangle$ for all $j \in \sigma \setminus \tau$.  
Without loss of generality, we can rearrange indices
so that $\tau = \{1,\ldots,m\}$ for $m \geq 1$.  By Lemma~\ref{new-lemma}, we have $(1-z_i)f \in J$ for all $i \in \tau$.
We can thus write:
$$f = (1-z_1)f + z_1(1-z_2)f + \ldots + z_1\cdots z_{m-1}(1-z_m) f + z_1 \cdots z_m f.$$
Observe that the first $m$ terms are each in $J$.  On the other hand, $f \in  \langle z_j \rangle$ for each $j \in \sigma \setminus \tau$ implies that the last term is in 
$\langle z_\tau \rangle \cap \langle z_{\sigma \setminus \tau} \rangle = \langle z_\sigma \rangle.$  Hence, $f \in \langle J, z_\sigma \rangle.$
\end{itemize}
We may thus conclude that $\bigcap_{i \in \sigma} \langle J, z_i \rangle \subseteq \langle J, z_\sigma \rangle$, as desired.
\end{proof}

  Note that if $\prod_{i \in \sigma} z_i \in J$, then this lemma implies  $J =  \bigcap_{i \in \sigma} \langle J, z_i \rangle,$
which is the key fact we will use in our algorithm.
This is similar to Lemma 2.1 in \cite[Monomial Ideals Chapter]{macaulay-book},
and suggests a recursive algorithm along similar lines to those that exist for monomial ideals.

The following observation will add considerable efficiency to our algorithm for pseudo-monomial ideals.

\begin{lemma}\label{lemma:x_i-reduction}
Let $J \subset \F_2[x_1,\ldots,x_n]$ be a pseudo-monomial ideal.  For any $z_i \in \{x_i,1-x_i\}$ we can write
$$J = \langle z_i g_1, \ldots ,z_i g_k, (1-z_i)f_1,\ldots,(1-z_i)f_\ell,h_1,\ldots,h_m\rangle,$$
where the $g_j$, $f_j$ and $h_j$ are pseudo-monomials that contain no $z_i$ or $1-z_i$ term.  (Note that
$k, \ell$ or $m$ may be zero if there are no generators of the corresponding type.)
Then, 
\begin{eqnarray*}
\langle J, z_i \rangle = \langle J|_{z_i=0}, z_i \rangle &=& \langle z_i, f_1,\ldots,f_\ell,h_1,\ldots,h_m \rangle.
\end{eqnarray*}
\end{lemma}

\begin{proof}
Clearly, the addition of $z_i$ in $\langle J,z_i \rangle$ renders the $z_ig_j$ generators unnecessary.
The $(1-z_i)f_j$ generators can be reduced to just $f_j$ because $f_j = 1\cdot(1-z_i)f_j + f_j \cdot z_i$.
\end{proof}

\noindent We can now state our algorithm.   Recall that an ideal $I \subseteq R$ is {\em proper} if $I \neq R$.

\subsubsection*{Algorithm for primary decomposition of pseudo-monomial ideals}

\noindent \textbf{Input}: A proper pseudo-monomial ideal $J \subset \F_2[x_1,\ldots,x_n]$.  This is presented as $J = \langle g_1,\ldots,g_r\rangle$ with each generator $g_i$ a pseudo-monomial.\\

\noindent \textbf{Output}: Primary decomposition of $J$.  This is returned as a set $\PP$ of prime ideals, with $J = \bigcap_{I \in \PP} I$.

\begin{itemize}
\item Step 1 (Initializion Step): Set $\PP = \emptyset$ and $D = \{J\}.$  Eliminate from the list of generators of $J$ those that are multiples of other generators.

\item Step 2 (Splitting Step): For each ideal $I \in D$ compute $D_I$ as follows.   
\begin{itemize}
\item[Step 2.1:] Choose a nonlinear generator $z_{i_1}\cdots z_{i_m} \in I$, 
where each $z_i \in \{x_i, 1-x_i\}$, and $m \geq 2$.  
(Note: the generators of $I$ should always be pseudo-monomials.)
\item[Step 2.2:] Set $D_I = \{ \langle I, z_{i_1} \rangle, \ldots, \langle I, z_{i_m} \rangle\}.$
By Lemma~\ref{lemma:z-decomp} we know that
$$I = \bigcap_{k=1}^m \langle I, z_{i_k} \rangle = \bigcap_{K \in D_I} K.$$
\end{itemize}

\item Step 3 (Reduction Step):  For each $D_I$ and each ideal  $\langle I, z_i\rangle \in D_I$, reduce the set of generators as follows.
\begin{itemize}
\item[Step 3.1:] Set $z_i = 0$ in each generator of $I$.  This yields a ``0'' for each multiple of $z_i$, and removes $1-z_i$ factors in each of the remaining generators.  
By Lemma~\ref{lemma:x_i-reduction}, $\langle I,z_i \rangle = \langle I |_{z_i = 0}, z_i \rangle$. 
\item[Step 3.2:] Eliminate $0$s and generators that are multiples of other generators.
\item[Step 3.3:] If there is a $``1"$ as a generator, eliminate $\langle I, z_i \rangle$ from $D_I$ as it is not a proper ideal.
\end{itemize}
 
 \item Step 4 (Update Step): Update $D$ and $\PP$, as follows.
 \begin{itemize}
 \item[Step 4.1:] Set $D = \bigcup D_I$, and remove redundant ideals in $D$.  That is, remove an ideal if it has the same set of generators as another ideal in $D$.
 \item [Step 4.2:] For each ideal $I \in D$, if $I$ has only linear generators (and is thus prime), move $I$ to $\PP$ by setting  $\PP = \PP \cup I$ and $D = D \setminus I$.
\end{itemize}

 \item Step 5 (Recursion Step): Repeat Steps 2-4 until $D = \emptyset$.
 
 \item Step 6 (Final Step): Remove redundant ideals of $\PP$.   That is, remove ideals that are not necessary to preserve the equality $J = \bigcap_{I \in \PP} I$.
 \end{itemize}

\begin{proposition}\label{prop:prim-decomp}
This algorithm is guaranteed to terminate, and the final $\PP$ is a set of irredundant prime ideals such that $J = \bigcap_{I \in \PP} I$.
\end{proposition}

\begin{proof}
For any pseudo-monomial ideal $I \in D$, let $\deg(I)$ be the sum of the degrees of all generating monomials of $I$.  
To see that the algorithm terminates, observe that for each ideal $\langle I, z_i \rangle \in D_I$,  $\deg(\langle I, z_i \rangle) < \deg(I)$ (this follows from
Lemma~\ref{lemma:x_i-reduction}).  The degrees of elements in $D$ thus steadily decrease with each recursive iteration, until they are removed as prime ideals
that are appended to $\PP$.  At the same time, the size of $D$ is strictly bounded at $|D| \leq 2^{n \choose 3}$, since there are only $n \choose 3$ pseudo-monomials in $\F_2[x_1,\ldots,x_n]$,
and thus at most $2^{n \choose 3}$ distinct pseudo-monomial ideals.

By construction, the final $\PP$ is an irredundant set of prime ideals.  Throughout the algorithm, however, it is always true that 
$J = \left(\bigcap_{I \in D} I\right) \cap \left(\bigcap_{I \in \PP} I \right)$.  Since the final $D = \emptyset$, the final $\PP$ satisfies $J = \bigcap_{I \in \PP} I$.
\end{proof}

\chapter{Neural Ring Homomorphisms and Maps between codes}

In the preceding chapters, we introduced the neural ring and neural ideal as algebraic objects associated to a neural code, and showed how to use them to extract information about the structure of receptive fields directly from the code.  In this chapter, we will examine homomorphisms between neural rings, and explore how they relate to maps between their corresponding neural codes.We find that for any pair of neural rings $R_\C, R_\D$, there is a natural bijection between ring homomorphisms $R_\D\rightarrow R_\C$ and maps between the codes $\C\rightarrow \D$.  Since any code map has a corresponding ring homomorphism, the existence of a ring homomorphism $R_\D\rightarrow R_\C$ doesn't guarantee any similarity in the structures of $\C$ and $\D$. Our ultimate goal is to arrive at a definition of \textit{neural ring homomorphism} which respects important structures in the codes and corresponds to ``nice" code maps. \\

In Section 7.1, we introduce the idea of maps between codes, and give some elementary examples.  In Section 7.2, we will show that ring homomorphisms of neural rings are in a natural bijection with code maps and show explicitly how to obtain one from the other.  In Section 7.3, we will work towards preserving the structure of the code by considering neural rings as modules, and showing how to relate ring homomorphisms with module homomorphisms. In Section 7.4, we define neural ring homomorphisms, a restricted class of homomorphisms which preserve the individuality of the neurons.  Our main result is Theorem \ref{thm:nrharecomps}, which describes the corresponding types of code maps. Finally, in Section 7.5, we examine the effect of the most basic code maps on the canonical form for the ideal $J_\C$.

A note on notation: We continue our common abuse of notation from the previous chapter, where we use $f$ or other polynomial notation (e.g. $x_i$) not only to denote a polynomial of $\F_2[x_1,...,x_n]$, but also for the equivalence class in $\F_2[x_1,...,x_n]/I_\C$ or $R[n]$ of which that polynomial is a representative, or even to represent the associated function $\C\rightarrow\{0,1\}$ which is given by evaluating the polynomial $f$ on $\C$. Wherever it appears that the choice of polynomial representative may affect the discussion or result, we have attempted to make clear that this choice of representative is not important and all relevant objects are well-defined.

\section{Elementary maps between codes}

A \textit{code map} is simply a function $q:\C\rightarrow \D$ which assigns to each element $c\in\C$ a unique image $q(c) \in \D$; this assignment need not be either injective or surjective.  Suppose $\C$ is a code on $n$ neurons, so $\C\subseteq \{0,1\}^n$.  We will denote the $i$th bit of the codeword $c$ as $c_i$, so we write $c=(c_1,c_2,...,c_n)$.  We now give some basic examples of code maps which can be used as building blocks to get any code map.

\begin{enumerate}

\item \textbf{Permutation of labels:} Two codes $\C$ and $\D$ on $n$ neurons which are identical up to relabeling of neurons are effectively the same.  To permute the labels of neurons, choose a permutation $\sigma\in \mathcal{S}_n$, and then define a code map $q:\C\rightarrow \D$ by $q(c) = d$, where $d=(c_{\sigma(1)},...,c_{\sigma(n)})$.  Here, $\D = q(\C)$.

\item \textbf{Dropping neuron:} We take one of our neurons and remove it entirely, via the projection map.  Suppose $\C$ is a code on $n$ neurons, and we wish to drop neuron $n$.  Define $q:\C\rightarrow \D$ by $q(c) = d$, where $d=(c_1,...,c_{n-1})$.  Here $\D=q(\C)$ is a code on $n-1$ neurons. 

\item \textbf{Adding a neuron}

The notion of {\em adding} a neuron is more complicated. It is not clear what adding a neuron means - are we adding a new neuron which is firing, or not firing, or some of each?  However, we can add a neuron if we make the map unambiguous.  We can easily make the new neuron a function of the original codeword.  That is, let $f\in \F_2[x_1,...,x_n]$.  Then define $q:\C\rightarrow \D$ by $q(c) = d$, where $d=(c_1,...,c_n, f(c))$.  Note that the same function $f$ defines the new neuron for all codewords.  Here $\D = q(\C)$ is a code on $n+1$ neurons.

\item \textbf{Adding a codeword:} Under this map, each codeword maps to itself, but unlike the previous examples, we have $q(\C)\subsetneq \D$.  Thus there are new codewords in the target code. That is, suppose $q:\C\hookrightarrow \D$ is the inclusion map, so that $\C\subsetneq \D$ and $q(c)=c$, but $\D$ contains codewords which are not in $\C$. These extra codewords have been ``added."

\end{enumerate}

\begin{proposition} All code maps can be written as compositions of these four elementary maps:
\begin{enumerate}
\item Permutation of labels
\item Dropping a neuron
\item Adding a neuron of the form $f(c)$
\item Adding new codewords
\end{enumerate}
\end{proposition}

\begin{proof} Let $\C$ be a code on $n$ neurons and $\D$ a code on $m$ neurons.  Suppose $q:\C\rightarrow \D$ is a code map. For $i=1,...,m$, define the function $f_i\in \F_2[x_1,...,x_n]$ such that $f_i(c) = [q(c)]_i$ for all $c\in \C$.  We can always do this, as any Boolean function $\{0,1\}^n\rightarrow \{0,1\}$ can be represented as a polynomial in $\F_2[x_1,...,x_n]$.  

First we define some intermediate codes: let $\C_0=\C$.  For $i=1,...,m$, let $\C_i = \{ (c_1,...,c_n,d_1,...,d_i) \mid c\in \C, d=q(c)\}\subset\{0,1\}^{n+i}$. For $j=1,...,n$, let $\C_{m+j} = \{ (d_1,...,d_m,c_1,...,c_{n-j+1}) \mid c\in \C, d=q(c)\}\subset\{0,1\}^{m+n-j+1}$. Finally, define $\C_{m+n+1} = q(\C)\subset \D$.

Now, for $i=1,...,m$, let the code map $q_i:\C_{i-1}\rightarrow \C_i$ be defined by $q_i(v) = (v_1,...,v_{n+i-1}, f_i(v))\in \C_i$.   Note that if $v=(c_1,...,c_n, d_1,...,d_{i-1})$, then $f_i(v)= f_i(c)$, as only the first $n$ places matter.  Thus, if $v=(c_1,...,c_n,d_1,...,d_{i-1})$ with $d=q(c)$, then $q_i(v) = (c_1,...,c_n, d_1,...,d_i)$. Neuron by neuron, we add the digits of $q(c)$ on to $c$.  

Then, take the permutation map given by $\sigma = (n+1,...,n+m,1,...,n)$, so all the newly added neurons are at the beginning and all the originals are at the end.  That is,  define $q_\sigma:\C_m\rightarrow \C_{m+1}$  so if $v=(v_1,...,v_{n+m}),$ then $q_\sigma(v) = (v_{n+1},...,v_{n+m},v_1,...,v_n)$. 

We then drop the neurons $m+1$ through $n+m$ one by one in $n$ code maps.   That is, for $j=1,...,n$ define $q_{m+j}:\C_{m+j}\rightarrow \C_{m+j+1}$ by $q_{m+j}(v) = (v_1,...,v_{m+n-j})$.

Lastly, if $q(\C)\subsetneq \D$, then add one last inclusion code map $q_a:q(\C)\hookrightarrow \D$ to add the remaining codewords of $\D$.

Thus, given $c= (c_1,...,c_n)$ with $q(c) = d =(d_1,...,d_m)$, the first $m$ steps give us $q_m\circ\cdots\circ q_1(c) =  (c_1,...,c_n,d_1,...,d_m) = x$. The permutation then gives us $q_\sigma(x) = (d_1,...,d_m,c_1,...,c_n) = y$, and then we compose $q_{m+n}\circ\cdots\circ q_{m+1}(y) = (d_1,...,d_n) = d = q(c)$.   Finally, if $q(\C)\subsetneq \D$, we do our inclusion map, but as $q_a(d) = d$, the overall composition is a map $\C\rightarrow \D$ takes $c$ to $q(c)=d$ as desired.

\end{proof}

Here are some examples of other interesting code maps that we can build from these basic maps.  See Figure 7.1 for an example of some of these maps, along with the basic four listed above.

\begin{itemize}

\item \textbf{The reverse map:} Let $\C$ be a code on $n$ neurons, and define $q:\C\rightarrow \D$ by $q(c_1,...,c_n) = (1-c_1,...,1-c_n)$.  That is, change all 0s to 1s, and 1s to 0s. Here, $\D$ is the code on $n$ neurons defined by the range of $q$. We can build this as a composition by adding $n$ new neurons, where the $i$th neuron is given by the function $f(x) = 1-x_i$. We then permute so these new neurons are at the beginning, and then drop the $n$ original neurons from the end.

\item \textbf{The parity map:} We can add a new neuron which ensures that each codeword has even parity.  Given a code $\C$, define $q(c) = (c_1,...,c_n,\sum_{i=1}^n c_i)$.  Here, we are adding one new neuron, given by the function $f(x) = \sum_{i=1}^n x_i$.  Here, $\D = q(\C)$.

\item \textbf{Repetition of a neuron:} One natural way to add a neuron in a is to add a new neuron which copies one of the original neurons. For example, suppose $\C$ is a code on $n$ neurons, and we choose to repeat neuron $i$.  Define a code map $q:\C\rightarrow \D$ by $q(c) = d$ where $d=(c_1,...,c_n, c_i)$; the new neuron is given by the function $f(x) = x_i$.  Here, $\D=q(\C)$ is a code on $n+1$ neurons.  All the original neurons are kept exactly the way they are and continue to interact in the same ways, so we consider this to preserve neuron structure.

\item \textbf{Adding trivial neurons} Suppose we add a new neuron which is never firing, (always 0), or always firing (always 1).  That is, let $\C$ be a code on $n$ neurons, and define $q:\C\rightarrow \D$ by $q(c)=(c_1,...,c_n,0)$ (or $q(c) = (c_1,...,c_n,1)$ respectively). The function which adds the new neuron is given by $f(x) = 1$ (or $f(x) = 0$).  Here $\D=q(\C)$ is a code on $n+1$ neurons.  

\item \textbf{Merging two neurons:} In cases where data has been sorted incorrectly, so what appeared to be two neurons is actually one, we wish to combine those neurons. Under this map, we take two neurons and merge them into one which fires exactly when one or the other (or both) of the original two neurons fired.  For example, suppose $\C$ is a code on $n$ neurons and we wish to merge the last two neurons, $n-1$ and $n$.  Then define $q:\C\rightarrow \D$ by $q(c) = d$, where $d=(c_1,...,c_{n-2}, c_{n-1}+c_n+c_{n-1}c_n)$, so this last neuron is 1 if and only if $c_{n-1}=1$ or $c_n=1$ or both.  Here, $f(x) = x_{n-1}+x_{n} + x_{n-1}x_n$, and $\D = q(\C)$.  

\end{itemize}

Many of these maps coincide with common operations from coding theory, though the vocabulary and motivation are slightly different. The act of dropping a neuron is equivalent to the coding theory operation of puncturing a code. The parity map described above is much more natural in coding theory, but has little meaning for neural codes. The repetition map is used in coding theory to add repetitive bits which can reduce errors in decoding, but the method is very inefficient and thus is rarely used.  As a neural code, however we may see this map arising as a sorting error, when a single neuron's spikes are attributed erroneously to two neurons.

\begin{figure}[h]
\begin{center}
\includegraphics[width=5in]{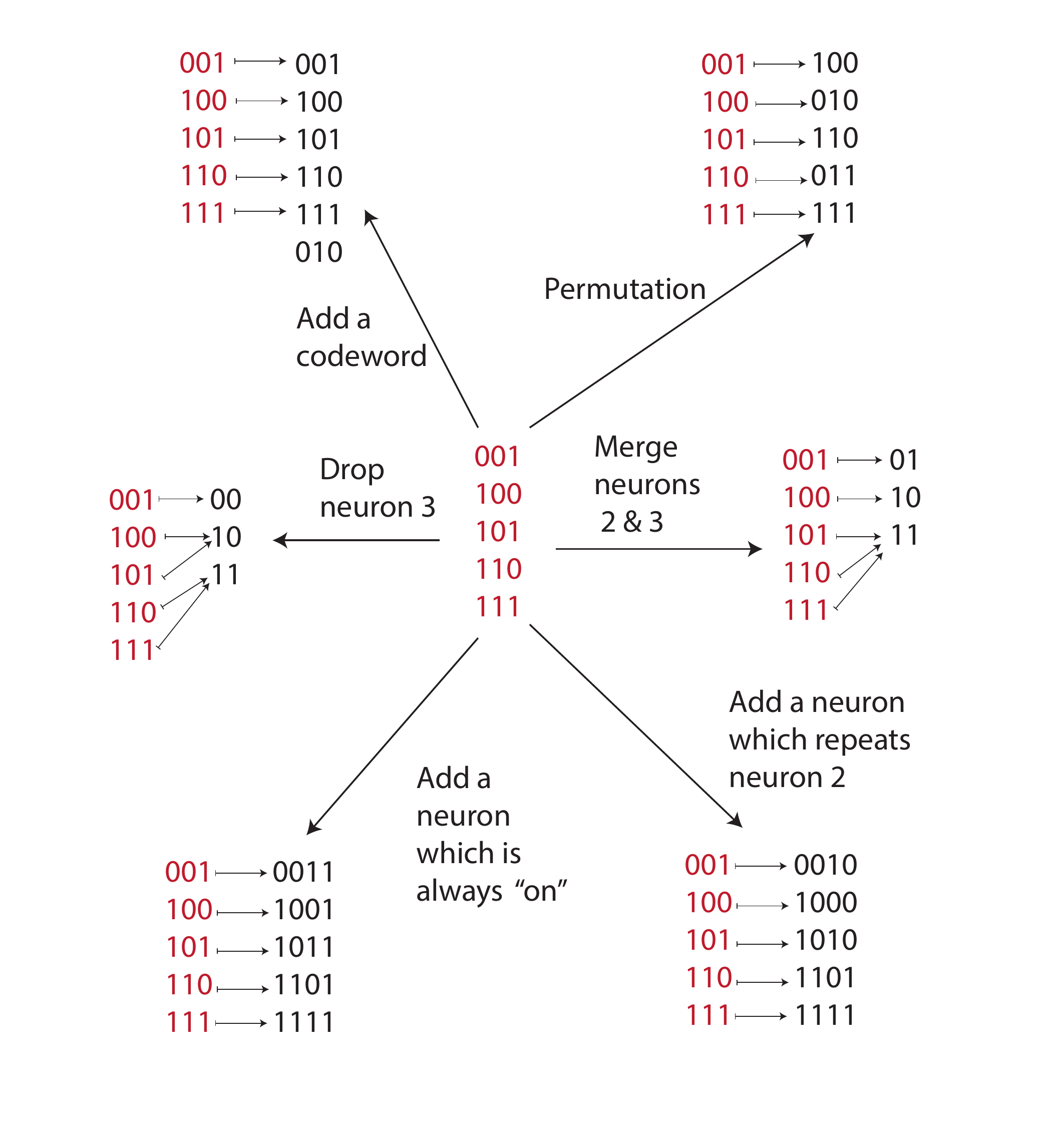}
\end{center}
\vspace{-.2in}
\caption{\small Several possible maps from the central code $\C = \{001,100,101,110,111 \}$, in red, to other related codes}
\end{figure}

\newpage

\section{Ring Homomorphisms between Neural Rings}
The most obvious type of map between neural rings is a ring homomorphism.  

\begin{definition} Let $R,$ be rings. A \textit{ring homomorphism} $\phi:R\rightarrow S$ is an assignment $\phi(r)\in S$ for every $r\in R$ so that the following properties hold:
\begin{itemize}
\item $\phi(a+b) = \phi(a) + \phi(b)$,
\item $\phi(ab) = \phi(a)\phi(b)$,
\item $\phi(1_R) = 1_S$.
\end{itemize}
\end{definition}

In this section, we show that there is a natural correspondence between ring homomorphisms of neural rings and code maps between the related codes.

 Let $f\in R_\D$, and recall that $f$ can be thought of as a function $f:\D\rightarrow \{0,1\}$.  Given a code map $q:\C\rightarrow \D$, we can ``pull back" $f$ to a function $q^*f: \C\rightarrow \{0,1\}$ by defining the pullback $q^*f = f\circ q$. 
 \[
 \xymatrix{ \C \ar[dr]_{q^*f = f\circ q} \ar[r]^q & \D\ar[d]^f\\  & \{0,1\}}
 \]
 
 Note that $q^*f \in R_\C$, and we can thus define a natural map between neural rings, $\phi_q:R_\D\rightarrow R_\C$, which takes each $f\in R_\D$ to its pullback by $q:\C\rightarrow \D$, so that $\phi_q(f) = q^*f = f\circ q$.  This leads us to the question: is the map $\phi_q$ a ring homomorphism?  Conversely, is every ring homomorphism $\phi:R_\D\rightarrow R_\C$ of the form $\phi_q$ for some code map $q:\C\rightarrow \D$?
 
 We first show that $\phi_q$ is always a ring homomorphism.

\begin{lemma}\label{lemma:codesgivehoms} For any code map $q:\C\rightarrow \D$, the map $\phi_q:R_\D\rightarrow R_\C$, where $\phi_q(f) = q^*f,$ is a ring homomorphism. 
\end{lemma}

\begin{proof}

To prove $\phi_q$ is a ring homomorphism, we need to show that addition, multiplication, and multiplicative unit are preserved.  Throughout this proof we will use the fact, discussed in the introduction, that two functions $f,g\in  R_\C$ are equivalent if and only if $f(c) = g(c)$ for all $c\in \C$.\\

 To show that $\phi_q(f+g) = \phi_q(f)+\phi_q(g)$, observe that for all $c\in \C$, $$\phi_q(f+g)(c) = q^*(f+g)(c) = (f+g)(q(c)) = f(q(c))+ g(q(c))$$ $$ = q^*f(c) + q^*g(c) = \phi_q(f)(c) )+ \phi_q(g)(c).$$ Since this is true for every $c\in \C$, we have $\phi_q(f+g) = \phi_q(f) + \phi_q(g)$.

 To show that $\phi_q(fg) = \phi_q(f)\phi_q(g)$, observe that for all $c\in \C$, $$\phi_q(fg)(c) = q^*(fg)(c)) = (fg)(q(c)) = f(q(c))g(q(c)) = (q^*f)(c)(q^*g)(c) = \phi_q(f)(c) \phi_q(g(c).$$  Since this is true for every $c\in \C$, we have $\phi_q(fg) = \phi_q(f)\phi_q(g)$.  

 Lastly, we must show $\phi_q(1_\D) = 1_\C$.  Again, observe that for all $c\in \C$, $$\phi_q(1_\D)(c) = (q^*1_\D)(c) = 1_\D(q(c)) = 1 = 1_\C(c).$$
Since this holds for all $c\in \C$, we get $\phi_q(1_\D) = 1_\C$.

Thus, $\phi_q$ is a ring homomorphism.
\end{proof}

\begin{example} 

Consider again the codes $\C = \{000,100,101\}$ and $\D = \{00,10,11\}$. If we look at the projection code map $q$ given by dropping the third neuron (so $000\mapsto 00, 100\mapsto 10$, and $101 \mapsto 10$) then the corresponding neural ring homomorphism $\phi_q:R_\D\rightarrow R_\C$ is given by $\phi(x_1) = x_1$ and $\phi(x_2)=x_2$, extending by linearity to all other elements.
\end{example}

It turns out that all ring homomorphisms of neural rings $\phi:R_\D\rightarrow R_\C$ are in fact of the form $\phi_q$, where $q:\C\rightarrow \D$ is a code map.  To prove this result, we introduce a useful basis.   

As we have seen, the neural ring $R_\C$ is equivalent to the ring of functions $f:\C\rightarrow \{0,1\}$.  Therefore, any element $f$ of $R_\C$ is completely determined by $f^{-1}(1)$, the set of codewords which $f$ ``detects."  If $f^{-1}(1) = \emptyset$,  then $f \equiv 0$.  As a natural basis for the neural ring $R_\C$, we therefore take the set of functions which detect a unique codeword of $\C$; that is, the set $\{f  \mid f^{-1}(1) = \{c\} \text{ for some } c\in \C\}$.  For each $c\in \C$, we denote the function which detects only $c$ by $\rho_c$; thus, our basis is exactly the set $\{\rho_c \mid c\in \C\}$. We write $\rho_c$ in polynomial notation as $$\rho_c = \prod_{c_i = 1} x_i \prod_{c_j = 0} (1-x_j).$$ With this polynomial representation, it's easy to see that $\rho_c$ acts as a characteristic function for $c$: that is, $\rho_c(v) = 1$ if and only if $v=c$.  We previously saw these elements in Chapter 3.

For any element $f\in R_\C$, we can write $f$ uniquely as a sum of these basis elements: 
\begin{eqnarray}
f &=\displaystyle{ \sum_{f(c) = 1} \rho_c}\,.
\end{eqnarray}

As these basis elements are characteristic functions in $\F_2$, we get the following properties immediately:
\begin{enumerate}
\item For any $c\in \C$, $\rho_c + \rho_c  = 0$.
\item For any $c\in\C$, $\rho_c\rho_c = \rho_c$.
\item For any $c,c'\in \C$ with $c\neq c'$, $\rho_c\rho_{c'} = 0$.
\item Combining properties 2 and 3, we get $\rho_c f = \rho_c$ if $f(c) = 1$, and $\rho_c f = 0$ if $f(c) =0$.
\end{enumerate}

 Under this notation, $0_\C$ will always be the empty sum, and $1_\C$ will always be the sum of all basis elements.  Note that while each $f\in R_\C$ may have many possible representations as a polynomial in $x_i$, the basis notation (1) is unique, and so we frequently prefer this notation in our proofs.

\begin{example}

Consider the code $\C = \{000,100,101\}$.  A basis for $R_\C$ is found by taking $\rho_{000} = (1-x_1)(1-x_2)(1-x_3), \rho_{100} = x_1(1-x_2)(1-x_3)$, and $\rho_{101} = x_1x_3(1-x_2)$.  Elements of $R_\C$ are given by taking all possible $\F_2$-combinations of these basis elements.  So, for example, we can take $\rho_{000}+\rho_{100} = (1-x_2)(1-x_3)$.  This function will evaluate to 1 on $000$ and $100$, but not on $101$.  We could also take the element $\rho_{000}+\rho_{100}+\rho_{001} = 1_{R_\C}$, which evaluates to $1$ on any of the three codewords.

 These choices of polynomial representatives are by no means unique, due to the relationships among variables which are particular to each ring. In this ring $R_\C$, we have $x_2=0$ and $x_1x_3 = x_3$.  Polynomial representations of $\rho_c$ thus include 
\begin{itemize}
\item  $\rho_{000} = (1-x_1)(1-x_2)(1-x_3)=(1-x_1)(1-x_3) = 1-x_1$, 
\item $\rho_{100} = x_1(1-x_2)(1-x_3) = x_1(1-x_3) = x_1-x_3$,
\item  $\rho_{101} = x_1x_3(1-x_2) = x_3$. 
\end{itemize}

\end{example}

Using this basis notation, we prove an important property of ring homomorphisms between neural rings.

\begin{lemma}\label{lemma:uniqueimage} Let $\phi:R_\D\rightarrow R_\C$ be a ring homomorphism between two neural rings.  Then for all $c\in \C$, there exists a unique $d\in \D$ such that $\phi(\rho_d) \rho_c = \rho_c$.
\end{lemma}

\begin{proof}
To show existence, we observe the following: for each $c\in \C$,

$$\rho_c =  1_\C \rho_c= \phi(1_\D)\rho_c = \phi\left(\sum_{d\in \D}\rho_d\right)\rho_c = \sum_{d\in \D} \phi(\rho_d) \rho_c $$

Using Property 4 above, this means that for at least one $d\in \D$, we have $\phi(\rho_d)\rho_c = \rho_c$.

To show uniqueness, suppose by way of contradiction that $\phi(\rho_d)\rho_c = \phi(\rho_{d'})\rho_c = \rho_c$ with $d\neq d'$. Then we have the following:
$$\rho_c = \phi(\rho_d)\rho_c\phi(\rho_{d'})\rho_c = \phi(\rho_d)\phi(\rho_{d'})\rho_c = \phi(\rho_d\rho_{d'}) \rho_c = 0_\C\rho_\C =0,$$
which is a contradiction.
\end{proof}

Lemma \ref{lemma:uniqueimage} essentially shows that the sets $C_d = \{c\in \C \mid \phi(\rho_d)\rho_c = \rho_c\}$ partition $\C$.  This allows us to define a code map associated to $\phi$ as follows: let $q_\phi:\C\rightarrow \D$ be given by $q_\phi(c) = d$, where $d$ is the unique element of $\D$ such that $\phi(\rho_d)\rho_c = \rho_c$.  The previous lemma shows that this map is well-defined.

\begin{definition} Let $\C$ and $\D$ be neural codes.  We define the following two sets:
\begin{itemize}
\item[-] $\mathrm{Hom}(R_\D, R_\C) \stackrel{\text{def}}{=} \{\phi:R_\D\rightarrow R_\C \mid \phi$ a ring homomorphism$\}$
\item[-] $\mathrm{Map}(\C, \D) \stackrel{\text{def}}{=} \{q:\C\rightarrow \D\mid q$ a function $\}$
\end{itemize}
\end{definition}

We now have a map in each direction between these two sets.  In Lemma \ref{lemma:codesgivehoms}, we showed how to find the a ring homomorphism $\phi_q$ from a code map using the pullback:
\begin{align*}
\mathrm{Map}(\C, \D)& \rightarrow \mathrm{Hom}(R_\D, R_\C)\\
q & \mapsto \phi_q
\end{align*}

And using Lemma \ref{lemma:uniqueimage}, we know we can find a code map $q_\phi$ from a ring homomorphism:
\begin{align*}
\mathrm{Hom}(R_\D, R_\C ) & \rightarrow \mathrm{Map}(\C,\D)\\
\phi&\mapsto q_\phi 
\end{align*}

\begin{theorem}\label{thm:codemaphombijection} Let $\C$ and $\D$ be neural codes, with $R_\C$ and $R_\D$ the respective associated neural rings.  Then the above defined maps between $\mathrm{Hom}(R_\D, R_\C)$ and $\mathrm{Map}(\C,\D)$ are inverses, and thus the sets are in bijection.  
\end{theorem}

\begin{proof}

We must show two things here: first, for any ring homomorphism $\phi:R_\D\rightarrow R_\C$, we have $\phi = \phi_{q_\phi}$; secondly, for any code map $q:\C\rightarrow \D$, we have $q=q_{\phi_q}$.\\

Key to both these proofs is the following fact: $\phi(\rho_d)\rho_c =\rho_c \Leftrightarrow \phi(\rho_d)(c) = 1$.
\begin{itemize}

\item $q=q_{\phi_q}$: Let $\phi = \phi_q$, so then $\phi(f) = q^*f$ for all $f\in R_\D$.  In particular, $\phi(\rho_d) = q^*\rho_d$, so $\phi(\rho_d)(c) = q^*\rho_d(c) = \rho_d(q(c)) = \left\{\begin{array}{ll} 1 & q(c) = d \\ 0 & q(c)\neq d\end{array}\right.$  

Thus, $\phi(\rho_d)\rho_c = \left\{\begin{array}{ll} \rho_c & q(c) = d\\ 0 & q(c)\neq d\end{array}\right.$ and so we define $q_\phi(c) = d \Leftrightarrow q(c)=d$; hence, $q_{\phi_q} = q$.\\

\item $\phi = \phi_{q_\phi}$: Let $q=q_\phi$, i.e., $q(c) = d$ for the unique $d$ with $\phi(\rho_d)\rho_c = \rho_c$. 

We must show $\phi_q(f) = \phi(f)$ for all $f\in R_\D$; it suffices to show $\phi_q(\rho_d) = \phi(\rho_d)$ for all $d\in \D$. These two functions are equal if they evaluate the same on all $c\in \C$; equivalently, they are the same if $\phi_q(\rho_d)\rho_c = \phi(\rho_d)\rho_c$ for all $c\in \C$.  By the definition of $\phi_q$, this means we must show that for all $c\in \C$,  $\phi(\rho_c) = q^*\rho_d\rho_c$.

To see this, observe that by the definition of $q=q_\phi$,
$$\phi(\rho_d)\rho_c = \left\{\begin{array}{ll} \rho_c & d=q(c)\\ 0 & d\neq q(c)\end{array}\right. .$$

On the other hand, $(q^*\rho_d)(c) = \rho_d(q(c))=\left\{\begin{array}{ll} 1 & d=q(c)\\ 0 & d\neq q(c)\end{array}\right.$ and therefore
$$q^*\rho_d\rho_c = \left\{\begin{array}{ll} \rho_c & d= q(c)\\ 0 & d\neq q(c)\end{array}\right. .$$
Thus, $\phi = \phi_{q_\phi}$.

\end{itemize}

\end{proof}

Theorem \ref{thm:codemaphombijection} gives us both good news and bad. On the positive side, we have discovered a very clean bijective correspondence between code maps and ring homomorphisms.  In particular, this theorem shows that not only does every code map $q$ induce a related ring homomorphism $\phi$ via pullbacks, but that {\em every} ring homomorphism between two neural rings can be obtained in this way; the relationship between the two notions is incredibly strong.On the other hand, the complete generality of this theorem iis not useful for our goal of selecting realistic maps which preserve neuron structure.  As we have just shown, any code map at all has a corresponding neural ring homomorphism. Even a random assignment $q:\C\rightarrow \D$ would have a related ring homomorphism. Another unsatisfying thing about this correspondence is that the notion of isomorphism captures very little actual similarity, but only the number of codewords, as the following lemma shows.

\begin{lemma}\label{lemma:iso} Two neural rings $R_\C$ and $R_\D$ are isomorphic if and only if $|C| = |D|$. 
\end{lemma}

\begin{proof}
As a neural ring $R_\C$ is exactly the ring of functions $f:\C\rightarrow \{0,1\}$, we have that $R_\C\cong \F_2^{|\C|}$ and $R_\D \cong \F_2^{|\D|}$, and then we use the fact that $\F_2^n \cong \F_2^m$ if and only if $n=m$.

\end{proof}

With these problems in mind, we consider another way to look at the neural rings which can preserve the structure given by the code: we consider them as modules.


\section{Module Homomorphisms between Neural Rings}

In this section, we will show that we can consider each neural ring $R_\C$ as a carefully designed module under the following ring: $$R[n] \stackrel{\text{def}}{=} R_{\{0,1\}^n} = \F_2[x_1,...,x_n]/\mathcal B$$ where $\C\subset\{0,1\}^n$. The module action, as we will show, preserves the structure of the original code $\C$.  Furthermore modules, like rings, have a well-defined notion of homomorphism.

\subsection{Neural rings as modules}

For a code $\C$ on $n$ neurons, we consider the neural ring $R_\C$ as an $R[n] = R_{\{0,1\}^n}$-module.  $R[n]$ will be referred to as the `ambient ring' when $n$ is understood.  Considering $R_\C$ as an $R[n]$-module  allows us to store the combinatorial structure of the code and retain the information about the presentation even if $R_\C$ is given only as an abstract ring. The module action is as follows:  given $r \in R[n]$ and $f \in R_\C$, define $$[r \cdot f] (c) = r(c) f (c). $$  That is, $(r \cdot f)^{-1}(1) = r^{-1}(1) \cap f^{-1}(1)$.   Note particularly, this intersection will be a subset of $\C$.  In other words, $f$ detects a certain set of codewords; $r\cdot f$ detects only those which are detected by $r$ as well. 

This module action is exactly multiplication of polynomials, with the result considered as an element of $R_\C$.  In particular, in any $R_\C$ we have the relationships $x_i(1-x_i) = 0$, which also means $x_i^2 = x_i$ and $(1-x_i)^2 = (1-x_i)$.   The use of these relationships can be seen more clearly in the following example:

\begin{example}

Consider again the code $\C = \{000,100,101\}$.  $R_\C$ is a module under $R[3] = R_{\{0,1\}^3}$.
\begin{itemize}
\item Consider the element $1-x_1$ of  $R[3]$.   Then  $$(1-x_1)\cdot (1-x_1)(1-x_2)(1-x_3) = (1-x_1)(1-x_2)(1-x_3) $$ whereas $$(1-x_1)\cdot x_1(1-x_2)(1-x_3) = 0.$$ 

\item Consider the element $x_1x_2$ of $R[3]$.  Although $x_1x_2$ is a nonzero element in $R[3]$, it evaluates to $0$ for all codewords in $\C$, so for any element $f\in R_\C$, we have $x_1x_2\cdot f = 0$.

\end{itemize}

As another way to look at this action, note that $R[3]$ is itself a neural ring, and so has basis elements $\{\rho_c\,| \, c\in \{0,1\}^3\}$.  We can look at the action in terms of these basis elements:  

\begin{itemize}

\item Consider $\rho_{000} + \rho_{001}$ in $R[3]$. Then $$(\rho_{000}+\rho_{001}) \cdot \rho_{000} = \rho_{000}$$
whereas $$(\rho_{000}+\rho_{001})\cdot \rho_{100} = 0.$$

\item Consider the element $\rho_{110}$ of $R[3]$.  Although this is a nonzero element in $R[3]$, it detects no codewords of $\C$, so it is equivalent to $0$ in $R_\C$.  Thus, for any element $f\in R_\C$, we have $\rho_{110}\cdot f = 0$.
\end{itemize}

\end{example} 

\noindent\textbf{Recovering the code}

The most powerful property of this action is the ability to recover the codewords purely from the module action.  To do so, we use the canonical generators $x_i \in R[n]$.  Note that $(x_i)^{-1}(1) = \{c\in \{0,1\}^n \mid c_i =1\}$.  Thus we see $x_i$ detects exactly those codewords in which we see neuron $i$ firing. Using these special elements, we can recover our code.  Here are the steps we use to recover a single codeword:

\begin{enumerate}
\item Select a basis element $\rho$.

\item For each $i=1,...,n$, consider $x_i \cdot \rho$.  We know $\rho$ detects exactly one codeword $c$, so we have two possibilities: if $c_i=1$, then $(x_i\rho^{-1})(1) = \{c\}$, and thus  $x_i\cdot \rho = \rho$; if $c_i=0$, then  $(x_i\rho)^{-1}(1) = \emptyset$ and thus $x_i\cdot\rho = 0$.    

\item Form the codeword $c$ by setting  $c_i = 1$ if $x_i \rho  = \rho$, and $c_i = 0$ if $x_i \rho = 0$.  
\end{enumerate}

Taking the set of codewords given by repeating these steps for every basis element $\rho$, we obtain the original code $\C$.

\begin{example}
Consider once again the code $\C = \{000,100, 101\}$ and the ring $R_\C$.
As an example of how to recover the codewords, take just one basis element, $\rho_{101}$.  Note that $x_1\cdot \rho_{101} = x_1\cdot x_1x_3(1-x_2) = x_1x_3(1-x_2)= \rho_{101}$, which tells us $c_1 = 1$.   Similarly, $x_2\cdot \rho_{101} = 0$, and $x_3\cdot \rho_{101} = \rho_{101}$.  So we know this basis element corresponds to the codeword $101$.
\end{example}

\subsection{Note on modules under different rings}

Now that we have a framework to consider a neural ring $R_\C$ on $n$ neurons as an $R[n]$-module that preserves the code structure, we can consider module homomorphisms between neural rings.  However, this is complicated by the fact that two neural rings $R_\C$ and $R_\D$ on $n$ and $m$ neurons respectively are considered modules under \textit{different} rings $R[n]$ and $R[m]$.  In order to consider $R$-module homomorphisms between $R_\C$ and $R_\D$ for some ring $R$, we need some way to think of both rings as modules under the \textit{same} ring $R$.  For this, we use the following standard construction from commutative algebra.

Suppose $R,S$ are rings with $\tau:R\rightarrow S$ a ring homomorphism.  Given an $S$-module $M$, we can also view $M$ as an $R$-module via the homomorphism $\tau$, using the action $r\cdot m = \tau(r)\cdot m$ for any $r\in R, m\in M$.  In the neural ring setup, this says that given a ring homomorphism $\tau:R[m]\rightarrow R[n]$, we can consider the $R[n]$-module $R_\C$ as an $R[m]$-module.

However, our module maps will be inspired by maps between the neural rings, rather than the overarching rings $R[n], R[m]$. Therefore,  we need a vocabulary for when the situation (unusual in commutative algebra) where one is first given a map between two modules under different rings $R$ and $S$, and wants to look for which ring homomorphisms between $R$ and $S$ (if any) would allow that map to be a module homomorphism.

\begin{definition} Given an $R$-module $M$, an $S$-module $N$, and a group homomorphism $\phi:M\rightarrow N$ (so $\phi(x+y) = \phi(x) + \phi(y)$), we say that a ring homomorphism $\tau:R\rightarrow S$ is  \textit{compatible} with $\phi$ if $\phi$ is an $R$-module homomorphism, where $N$ is viewed as an $R$-module via $\tau$.  In other words, for every $r\in R$, we want the following diagram to commute:  
\[
\xymatrix{ M\ar[d]^\phi \ar[r]^{r\cdot} & M \ar[d]^\phi  \ar[d]^\phi \\  N \ar[r]^{\tau(r)\cdot} & N }
\]

That is, $\tau$ is compatible with $\phi$ if $\phi(r \cdot x ) = \tau(r) \cdot \phi(x)$ for all $r\in R, x\in M$. 
\end{definition}

It is worth noting that not every group homomorphism between two neural rings has a compatible ring homomorphism.  

\begin{example} Consider the codes $\C = \{000,100,101\}$ and $\D = \{00,10,11\}$ , and let the map $\phi:R_\D\rightarrow R_\C$ be given by $\phi(\rho_{00}) = \rho_{000} + \rho_{100}$, $\phi(\rho_{10}) = \rho_{100}$, and $\phi(\rho_{11})= 0$.  Extending by linearity to all elements of $R_\D$ gives us a group homomorphism, which is easy to check.  

As a polynomial map, this is the group homomorphism given by: $x_1\rightarrow y_1(1-y_2)(1-y_3)$, $x_2\rightarrow 0$.  There is, however, no compatible ring homomorphism $\tau:R[2]\rightarrow R[3]$ so that $\phi$ is an $R[2]$-module homomorphism. To see this, note that any such homomorphism $\tau$ would need the following properties: 

$\phi(\rho_{00}\cdot \rho_{00}) = \tau(\rho_{00})\phi(\rho_{00}) = \tau(\rho_{00}) \cdot[\rho_{000} + \rho_{100}]$. 

But as $\rho_{00} \cdot \rho_{00} = \rho_{00}$, this must equal $\rho_{000} + \rho_{100}$.  So $\tau(\rho_{00})$ must preserve $\rho_{000}$ and $\rho_{100}$.  Similarly, $\tau(\rho_{10})$ must preserve $\rho_{100}$.  So $\tau(\rho_{10}) \tau(\rho_{00})$ must preserve $\rho_{100}$ at least, so $\tau(\rho_{10})\tau(\rho_{00})\neq 0$.
Note $\tau(\rho_{00}\rho_{10}) = \tau(0)=0$, but as $\tau$ is a ring homomorphism, we also have $\tau(\rho_{00}\rho_{10}) = \tau(\rho_{00})\tau(\rho_{10}) \neq 0$.  So no such $\tau$ can exist; there is no compatible ring homomorphism for $\phi$.

\end{example}

Luckily, one class of group homomorphisms between neural rings which are guaranteed to have a compatible $\tau$ are those which are also ring homomorphisms.

For this result, we will use the idea that elements of $R_\D$ can also be thought of as elements of the ambient ring $R[m]$.  For example, each basis element $\rho_d$ of $R_\D$ is the function which detects only the codeword $d$; since $d\in \{0,1\}^m$, we know $R[m]$ has a basis element $\rho_d$ which detects only $d$ as well, and we consider these two $\rho_d$ to be essentially the same.  Likewise, any function $f\in R_\D$ corresponds the subset  $f^{-1}(1)\subset \D$ which it detects, so we can consider $f$ as a function in $R[m]$ which detects the same set of codewords.

\begin{proposition} \label{prop:compatible}Suppose $\C$ and $\D$ are neural codes on $n$ and $m$ neurons, respectively.  If $\phi:R_\D\rightarrow R_\C$ is a ring homomorphism, then there exists a ring homomorphism $\tau:R[m]\rightarrow R[n]$ which is compatible with $\phi$, and thus $\phi$ is an $R[m]$-module homomorphism. Furthermore, the set of compatible $\tau$ is exactly the set of ring homomorphisms $R[m]\rightarrow R[n]$ which are extensions of $\phi$, in the sense for all $f\in R_\D$, $(\phi(f))^{-1}(1) \subseteq (\tau(f))^{-1}(1)$.  

\end{proposition}

\begin{proof} 
Let $\phi:R_\D\rightarrow R_\C$ be a ring homomorphism.  To construct a compatible ring homomorphism $\tau:R[m]\rightarrow R[n]$, first select one basis element  $\rho_d$ of $R_\D$.  Note that $\rho_d$ (as the function which detects exactly the codeword $\{d\}$) is also a basis element of $R[m]$, and define $\tau(\rho_d ) = \phi(\rho_d) + \sum_{v\in \{0,1\}^n \backslash \C}\rho_v$; that is, $\tau(\rho_d)$ will detect all the same codewords as $\phi(\rho_d)$, but also all the codewords of $\{0,1\}^n$ which are not part of $\C$.  For all other $d\in \D$, define $\tau(\rho_d) = \phi(\rho_d)$, and for all $v\in \{0,1\}^m\backslash \D$, define $\tau(\rho_v ) = 0$.  Extend $\tau$ to all elements of $R[m]$ by linearity; that is, if $f = \sum \rho_c$, then $\tau(f) = \sum \tau(\rho_c)$. This gives a ring homomorphism.

Now, we will show that the property of compatibility is equivalent to the property of extensions. Let $\tau:R[m]\rightarrow R[n]$ be a ring homomorphism, and let $f=\sum_{d\in f^{-1}(1)} \rho_d \in R_\D$.   As $f\cdot f = f$, then $\tau$ is compatible with $\phi$ if and only if we have  $\phi(f)=\phi(f\cdot f) = \tau(f)\cdot \phi(f)$, which occurs if and only if $\tau(f)$ detects at least the same codewords as $\phi(f)$, which happens if and only if  $\tau(f)^{-1}(1)\supseteq \phi(f)^{-1}(1)$.
\end{proof}

\subsection{Compatible $\tau$ extend code maps}

It is important to note that each map $\tau:R[n]\rightarrow R[m]$ is in fact also a ring homomorphism between neural rings - in this case, the neural ring for the complete code -  as $R[n] = R_{\{0,1\}^n}$.  We have shown that code maps correspond to ring homomorphisms, and thus each ring homomorphism $\tau$ corresponds to a unique code map $q_\tau:\{0,1\}^m \rightarrow \{0,1\}^n$ between the complete codes.  Furthermore, Proposition \ref{prop:compatible} shows that $\tau$ is compatible with $\phi$ if and only if it is an extension of $\phi$, in that $\phi(f)^{-1}(1)\subseteq \tau(f)^{-1}(1)$. Therefore, we have the following lemma:

\begin{lemma} \label{lemma:compmaps} $\tau$ is compatible with $\phi$ if and only if $q_\phi = q_\tau\big|_\C$.

\end{lemma}

\begin{proof} 
Suppose $\tau$ is compatible with $\phi$.  Note that $q_\tau\big|_\C =q_\phi$ if and only if $q_\tau(c) = q_\phi(c)$ for all $c\in \C$.  So, suppose by way of contradiction that $q_\tau(c) \neq q_\phi(c)$ for some $c\in \C$.  Let $d=q_\phi(c)$.  Then, $\phi(\rho_d)(c) =1$.  But  $\phi(\rho_d) = \phi(\rho_d\cdot \rho_d) =  \tau(\rho_d) \cdot \phi(\rho_d)$, and we know $\tau(\rho_d)(c) = 0$, not $1$, so $\tau$ and $\phi$ cannot be compatible. This is a contradiction.

Now, suppose $q_\phi = q_\tau\big|_\C$.  Suppose $c\in \phi(\rho_d)^{-1}(1)$.  By our code map-homomorphism correspondence, this means that $q_\phi(c) = d$.  So then $q_\tau(c) = d$ also, and thus again by the proven correspondence, $c\in \tau(\rho_d)^{-1}(1)$.  Thus, for each $d\in \D$, we have $\tau(\rho_d)^{-1}(1) \supseteq \phi(\rho_d)^{-1}(1)$, and by Proposition 6, $\tau$ is compatible with $\phi$.
\end{proof}

We know based on our earlier work that given a ring homomorphism $\phi:R_\D\rightarrow R_\C$, we can always find a ring homomorphism $\tau:R[m]\rightarrow R[n]$ which is compatible with $\phi$.  We now confirm that idea from the code maps side:  we can generate a possible $\tau$ by taking any code map $q':\{0,1\}^n\rightarrow \{0,1\}^m$ which extends $q$ (so $q'(c)=q(c)$ for all $c\in \C$) and taking the corresponding ring homomorphism $\tau_q$.  Then $\tau_q$ will take each function to its pullback by $\tau$, and it will be compatible with $\phi$.

\begin{example}
Consider again the codes $\C = \{000,100,101\}$ and $\D = \{00,10,11\}$. Let the  map $\phi:R_\D\rightarrow R_\C$ be given by $\phi(\rho_{00}) = \rho_{000}$, $\phi(\rho_{10} )= \rho_{100} + \rho_{101}$, and $\phi(\rho_{11}) = 0$, and extend by linearity to all other elements. Then $\phi$ is a ring homomorphism.

First, consider $\tau_1:R[2]\rightarrow R[3]$ given by $\tau_1(\rho_{00}) = \rho_{000} + \rho_{001}+\rho_{010}+\rho_{110} + \rho_{011}+ \rho_{111}, \tau_1(\rho_{10}) = \rho_{100} + \rho_{101}$, and $\tau_1(\rho_{11}) = 0$. Extend again by linearity.  Then $\tau_1$ is a compatible ring homomorphism, which is not hard to check.

Now, consider $\tau_2:R[2]\rightarrow R[3]$ given by $\tau_2(\rho_{00}) = \rho_{000} + \rho_{001}+\rho_{010}+\rho_{110} + \rho_{011}+ \rho_{111}, \tau_2(\rho_{10}) = \rho_{100}$, and $\tau_2(\rho_{11}) = \rho_{101}$.  $\tau_2$ is not a compatible ring homomorphism, as if it were, we would have $$\rho_{100} + \rho_{101} = \phi(\rho_{10})  = \phi(\rho_{10}\rho_{10}) = \tau_2(\rho_{10})\phi(\rho_{10}) =\rho_{100}\cdot(\rho_{100} + \rho_{101}) = \rho_{100}$$ which is a contradiction.

\end{example}


\section{Neural Ring Homomorphisms}

\subsection{Neuron-preserving homomorphisms}

We have now established that ring homomorphisms $\phi:R_\D\rightarrow R_\C$ between two neural rings are in correspondence with the set of possible functions $q:\C\rightarrow \D$.  But this is not entirely a satisfying definition for \textit{neural} ring homomorphism.  By looking at the neural rings as modules, we had hoped to preserve structure; this result makes it clear that additional restrictions are needed, since not every code map preserves structure, but every code map generates a related ring homomorphism and therefore a related module homomorphism.   So, using module properties we can extract code structure, but we cannot ensure that structure is preserved across maps.

This motivates us to consider preservation of neurons. We have seen that the activity of neuron $i$ is recovered by considering the action of the variable $x_i$. This allows us to figure out which basis elements correspond to which codewords. Under a general ring homomorphism, we place no unusual restriction on the images of these special elements of the ambient ring, and so we don't carry that structure over to the image ring.  What could be learned if we restricted to compatible maps where elements which detect neurons to map to other elements that detect neurons?    This motivates the following definition:

\begin{definition} Write $R[m] = \F_2[y_1,...,y_m]/\B$, and $R[n] = \F_2[x_1,..,x_n]/\B$.  A ring homomorphism $\tau:R[m]\rightarrow R[n]$ is called {\em neuron-preserving} if $\tau(y_i)\in \{x_1,...,x_n,0,1\}$ for all $i=1,...,m$.
\end{definition}

Not all ring homomorphisms are neuron-preserving, as shown in the following example:

\begin{example}
 Consider the map $\tau:R[1]\rightarrow R[2]$ given by $\tau(\rho_1) = \rho_{01} + \rho_{10}$ and $\tau(\rho_0) = \rho_{00} + \rho_{11}$.  Here $\tau(y_1) = x_1+x_2$, and $\tau(y_i)\notin \{x_1,x_2,0,1\}$ as would be required.

\end{example}

Observe that a neuron-preserving homomorphism $\tau:R[m]\rightarrow R[n]$  is defined by the vector $S=(s_1,...,s_m)$, where $S_i\in [n]\cup\{0,u\}$, so that $$\tau(y_i) = \left\{\begin{array}{ll} x_j & \text{if }s_i = j \\ 0 & \text{if }s_i = 0\\ 1 & \text{if }s_i=u\end{array}\right. .$$  $S$ is a vector which stores the pertinent information about $\tau$, and each possible $S$ with $s_i \in [n]\cup \{0,u\}$ defines a possible neuron-preserving $\tau$.  We refer to the $\tau$ defined by $S$ as $\tau_S$.

\begin{remark} We make a careful choice to define neuron-preserving as a property of maps between the Boolean rings only, and not between neural rings in general.  This is due in part to our notational conventions.  The polynomial representative $x_1$ means a different thing in $R[n]$ than it does in $R_\C$.  In particular, in a neural ring, we may have $x_i = x_j$ for $i\neq j$, whereas in the Boolean ring these are necessarily distinct.  This allows us to define a neuron-preserving homomorphism $\tau$ for any given $S$ without fear; we don't need to worry about relationships amongst the $x_i$ being preserved, since there are no relationships in $R[n]$ to speak of.  Thus, any choice of images $\tau(x_i)$ will give a ring homomorphism.
\end{remark}
 
 \begin{lemma}  The composition of two neuron-preserving homomorphisms is neuron-preserving. 
 
 \end{lemma}
 
 \begin{proof} Suppose $S=(s_1,...,s_n)$ with $s_i\in [m]\cup\{0,u\}$ and $T=(t_1,...,t_m)$ with $t_i \in [\ell]\cup\{0,u\}$ are given as above, with $\tau_{S}:R[n]\rightarrow R[m]$ and $\tau_{T}: R[m]\rightarrow R[\ell]$.  To prove the lemma, we need to find $W = (w_1,...,w_n)$ with $w_i\in [\ell]\cup\{0,u\}$ so $\tau_{W} = \tau_T\circ \tau_S:R[n]\rightarrow R[\ell]$.

Define the vector $W$ by $$w_i =\left\{\begin{array}{ll} t_{s_i} & \text{if }s_i \in [m]\\ 0 & \text{if }s_i = 0 \\ u &\text{if } s_i = u \end{array}\right.$$ 

Use variables $z_i$ for $R[n]$, $y_i$ for $R[m]$, $x_i$ for $R[\ell]$.  Then, unraveling the definitions, $$\tau_W(z_i) = \left\{\begin{array}{ll} x_j &\text{if } t_{s_i} = j \\ 0 &\text{if } t_{s_i} = 0 \\ & \text{or if }s_i = 0 \\ 1 & \text{if }t_{s_i} = u\\ & \text{or if } s_i = u\end{array}\right. 
= \left\{\begin{array}{ll} x_j &\text{if } s_i = k   \text{ and } t_k=j \\ 0 & \text{if }s_i = k \text{ and } t_k = 0\\ & \text{or if }s_i = 0 \\ 1 &\text{if } s_i = k \text{ and } t_k = u  \\ & \text{or if }s_i = u\end{array}\right. \hspace{1in}$$ $$
\hspace{1in} =  \left\{\begin{array}{ll} x_j &\text{if } \tau_S(z_i) = y_k   \text{ and } \tau_T(y_k ) = x_j \\ 0 &\text{if } \tau_S(z_i) = y_k \text{ and } \tau_T(y_k) = 0\\ &\text{or if } \tau_S(z_i) = 0 \\ 1 & \text{if }\tau_S(z_i) = y_k \text{ and } \tau_T(y_k) = 1  \\ & \text{or if }\tau_S(z_i) =1 \end{array}\right. =   \left\{\begin{array}{ll} x_j & \text{if }\tau_T\circ \tau_S(z_i) = x_j \\ 0 & \text{if }\tau_T\circ\tau_S(z_i) = 0 \\ 1 & \text{if }\tau_T\circ \tau_S(z_i) = 1 \end{array}\right. .$$

 \end{proof}

 \subsection{Neuron-preserving code maps}
 We now define what it means for a code map to be neuron-preserving, and relate the two notions.

  \begin{definition} Let $\C$ be a code on $n$ neurons and $\D$ a code on $m$ neurons.  A code map $q:\C\rightarrow \D$ is \textit{neuron-preserving} if there exists some $S=(s_1,...,s_m)$, $s_i\in [n]\cup\{0,u\}$ such that $q(c) = d$ if and only if $d_i = \left\{\begin{array}{ll} c_j &\text{if } s_i = j \\ 0 &\text{if } s_i = 0\\ 1 & \text{if }s_i = u \end{array}\right. .$ 
  
   If $q$ is neuron-preserving with vector $S$, we write $q=q_S$.  ($S$ may not be unique.)
 \end{definition}

  In particular, given $S=(s_1,...,s_m), s_i\in [n]\cup\{0,u\}$ we can always define a neuron-preserving code map $q_S:\{0,1\}^n \rightarrow \{0,1\}^m$ by $q_S(c)=d$ where $$d_i = \left\{\begin{array}{ll} c_j & \text{if }s_i = j \\ 0 & \text{if }s_i = 0 \\ 1 & \text{if }s_i  = u \end{array}\right. .$$

\begin{lemma}\label{lemma:ambmaps} Let $\C = \{0,1\}^n$ and $\D = \{0,1\}^m$, and suppose $q_S:\C\rightarrow \D$ is neuron-preserving.  Then $\phi_{q_S} = \tau_S$, and thus $q_S= q_{\tau_S}$ by Lemma \ref{lemma:compmaps}.
\end{lemma}

\begin{proof}  Write $q_S= q$ and $\tau_S = \tau$.   Note $\phi_q(y_i) = q^*y_i$, so for any $c\in \C$, $$q^*y_i(c) = y_i(q(c)) = d_i = \left\{\begin{array}{ll} c_j &\text{if } s_i = j  \\ 0 & \text{if }s_i = 0 \\ 1 & \text{if }s_i  = u\end{array}\right. .$$

On the other hand, we also have

$$\tau(y_i)(c) = \left\{\begin{array}{ll} x_j(c) &\text{if } \tau(y_i) = x_j \\ 0 & \text{if }\tau(y_i) = 0 \\ 1 & \text{if }\tau(y_i) = 1 \end{array}\right. = \left\{\begin{array}{ll} x_j(c) & \text{if }s_i = j\\ 0 &\text{if } s_i = 0 \\ 1 & \text{if }s_i = u   \end{array}\right. = \left\{\begin{array}{ll} c_j &\text{if } s_i = j \\ 0 & \text{if }s_i = 0 \\ 1 &\text{if } s_i = 1  \end{array}\right..$$

Thus $\tau$ and $\phi_q$ are identical on $\{y_i\}$, and hence are identical everywhere.
\end{proof}

If $\C$ is a code on $n$ neurons and $\D$ a code on $m$ neurons so $q:\C\rightarrow \D$ is neuron-preserving with vector $S$, and if $q_S:\{0,1\}^n\rightarrow \{0,1\}^m$ is the neuron-preserving map defined above, observe that $q = q_S\big|_\C$, and therefore that $q=q_{\tau_S}\big|_\C$ where $\tau_S$ is the neuron-preserving ring homomorphism defined above.  As $S$ is not necessarily unique, there are often many such possible $\tau_S$.

\begin{lemma} The composition of two neuron-preserving code maps is neuron-preserving.

\end{lemma}

\begin{proof}

Let $q_T:\C\rightarrow \D$ and $q_S:\D \rightarrow \E$ be neuron-preserving. Suppose $q_T(c) = d$ and $q_S(d) = e$. Let $W$ be defined so $w_i = t_{s_i}$ (where $t_0 = 0$ and $t_u = u$) Then,

$$e_i = \left\{\begin{array}{ll} d_j & \text{if }s_i = j\\ 0 &\text{if } s_i = 0\\ 1 & \text{if }s_i = u \\ \end{array}\right.
 = \left\{\begin{array}{ll}  c_k &\text{or if } s_i = j \text{ and } t_j = k  \\ 0 & \text{if }s_i = j \text{ and } t_j = 0\\ &\text{or if } s_i = 0 \\ 1 &\text{if } s_i = j \text{ and } t_j = u\\ & s_i = u\end{array}\right. $$
  $$ = \left\{\begin{array}{ll} c_k &\text{if } t_{s_i} = k \\  0 &\text{if } t_{s_i} = 0 \\ 1 &\text{if } t_{s_i} = u\\\end{array}\right. 
  = \left\{\begin{array}{ll} c_k &\text{if } w_i = k \\ 0 & \text{if }w_i = 0 \\ 1 &\text{if } w_i = u\end{array}\right.$$

  Thus, $q_S(q_T(c)) = q_W(c)$, and so $q_S\circ q_T$ is neuron-preserving.

 \end{proof}

\noindent\textbf{Elementary neuron-preserving code maps}

Of our four original elementary code maps, three of them are neuron-preserving without any restrictions. The only problem is adding a neuron, which may or may not be neuron-preserving, depending on the definition of the function $f(x)$ which defines the new neuron. Here we list the elementary code maps which are neuron preserving, and give their respective $S$-vectors. Throughout, let $\C$ be a code on $n$ neurons and $c=(c_1,...,c_n)$ an element of $\C$.

\begin{enumerate}

\item \textbf{Dropping the last neuron}: Let $S=(1,2,...,n-1)$. Then $q_S(c) = d$, where  $d=(c_1,...,c_{n-1})$.  We require $q_S(\C) = \D$.
\item \textbf{Adding a 1 (respectively,  0)} to the end of each codeword: $S= (1,2,...,n,1)$ (respectively $S=(1,2,...,n,0)$ ). Then $q_S(c) = d$ where $d=(c_1,...,c_n,1)$ [respectively $d=(c_1,...,c_n,0)$].
We require $q_S(\C) = \D$.
\item \textbf{Adding a neuron which repeats neuron $i$} to the end of each word: $S = (1,2,...,n,i)$. Then $q_S(c) = d$, where $d = (c_1,...,c_n, c_i)$.  We require $q_S(\C) = \D$.
\item \textbf{Permuting the labels}: Let $\sigma\in \mathcal S_n$ be a permutation. To relabel the code so neuron $i$ is relabeled $\sigma(i)$, we use $S = (\sigma(1),...,\sigma(n))$.  Then $q_S(c) = d$, where $d=(c_{\sigma(1)},...,c_\sigma(n))$. We require $q_S(\C) = \D$.
\item \textbf{Adding a codeword}: Let $S = (1,2,...,n)$.  This defines an inclusion map, so $q(c) = c$. We use this anytime we have $q(\C)\subsetneq \D$.  Then all codewords in $\D\setminus q(\C)$ are ``added."

\end{enumerate}

\begin{proposition}\label{prop:comps} All neuron-preserving code maps $q$ are compositions of these elementary neuron-preserving code maps:
\begin{enumerate}
\item Permutation of labels
\item Dropping the last neuron
\item Adding a 1 to the end of each codeword
\item Adding a 0 to the end of each codeword
\item Adding a new neuron which repeats another neuron
\item Adding a codeword
\end{enumerate}
\end{proposition}

\begin{proof}
Suppose $\C$ is a code on $n$ neurons and $\D$ is a code on $m$ neurons.  Let $q:\C\rightarrow \D$ be a neuron-preserving code map, with $q=q_S$ for  $S=(s_1,...,s_m)$ , with $s_i \in [n]\cup \{0,u\}$.  

We use the same process as we used in the proof of Proposition 5.  Since permutation, dropping neurons, and adding codewords are all neuron-preserving code maps, It is enough to show that the functions $f$ which we use to define the new neurons we add do indeed correspond to the three operations above: adding a 1, a 0, or a repeating neuron.  That is, we must show $f_i \in \{x_1,...,x_n, 0,1\}$ for $i=1,...,m$.  This is easy; simply define $$f_i(x) = \left\{\begin{array}{ll} x_j & \text{ if } s_i = j\\ 0 & \text{ if } s_i = 0\\ 1 &  \text{ if } s_i =1\end{array}\right.$$
Then, we have  that $$f_i(c) = \left\{\begin{array}{ll} c_j & \text{ if } s_i = j\\ 0 & \text{ if } s_i = 0\\ 1 &  \text{ if } s_i =u\end{array}\right.$$ just as we wish.

\end{proof}

\subsection{Neural ring homomorphisms}

With the idea of neuron-preserving in mind, we define a better notion for neural ring homomorphism.

\begin{definition} Suppose $\C$ is a code on $n$ neurons and $\D$ a code on $m$ neurons.  A ring homomorphism $\phi:R_\D\rightarrow R_\C$ is called a \textit{neural ring homomorphism} if there exists a neuron-preserving compatible ring homomorphism $\tau:R[m]\rightarrow R[n]$.

A ring isomorphism $\phi:R_\D\rightarrow R_\C$ is called a {\em neural ring isomorphism} if there exists a neuron-preserving compatible ring isomorphism $\tau:R[m]\rightarrow R[n]$.
\end{definition}

As we have now proven that ring homomorphisms are in correspondence with code maps, we immediately have the following natural question:

\begin{question} Which code maps correspond to neural ring homomorphisms? to neural ring isomorphisms?
\end{question}

\begin{lemma} \label{lemma:nrharenp} $\phi$ is a neural ring homomorphism if and only if $q_\phi$ is a neuron-preserving code map.

\end{lemma}

\begin{proof} Suppose $\C$ is a code on $n$ neurons and $\D$ a code on $m$ neurons.  

$\phi:R_\D\rightarrow R_\C$ is a neural ring homomorphism if and only if there exists some neuron-preserving  $\tau:R[m]\rightarrow R[n]$ compatible with $\phi$.  Let $\tau=\tau_S$, and let $q_S:\{0,1\}^n\rightarrow \{0,1\}^m$ be the neuron-preserving map defined at the beginning of the previous section.   $q_{\tau_S} = q_S$ by Lemma \ref{lemma:ambmaps} and $q_\phi = q_{\tau_S}\big|_\C$ by Lemma \ref{lemma:compmaps} so $q_\phi  = q_S\big|_\C =q$ is a neuron-preserving code map.
\end{proof}

\begin{theorem}\label{thm:nrharecomps} $\phi$ is a neural ring homomorphism if and only if $q_\phi$ is a composition of the following elementary code maps:
\begin{enumerate}
\item Dropping a neuron
\item Permutation of labels
\item Repeating a neuron
\item Adding a trivial (1 or 0) neuron
\item Adding codewords
\end{enumerate}
\end{theorem}

\begin{proof}

By Lemma \ref{lemma:nrharenp}, $\phi$ is a neural ring homomorphism if and only if $q_\phi$ is neuron-preserving, and by Proposition \ref{prop:comps}, $q_\phi$ is neuron preserving implies it is a composition of the elementary code maps.
\end{proof}

\begin{lemma} Neural ring isomorphisms correspond exactly to those code maps which permute the labels on neurons.
\end{lemma}

\begin{proof}

Suppose $\tau = \tau_S:R[n]\rightarrow R[m]$ is an isomorphism. Then since $R[n]$ and $R[m]$ are finite, they must have the same size. So we must have $2^{2^n} = |R[n]| = |R[m]| = 2^{2^m}$, and thus $n=m$. Thus we can write $\tau:R[n]\rightarrow R[n]$. 

 As $\tau$ is an isomorphism, we know $\ker \tau = \{0\}$. So we cannot have $\tau(x_i) = 0$ since then $x_i \in \ker \tau_S$, and similarly we cannot have $\tau(x_i) = 1$, since then $1-x_i \in \ker\tau$. As $\tau$ is neuron-preserving, this means that  $\tau(x_i) \in \{x_1,...,x_n\}$ for all $i=1,...,n$; as $\tau=\tau_S$ for some $S$, this means that $s_i\in [n]$ for all $i$.
 
   If we had $\tau(x_i) = \tau(x_j) = x_k$ for $i\neq j$, then $x_i-x_j \in \ker\tau$, which is a contradiction as $x_i-x_j\neq 0$.  So $\tau$ induces a bijection on the set of variables $\{x_1,...,x_n\}$; i.e., $S$ contains each index in $[n]$ exactly once.

Now,consider the corresponding code map $q_\tau$. Let $c\in \{0,1\}^n$, and $q_\tau(c) = d$. We must have $\tau(f)(c)  = f(d)$.  In particular, we must have $x_j(c) = x_i(d)$, or rather, $c_j = d_i$.  So $q_\tau$ takes each codeword $c$ to its permutation where $j\rightarrow i$ iff $\tau(x_i) =x_j$. 

Now, we know that if $\tau$ is compatible with $\phi$, then $\phi$ is merely a restriction of the map $\tau$, and so $q_\phi(c) = q_\tau(c)$ for all $c\in \C$.  And as $\phi$ is an isomorphism, $q_\phi$ is a bijection, so every codeword in $\D$ is the image of some $c\in \C$; thus, $q_\phi$ is a permutation map on $\C$, and no codewords are added.
\end{proof}

\section{The effect of elementary code maps on the canonical form}

Here we look at the effect of the elementary moves on the canonical form of the ring.  Throughout, $\C$ is a code on $n$ neurons with canonical form $CF(J_\C)$, and $z_i \in \{x_i, 1-x_i\}$ represents a linear term.

\begin{enumerate}
\item Permutation: as this map simply permutes the labels on the variables, the canonical form stays nearly the same, but the labels are permuted using the reverse permutation $\sigma^{-1}$. That is, let $\D$ be the code obtained by applying the permutation $\sigma\in \mathcal S_n$ to $\C$.  Then  $f=z_{i_1}\cdots z_{i_k} \in CF(J_\C)$ if and only if $f_\sigma = z_{\sigma^{-1}(i_1)}\cdots z_{\sigma^{-1}(i_k)} \in CF(J_\D)$.

\item Dropping a neuron: Let $\C$ be a code on $n$ neurons, and $\D$ the code on $n-1$ neurons obtained by dropping the $n$th neuron. Then $CF(J_\D) = CF(J_\C)\setminus\{ f \mid f=g z_n, g $ a pseudo-monomial $\}$.  That is, we simply remove all pseudo-monomials which involved the variable $x_n$.

\item Adding a new neuron which is always 1 (always 0): Let $\D$ be the code on $n+1$ neurons obtained by adding a 1 (respectively 0) to the end of each codeword in $\C$.  Then $CF(J_\D) = CF(J_\C)\cup \{1-x_{n+1}\}$ (respectively $CF(J_\C) \cup \{x_n\}$).

\item Adding a new neuron which repeats another neuron: Let $\D$ be the code on $n+1$ neurons obtained from $\C$ by adding a new neuron which repeats the action of neuron $i$ for all codewords. 

Let $F = \{f\in CF(J_\C)\,|\, f=z_i\cdot g$ for $ g\text{ a pseudo monomial} \}$; let $H$ be the set formed by replacing $x_i$ with $x_{n+1}$ for all $f\in F$.   Then in most cases, $CF(J_\D) = CF(J_\C) \cup\{x_i(1-x_{n+1}), x_{n+1}(1-x_i)\} \cup H$. The only exception is if $z_i\in CF(J_\C)$; then $CF(J_\D)$ is simply $CF(J_\C)\cup H$.

\item Adding codewords: This is by far the most complicated.  The algorithmic process for obtaining $CF(\C\cup\{v\})$ from $CF(\C)$ is described in  Algorithm 2. 

\end{enumerate}

\backmatter

\appendix


\chapter{Neural codes on three neurons}\label{sec:appendix2}

\vspace{-.2in}

\begin{table*}[!h]
\begin{small}
\begin{tabular}{l | l | l}
Label & Code  $\C$ & Canonical Form $CF(J_\C)$ \\

 \hline
A1 & 000,100,010,001,110,101,011,111 & $\emptyset$\\
A2 & 000,100,010,110,101,111 & $x_3(1-x_1)$\\
A3 & 000,100,010,001,110,101,111 & $x_2x_3(1-x_1)$\\
A4 & 000,100,010,110,101,011,111 & $x_3(1-x_1)(1-x_2)$\\
A5 & 000,100,010,110,111 & $x_3(1-x_1), \,x_3(1-x_2)$\\
A6 & 000,100,110,101,111 & $x_2(1-x_1), \,x_3(1-x_1)$\\
A7 & 000,100,010,101,111 & $x_3(1-x_1),\, x_1x_2(1-x_3)$\\
A8 & 000,100,010,001,110,111 & $x_1x_3(1-x_2), \, x_2x_3(1-x_1)$\\
A9 & 000,100,001,110,011,111 & $x_3(1-x_2), \, x_2(1-x_1)(1-x_3)$\\
A10 & 000,100,010,101,011,111& $x_3(1-x_1)(1-x_2), \,x_1x_2(1-x_3)$\\
A11 & 000,100,110,101,011,111 & $x_2(1-x_1)(1-x_3), \,x_3(1-x_1)(1-x_2)$\\
A12 & 000,100,110,111 & $x_3(1-x_1),\, x_3(1-x_2), \,x_2(1-x_1)$\\
A13 & 000,100,010,111 & $x_3(1-x_1), \,x_3(1-x_2), \,x_1x_2(1-x_3)$\\
A14 & 000,100,010,001,111 & $x_1x_2(1-x_3), \, x_2x_3(1-x_1), \, x_1x_3(1-x_2)$\\
A15 & 000,110,101,011,111  & $x_1(1-x_2)(1-x_3), \, x_2(1-x_1)(1-x_3),$\\
& & \hspace{.2in} $  x_3(1-x_1)(1-x_2)$ \\
A16* & 000,100,011,111 & $x_2(1-x_3),\, x_3(1-x_2)$\\
A17* & 000,110,101,111 & $x_2(1-x_1), \,x_3(1-x_1),\, x_1(1-x_2)(1-x_3)$\\
A18* & 000,100,111 & $x_2(1-x_1), \, x_2(1-x_3), \, x_3(1-x_1), \, x_3(1-x_2)$\\
A19* & 000,110,111 & $x_3(1-x_1), \, x_3(1-x_2), \, x_1(1-x_2), \, x_2(1-x_1)$\\
A20* & 000,111 & $x_1(1-x_2), \, x_2(1-x_3), \,x_3(1-x_1), \, x_1(1-x_3),$\\
&  &\hspace{.2in} $x_2(1-x_1), \,x_3(1-x_2)$\\

\hline
\end{tabular}
\caption{\small Continues next page}
\end{small}

\end{table*}

\begin{table*}[!h]
\begin{small}
\begin{tabular}{l | l | l}

Label & Code  $\C$ & Canonical Form $CF(J_\C)$ \\
\hline
B1 & 000,100,010,001,110,101 & $x_2 x_3$ \\
B2 & 000,100,010,110,101& $x_2x_3, \,x_3(1-x_1)$\\
B3 & 000,100,010,101,011 & $x_1x_2, \, x_3(1-x_1)(1-x_2)$\\
B4 & 000,100,110,101 & $x_2x_3, \, x_2(1-x_1), \, x_3(1-x_1)$\\
B5 & 000,100,110,011 & $x_1x_3, \, x_3(1-x_2), \, x_2(1-x_1)(1-x_3)$ \\
B6* & 000,110,101 & $x_2x_3, \, x_2(1-x_1), \, x_3(1-x_1), \,x_1(1-x_2)(1-x_3)$\\

\hline
C1 & 000,100, 010,001, 110 & $x_1x_3, \, x_2x_3$\\
C2 & 000,100,010,101 & $x_1x_2, \, x_2x_3, \, x_3(1-x_1)$\\
C3* & 000,100,011 & $x_1x_2, \, x_1x_3, \, x_2(1-x_3), \, x_3(1-x_2)$\\
\hline
D1 & 000,100,010,001 & $x_1x_2, \, x_2x_3, \, x_1x_3$\\
\hline
E1& 000,100,010,001,110,101,011 & $x_1x_2x_3 $\\
E2 & 000,100,010,110,101,011 & $x_1x_2x_3, \, x_3(1-x_1)(1-x_2)$ \\
E3 & 000,100,110,101,011 & $x_1x_2x_3, \, x_2(1-x_1)(1-x_3), \, x_3(1-x_1)(1-x_2)$\\
E4 & 000,110,011,101 & $x_1x_2x_3, \, x_1(1-x_2)(1-x_3), \, x_2(1-x_1)(1-x_3),$\\
& & \hspace{.2in} $ x_3(1-x_1)(1-x_2)$\\
\hline
F1* & 000,100,010,110 & $x_3$ \\
F2* & 000,100,110 & $x_3, \, x_2(1-x_1)$\\
F3* & 000,110 & $x_3, \,x_1(1-x_2), \, x_2(1-x_1)$\\
\hline 
G1* & 000,100 & $x_2, \, x_3$\\
\hline
H1* & 000 & $x_1, \, x_2, \,  x_3$\\
\hline
I1* & 000,100,010 & $x_3, \, x_1x_2$\\
\end{tabular}
\end{small}
\caption*{Table A.1: \small
Forty permutation-inequivalent codes, each containing $000$, on three neurons.  Labels A--I indicate the various families of Type 1 relations present in $CF(J_\C)$, organized as follows (up to permuation of indices):
(A) None,
(B) $\{x_1x_2\}$,
(C) $\{x_1x_2, x_2x_3\}$,
(D) $\{x_1x_2, x_2x_3, x_1x_3\}$,
(E) $\{x_1x_2x_3\}$,
(F) $\{x_1\}$,
(G) $\{x_1, x_2\}$,
(H) $\{x_1, x_2, x_3\}$,
(I) $\{x_1, x_2x_3\}$.
All codes within the same A--I series share the same simplicial complex, $\Delta(\C)$.
The $*$s denote codes that have $U_i = \emptyset$ for at least one receptive field (as in the F, G, H and I series) as well as codes that require $U_1 = U_2$ or $U_1 = U_2 \cup U_3$ (up to permutation of indices); these are considered to be highly degenerate.  The remaining 27 codes are depicted with receptive field diagrams (Figure 6) and Boolean lattice diagrams (Figure 7).
}
\end{table*}

\newpage

\begin{figure}[!h]
\begin{center}
\includegraphics[width=6in]{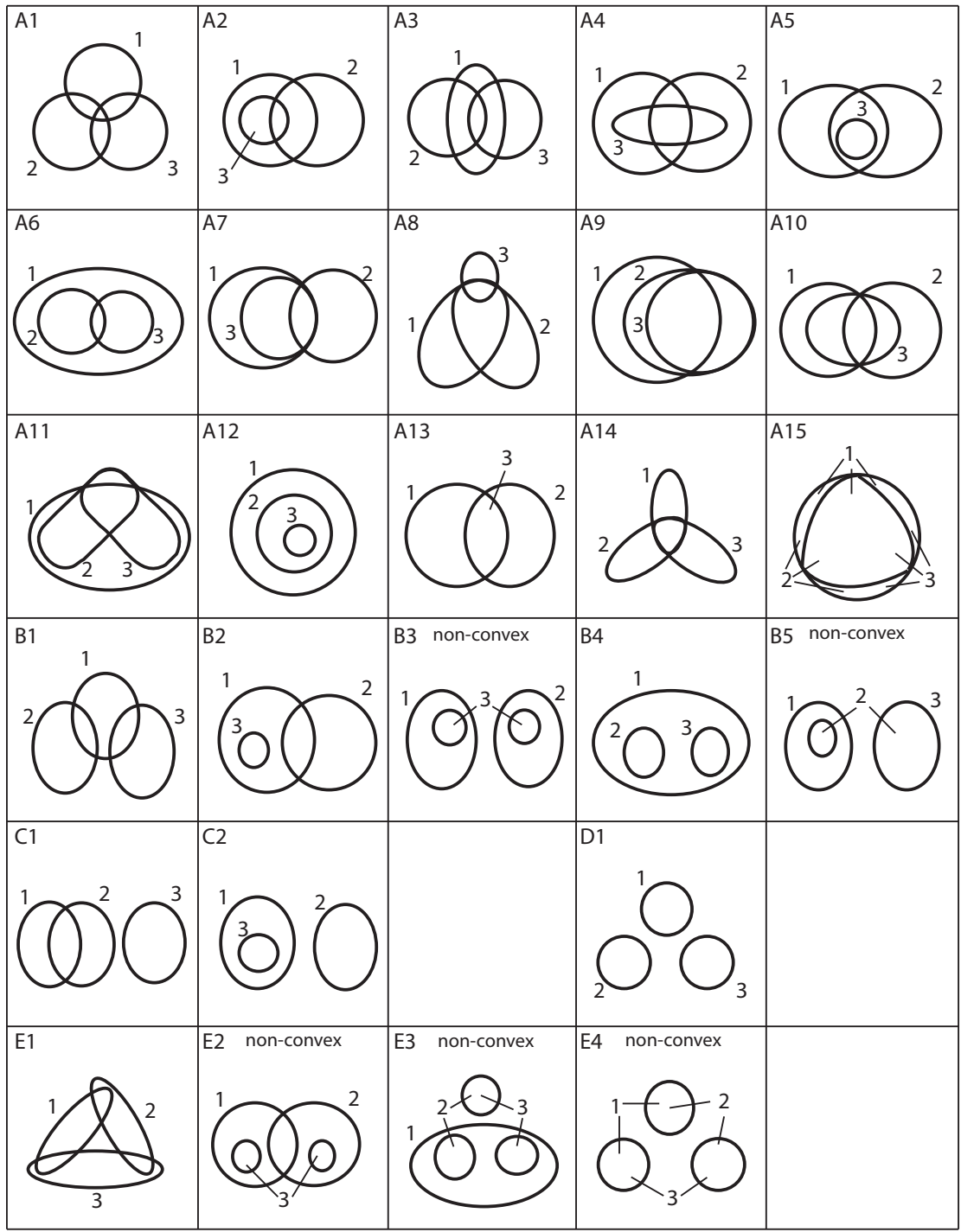}
\end{center}
\vspace{-.2in}
\caption{\small Receptive field diagrams for the 27 non-$*$ codes on three neurons listed in Table 1.  
Codes that admit no realization as a convex RF code are labeled ``non-convex.''  The code E2 is the one from Lemma~\ref{lemma:convex-counterexample}, while
A1 and A12 are permutation-equivalent to the codes in Figure 3A and 3C, respectively.  Deleting the all-zeros codeword from A6 and A4 yields codes permutation-equivalent 
to those in Figure 3B and 3D, respectively.}
\end{figure}

\newpage

\begin{figure}[!h]
\begin{center}
\includegraphics[width=6in]{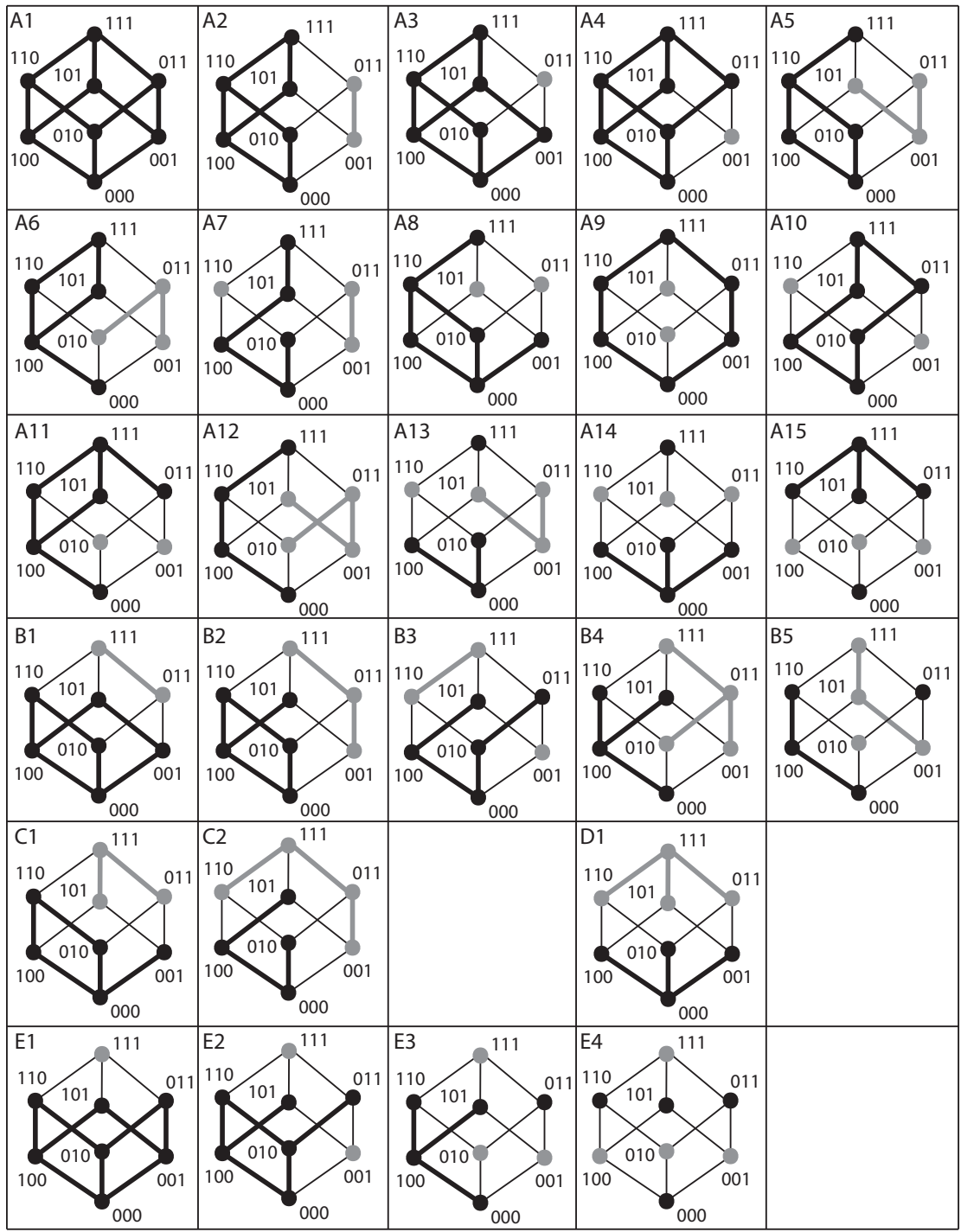}
\end{center}
\vspace{-.2in}
\caption{\small Boolean lattice diagrams for the 27 non-$*$ codes on three neurons listed in Table 1.  Interval decompositions (see Section~\ref{sec:boolean-lattice}) 
for each code are depicted in black,
while decompositions of code complements, arising from $CF(J_\C)$, are shown in gray.  Thin black lines connect elements of the Boolean lattice
that are Hamming distance 1 apart.  Note that the lattice in A12 is permutation-equivalent to the one depicted in Figure 5.}
\end{figure}

\chapter{ MATLAB Code}
The following code is used to obtain the canonical form of the ideal $J_\C$, given a code $\C$.  
We follow the inductive algorithm given in Chapter 5.  Although an algebraic language such as Macaulay2 would be more natural for an algorithm involving rings and ideals, we use Matlab.  Primarily, this is because the data the algorithm is built to analyze is most likely to be in Matlab data, and this removes the necessity of a translation step between the two programs.  Furthermore, because of the simple (indeed, binary) nature of most of our data, and the properties of pseudo-monomials, it is not difficult to translate the necessary information and operations into matrix format in this instance.

A note on this translation: pseudo-monomials in $\F_2[x_1,...,x_n]$ are stored as the rows of a matrix of 1s, 0s, and 5s via the following rule: the row $(v_1,...,v_n)$ corresponds to the pseudo monomial $\prod_{v_i=0} x_i\prod_{v_i = 1} (1-x_i)$.  The selection of 5 as the `empty' slot is not random; we wanted information about the similarity between two pseudo-monomials to be extracted by subtracting the rows. Choosing a 2, though it seems more natural, leads to confusion as $2-1 = 1-0$; therefore we select 5 instead since it is large enough to not have this problem.

The function FinalIntersectIdeals finds the canonical form of $CF(J_{\C_i})$ from $CF(J_{\C_{i-1}})$ and the next codeword $c^i$; the overarching function Code2CF takes the code, and calls FinalIntersectIdeals with each new $CF(J_{\C_i})$ until all codewords are used; it then outputs $CF(J_{\C_n}) = CF(J_\C)$.\\

{\small
\begin{lstlisting}
function [CF,t] = Code2CF(C)
%---------------------------------------------------------
%Input:	 C, a binary matrix representing the code;
%		 each row is a codeword.  
%Output: a matrix CF of 0s, 1s, and 5s, representing
%    		   the canonical form CF(J_C).  
%		 a number t: the number of operations it performed, 
%		   used for time complexity analysis
%-----------------------------------------------------------
L=diag(C(1,:))+ones(size(C,2),size(C,2))*5-eye(size(C,2))*5;
[I,t]=FinalIntersectIdeals(L,C(2,:));

for i=3:size(C,1)
    r=C(i,:);
    [I,t0]=FinalIntersectIdeals(I,r);
    t=t+t0;
end
CF=I;
\end{lstlisting}

\begin{lstlisting}

function [Ideal,t] = FinalIntersectIdeals(L,r)
%-------------------------------------------------------------------------
%Input:  L, a matrix with rows representing
%		 the pseudo-monomials of CF(J_{C_{i-1}})
%		 r, a vector representing the next codeword c^i
%Output: Ideal, a matrix Ideal of 1s, 0s, and 5s, which
% 		 gives a set of pseudo-monomials 
%		for CF(C\cup r).
%	 	 t, the number of operations that occurred.
%-------------------------------------------------------------------------


n = size(L,2);
m = size(L,1);
I = ones(m*n,n)*5;
k=1;
w=1;
L0=ones(0,n);

%if any element in any row of L matches a monomial in r, 
%put that row in I. otherwise, multiply it by
% each monomial in r in the next step
for i=1:m
    if (size(find((L(i,:)+r)==2),2)>0 || size(find((L(i,:)+r)==0),2)>0)
        I(k,:)=L(i,:);
        k=k+1;  
    else
        L0(w,:)=L(i,:);
        w=w+1;
    end    
end

k1=k-1; 
w=w-1; 
t=0;

%multiply each p-m in L0 by linear terms in p_r
%check for mutliples
for i=1:w
    for j=1:n
        if L0(i,j)==5   %otherwise, we get xi(1-xi)
            I(k,:)=L0(i,:); %put the i'th row of L0 in the new ideal
            I(k,j)=r(j);    %multiply it by the j'th monomial of r
            if k>1
                M=0; %is it a multiple of something?
                for l=1:k1;
                    t=t+1; 
                    diff = I(l,:)-I(k,:);
                    if size(find(abs(diff)==1),2)==0 
                        if size(find(diff<-1),2)==0 
                            M=1; %yes, it's a multiple
                            break; 
                            %definitely not going to put it in, stop
                        end
                    end
                end
                if M==0 %no multiples
                    k=k+1;
                end
            else % if k is 1 or 2, 
                k=k+1;
            end
        end
    end
end
Ideal = I(1:k-1,:);


\end{lstlisting} }

\bibliographystyle{unsrt}
\bibliography{thesis-references}


\end{document}